%% file: arxiv.tex
\documentclass[12pt]{article}

\usepackage{amsmath}
\usepackage{amssymb}
\usepackage{xspace}
\usepackage{subcaption}
\usepackage{bm}
\usepackage{tikz}

\usepackage{hyperref}

\hypersetup{
	hidelinks,
	colorlinks=true,
	allcolors=black,
	pdfstartview=Fit,
	breaklinks=true
}

\newtheorem{xdefinition}{Definition}
\newtheorem{xobservation}{Observation}
\newtheorem{xtheorem}{Theorem}
\newtheorem{xlemma}{Lemma}
\newtheorem{xclaim}{Claim}
\newtheorem{xproposition}{Proposition}
\newtheorem{xcorollary}{Corollary}
\newtheorem{xexample}{Example}
	{\hspace*{\fill}\raisebox{-1pt}{\boldmath$\Box$}\end{xdefinition}}
	{\hspace*{\fill}\raisebox{-1pt}{\boldmath$\Box$}\end{xobservation}}
\newenvironment{theorem}{\begin{xtheorem}\rm}{\end{xtheorem}}
\newenvironment{lemma}{\begin{xlemma}\rm}{\end{xlemma}}

\newenvironment{proposition}{\begin{xproposition}\rm}{\end{xproposition}}
\newenvironment{corollary}{\begin{xcorollary}\rm}{\end{xcorollary}}
\newenvironment{proof}{\begin{trivlist}\item[]{\bf Proof }}%
	{\hspace*{\fill}\raisebox{-1pt}{\boldmath$\Box$}\end{trivlist}}
\newenvironment{example}{\begin{xexample}\rm}%
	{\hspace*{\fill}\raisebox{-1pt}{\boldmath$\Box$}\end{xexample}}

\newcommand{\OPT}{\ensuremath{\operatorname{\textsc{Opt}}}\xspace}
\newcommand{\ALG}{\ensuremath{\operatorname{\textsc{Alg}}}\xspace}
\newcommand{\CRS}{\ensuremath{\operatorname{\textsc{CRS}}}\xspace}
\newcommand{\RT}{\ensuremath{\operatorname{\textsc{RT}}}\xspace}

\newcommand{\TG}{\ensuremath{\operatorname{\textsc{TG}}}\xspace}

\newcommand{\TRUST}{\ensuremath{\operatorname{\textsc{Trust}}}\xspace}
\newcommand{\GREEDY}{\ensuremath{\operatorname{\textsc{Greedy}}}\xspace}
\newcommand{\TRUSTGREEDY}{\ensuremath{\operatorname{\textsc{TrustGreedy}}}\xspace}

\newcommand{\ROBUSTTRUST}{\textsc{RobustTrust}\xspace}
\newcommand{\FN}{\ensuremath{\operatorname{\textsc{FN}}}\xspace}
\newcommand{\FP}{\ensuremath{\operatorname{\textsc{FP}}}\xspace}
\newcommand{\TP}{\ensuremath{\operatorname{\textsc{TP}}}\xspace}
\newcommand{\TN}{\ensuremath{\operatorname{\textsc{TN}}}\xspace}

\newcommand{\F}{\ensuremath{\operatorname{\mathcal{F}}}\xspace}
\newcommand{\FEAS}{\ensuremath{F}\xspace}
\newcommand{\IOPT}{\ensuremath{\operatorname{\mathit{I^\ast}}}\xspace}
\newcommand{\IPRED}{\ensuremath{\operatorname{\mathit{\hat{I}}}}\xspace}
\newcommand{\EPRED}{\ensuremath{\operatorname{\mathit{\hat{E}}}}\xspace}

\newcommand{\CEIL}[1]{\left\lceil#1\right\rceil}
\newcommand{\FLOOR}[1]{\left\lfloor#1\right\rfloor}

\newcommand{\SET}[1]{\left\{#1\right\}}
\newcommand{\SETOF}[2]{\SET{#1 \mid #2}}

\newcommand{\OPTI}{\ensuremath{\OPT}\xspace}
\newcommand{\OPTFN}{\ensuremath{\OPT^{\FN}}\xspace}
\newcommand{\OFN}{\ensuremath{O^{\FN}}\xspace}
\newcommand{\TGFN}{\ensuremath{\TG^{\FN}}\xspace}
\newcommand{\TGTP}{\ensuremath{\TG^{\TP}}\xspace}
\newcommand{\TGI}{\ensuremath{\TG}\xspace}

\newcommand{\OPTHAT}{\ensuremath{\OPT^{\TP}}\xspace}
\newcommand{\error}{\ensuremath{\eta}\xspace}
\newcommand{\ernorm}{\ensuremath{\gamma(\IPRED,I)}\xspace}
\newcommand{\ernormw}{\ensuremath{\gamma(\IPRED_w,I_w)}\xspace}
\newcommand{\ernormE}{\ensuremath{\gamma(\EPRED,E)}\xspace}

\newcommand{\eps}{\ensuremath{\varepsilon}\xspace}

\newlength{\mydaggerlen}
\input{figs.tex}

\begin{document}

\title{Online Interval Scheduling with Predictions\,\thanks{Boyar, Favrholdt, and Larsen were supported in part by the Danish Council for Independent Research grants DFF-0135-00018B and DFF-4283-00079B, and Kamali was supported in part by the Natural Sciences and Engineering Research Council of Canada (NSERC). A preliminary extended abstract of some of these results was published by the same authors in the 18th International Algorithms and Data Structures Symposium (WADS), volume 14079 of Lecture Notes in Computer Science, pages 193-207. Springer, 2023.}}

\author{\begin{tabular}{c@{\hspace*{2em}}c}Joan Boyar\textsuperscript{$\dagger$} & Lene M. Favrholdt\textsuperscript{$\dagger$} \\[1ex] Shahin Kamali\textsuperscript{$\ddagger$} & Kim S. Larsen\textsuperscript{$\dagger$}\end{tabular} \\[5ex]
\small
\begin{tabular}[t]{l}
  {}\textsuperscript{$\dagger$}University of Southern Denmark \\ \hspace{\mydaggerlen} \texttt{https://imada.sdu.dk/u/\{joan,lenem,kslarsen\}} \\[2ex]
  {}\textsuperscript{$\ddagger$}York University \\ \hspace{\mydaggerlen} \texttt{https://www.eecs.yorku.ca/\raisebox{-.7ex}{\textasciitilde{}}kamalis/}\end{tabular} \\ \mbox{}}

\date{January 5, 2025}

\maketitle


\begin{abstract}
In online interval scheduling, the input is an online sequence of intervals, and the goal is to accept a maximum number of non-overlapping intervals. In the more general disjoint path allocation problem, the input is a sequence of requests, each consisting of pairs of vertices of a known graph, and the goal is to accept a maximum number of requests forming edge-disjoint paths between accepted pairs. We study a setting with a potentially erroneous prediction specifying the set of requests forming the input sequence and provide tight upper and lower bounds on the competitive ratios of online algorithms as a function of the prediction error. We also present asymptotically tight trade-offs between consistency (competitive ratio with error-free predictions) and robustness (competitive ratio with adversarial predictions) of interval scheduling algorithms. Finally, we provide experimental results on real-world scheduling workloads that confirm our theoretical analysis.
\end{abstract}

\section{Introduction}
In the \emph{interval scheduling problem}, the input is a set of intervals with integral endpoints, each representing timesteps at which a process starts and ends. A scheduler's task is to decide whether to accept or reject each job so that the intervals of accepted jobs do not overlap except possibly at their endpoints. The objective is to maximize the number of accepted intervals, referred to as the \emph{profit} of the scheduler. This problem is also known as \emph{fixed job scheduling} and \emph{k-track assignment}~\cite{kolen2007interval}. 

Interval scheduling is a special case of the \emph{disjoint path allocation problem}, 
where the input is a graph $G$ and a set of $n$ \emph{requests}, each defined by a pair of vertices in $G$. An algorithm can accept or reject each pair, given that it can form a path between the vertices of each of the accepted pairs, such that all of the paths are edge-disjoint.
Interval scheduling is the particular case when $G$ is a path graph.
The disjoint path allocation problem can be solved in polynomial time for
trees~\cite{GVV97} and outerplanar graphs 
by a combination of~\cite{WW95,MNS85,F85},
but the problem is NP-complete for general graphs~\cite{EIS76},
and even on quite restricted graphs such as series-parallel graphs~\cite{NVZ01}.
The disjoint path problem is the same as call
control/call allocation with all calls having unlimited duration and all bandwidths (both of the calls and the edges
they would be routed on) being equal to 1 and as the maximum multi-commodity
integral flow problem with edges having unit capacity.

In this work, we focus on the \emph{online} variant of the problem, in which the set of requests is not
known in advance but is revealed in the form of a request sequence, $I$. Each new request must either be irrevocably accepted or rejected.
On acceptance, the algorithm selects and commits to a particular path, edge-disjoint from
any previously selected path.
We analyze an online algorithm via a comparison with an optimal offline algorithm, \OPT. 
An online algorithm \ALG is \emph{$c$-competitive}, if there exists a constant $b$ such that,
for any input graph $G$ and any sequence $I$ of requests on $G$, $\ALG(I) \geq c \cdot \OPT(I) - b$, where $\ALG(I)$ and $\OPT(I)$, respectively, denote the profit of \ALG and \OPT on $I$ (for randomized algorithms, $\ALG(I)$ is the expected profit of $\ALG$).
If $b=0$, \ALG is \emph{strictly $c$-competitive}.
The \emph{competitive ratio} of \ALG is defined as $\sup \SET{c \mid \text{\ALG is $c$-competitive}}$ and its \emph{strict competitive ratio} is $\sup \SET{c \mid \text{\ALG is strictly $c$-competitive}}$.
Since we consider a maximization problem, our ratios are between zero
and one.
Our results are strongest possible in the sense that all lower bounds
(positive results) are valid with respect to the strict competitive
ratio and all upper bounds (negative results) are valid with respect
to the (non-strict) competitive ratio. 

In interval scheduling, there is a distinction between the \emph{any-order} setting~\cite{BorodinK23}, where the adversary determines the ordering of intervals, and the \emph{sorted-order}~\cite{LiptonT94} setting, where intervals arrive in order of their starting times. The focus of this paper is on the any-order setting. Since
the positive results presented in the paper are established for the any-order setting, they also hold for the sorted-order setting.
Note that, in contrast to~\cite{BorodinK23}, we do not allow preemption.

For interval scheduling on a path graph with $m$ edges, the {\em strict} competitive ratio of the best deterministic algorithm
is $\frac{1}{m}$~\cite{BE98}, and no deterministic algorithm can be better than $\frac{1}{\sqrt{m}}$-competitive as seen by the following adversarial strategy. First $\sqrt{m}$ disjoint intervals of length $\sqrt{m}$ are given. Then, for each interval $S$ accepted by the algorithm, $\sqrt{m}$ disjoint unit intervals overlapping $S$ are given.
Note that, in~\cite{BE98}, intervals model calls of limited duration. In that model, sequences can be repeated an arbitrary number of times, and thus, there is no difference between the strict and the non-strict competitive ratio.
The best randomized algorithm has a competitive ratio of $\Theta(\frac{1}{\log m})$~\cite{BE98}. For this result, the upper bound (negative result) proof of~\cite{BE98} also holds for our model. Moreover, the upper bound follows from our Theorem~\ref{th:constrob} with $\alpha=0$.
These results suggest that the constraints on online algorithms must be relaxed to compete with \OPT. Specifically, the problem has been considered in the \emph{advice complexity model} for path graphs~\cite{BBFGJKSS14,GKKKS15}, trees~\cite{BBK22}, and grid graphs~\cite{BKW22}. Under the advice model, the online algorithm can access error-free information on the input called advice.
The objective is to quantify the trade-offs between the competitive ratio and the size
of the advice.
Another relaxation of the online model is to consider \emph{priority algorithms}~\cite{BNR02},
where the algorithm is allowed to give priorities to the entire set of possible input
items and always receive the input with highest priority next. The algorithm processes
the items it receives in an online manner. Priority algorithms are a model for
greedy algorithms that has also been studied for graph problems, in particular~\cite{BBLM10j,DI09}. Priority algorithms for the disjoint path allocation problem
for path graphs, trees, and grid graphs are studied in~\cite{BFH24}, along with
priority algorithms in a model that includes advice~\cite{BLP24j}.

In recent years, there has been an increasing interest in improving the performance of online algorithms via the notion of \emph{prediction}. Here, it is assumed that the algorithm has access to a prediction, for instance in the form of machine-learned information. Unlike the advice model, the prediction may be erroneous and is quantified by an \emph{error measure $\eta$}. The objective is to design algorithms whose competitive ratio degrades gently as a function of $\eta$. 
Several online optimization problems have been studied under the prediction model, including non-clairvoyant scheduling~\cite{NIPS2018_8174,WeiZ20}, makespan scheduling~\cite{lattanzi2020online}, contract scheduling~\cite{AK21,abs-2111-05281}, and other variants of scheduling problems~\cite{AzarLT21,LeeMHLSL21,BampisDKLP22,BalkanskiGT23}.

Other online problems studied under the prediction model include 
bin packing~\cite{0001DJKR20}, knapsack~\cite{Zeynali0HW21,knapsack22,BoyarFL22}, caching~\cite{LV21,rohatgi2020near,W20,succinctPaging},
matching problems~\cite{AGKK20,LLMV20,LavastidaM0X21}, time series search~\cite{AAAI22Search}, and various graph problems~\cite{ChenSVZ22,EberleLMNS22,ChenEILNRSWWZ22,AzarPT22,BanerjeeC0L23}.
See also the survey by Mitzenmacher and Vassilvitskii~\cite{mitzenmacher2020algorithms} and the collection at~\cite{ALPS}.

\subsection{Contributions}
We study the online disjoint path allocation problem under a setting where the scheduler is provided with a set $\IPRED$ of requests predicted to form the input sequence $I$.
Given the erroneous nature of the prediction, some 
requests in $\IPRED$ may be incorrectly predicted to be in $I$ (false positives), and some requests in $I$ may not be included in $\IPRED$ (false negatives). We let the \emph{error set} be the set of requests that are false positives or false negatives and define the error parameter $\eta(\IPRED,I)$ to be the cardinality of the largest set of requests in the error set that can be accepted. For interval scheduling, this is the largest set of non-overlapping intervals in the error set. Thus, $\eta(\IPRED,I) = \OPT(\FP \cup \FN)$, where \FP and \FN are the sets of false positives and negatives, respectively. 
We explain later that this definition of $\eta$ has specific desired properties for the prediction error (Proposition~\ref{propo:proper}).
In the following, we use $\ALG(\IPRED, I)$ to denote the profit of an algorithm \ALG for prediction $\IPRED$ and input $I$. We also define $\gamma(\IPRED,I) = \eta(\IPRED,I)/\OPT(I)$; this \emph{normalized error} measure is helpful in describing our results because the point of reference in the competitive analysis is $\OPT(I)$. 
Our first result concerns general graphs.

\begin{description}
\item[Disjoint-Path Allocation]\mbox{}

We study a simple algorithm \TRUST, which accepts a request only if it
belongs to a given optimal solution for $\IPRED$.
We show that, \TRUST is strictly $(1-2\gamma)$-competitive (Theorem~\ref{th:trustCompetitive}).
Furthermore, this is best possible among deterministic algorithms,
even on the graph class trees (Theorem~\ref{th:lowerstar}).
\end{description}

The above result demonstrates that even for trees, 
the problem is so hard that no algorithm can do better than the trivial \TRUST.
Therefore, our main results concern the more interesting case of path graphs, that is, interval scheduling:
\begin{description}
\item[Interval Scheduling] \mbox{}
  
We first show that no deterministic interval scheduling algorithm can
be better than $(1-\gamma)$-competitive
(Theorem~\ref{thm:generallower}).

Next, we show that \TRUST is no better on the path than on trees; its competitive ratio is only $1-2\gamma$
(Theorem~\ref{th:intervaltrustlower}). This suggests that there is room for improvement over \TRUST.

Finally, we introduce our main technical result, a deterministic
algorithm \TRUSTGREEDY that achieves an optimal competitive ratio of
$1-\gamma$ for interval scheduling (Theorem~\ref{th:trustgreedymain}). \TRUSTGREEDY is similar to \TRUST in that it maintains an optimal solution for $\IPRED$, but unlike \TRUST, it updates its planned solution to accept requests greedily when it is possible without a decrease in the profit of the maintained solution.

\begin{description}
\item[Consistency-Robustness Trade-off]\mbox{}

  We study the trade-off between \emph{consistency} and \emph{robustness}, which measure an algorithm's competitive ratios in the extreme cases of error-free prediction (consistency) and adversarial prediction (robustness)~\cite{DBLP:conf/icml/LykourisV18,LV21}. We show that no deterministic algorithm with a constant consistency can have robustness $\omega(1/m)$. Thus, we focus on the more interesting case of randomized algorithms. (Proposition~\ref{prop:deterconstrob}).
  Suppose that for any input $I$, an algorithm \ALG
  guarantees a consistency of $\alpha$
    and robustness of $\beta \geq
  \frac{1}{m^{1-\eps}}$, $\eps>0$.
We show that
\begin{align*}
& \alpha \leq 1-\beta \left( \frac{\eps \log m}{2} - O(1) \right)
\text{ and } \\
& \beta \leq (1-\alpha) \frac{2}{\eps \log m} + O\left(\frac{1}{\log^2 m} \right) 
\end{align*}
(Theorem~\ref{th:constrob}).
Note that, if $\beta \in \Omega(\frac{1}{\log m})$, the constraint $\beta \geq \frac{1}{m^{1-\eps}}$ is fulfilled for any $\eps < 1$.
Thus, for example, to guarantee a robustness of $\frac{1}{5 \log m}$, the consistency must be at most $\frac{9}{10}+O(\frac{1}{\log m})$, and to guarantee a consistency of $\frac{9}{10}$, the robustness must be at most $\frac{1}{5 \log m} + O(\frac{1}{\log^2 m})$. 
We also present a family of randomized algorithms that provides an almost \emph{Pareto-optimal} trade-off between consistency and robustness (Theorem~\ref{th:robtrust}).  

\item[Experiments on Real-World Data]\mbox{}

We compare our algorithms with \OPT and the online \GREEDY algorithm (which accepts an interval if and only if it does not overlap previously accepted intervals) on real-world scheduling data from~\cite{ChapinCFJLSST99}. Our results are in line with our theoretical analysis: both \TRUST and \TRUSTGREEDY are close-to-optimal for small error values; \TRUSTGREEDY is almost always better than \GREEDY even for large values of error, while \TRUST is better than \GREEDY only for small error values.   

\end{description}
\item[Matching and Independent Set]\mbox{}

  We explain that our results on disjoint-path allocation carry over to matching in the edge-arrival model (Corollary~\ref{cor:matching}) and to  
  independent set on line graphs in the vertex-arrival model (Corollary~\ref{cor:independentset}).
\end{description}

\section{Model and Predictions}
Throughout the paper we let $m$ denote the number of edges in the
graph and let $n$ denote the number of requests in the input sequence.

We assume that an oracle provides the online algorithm with a set \IPRED of requests predicted to form the input sequence $I$. 
This type of prediction arises naturally in scenarios such as call admission, where the input sequence represents the calls made on a given day. By analyzing historical data from previous days, one can predict, albeit with some error, whether two nodes in the call network will initiate a call on that day.

One may consider alternative predictions, such as statistical information about the input. While these predictions are compact and can be efficiently learned, they cannot help achieve close-to-optimal solutions. In particular, for interval scheduling on a path with $m$ edges, since the problem is AOC-complete, one cannot achieve a competitive ratio  $c \leq 1$ with fewer than $c n/(e \ln 2)$ bits~\cite{BFKM17},
even if all predictions are correct~\cite{GKKKS15}. In particular, to achieve a competitive ratio $c \in \Omega(1/\log m)$---the best attainable by randomized algorithms without prediction---one requires a prediction of size $\Omega(n/\log m)$, which grows linearly with the input length (and thus cannot merely be statistical information).

\subsection{Error Measure}
In what follows, \emph{true positive} (respectively, \emph{negative}) intervals are correctly predicted to appear (respectively, not to appear) in the
request sequence. \emph{False positives} and \emph{negatives} are defined analogously as
those incorrectly predicted to appear or not appear.
We let \TP, \TN, \FP, \FN denote the four sets containing these different types
of intervals.
Thus, $I = \TP \cup \FN$ and $\IPRED = \TP \cup \FP$.
We use $\eta(\IPRED,I)$, to denote the error for the input formed by the sequence $I$, when the prediction is the set \IPRED.
When there is no risk of confusion, we use $\eta$ instead of
$\eta(\IPRED,I)$.
The error measure we use here is \[\eta=\OPT(\FP\cup\FN)\,,\] and hence, the normalized error measure is \[\gamma=\frac{\OPT(\FP\cup\FN)}{\OPT(I)}\,.\]

Our error measure has three desirable properties (see below), the first
two of which were recommended in Im et al.~\cite{IKQP21a}.
The first property ensures that improving the prediction does not lead
to a larger error.
The other two ensure that the error is neither too small nor too large,
which is crucial for distinguishing between good and bad algorithms.

In Section~\ref{sec:otherMeasures}, we discuss other natural error models, such as the Hamming distance between the request sequence and prediction, and explain why these measures do not have our desired properties. 

\begin{description}
\item[Monotonicity]\mbox{}
  
	This property ensures that improving the prediction to increase the number
	of true positives or negatives does not
	increase the error.
	To be more precise, if we increase $|\TP|$ by one unit (decreasing $|\FN|$ by
	one unit) or increase $|\TN|$ by one unit (decreasing $|\FP|$ by one unit),
	the error must not increase.
	Formally, for any $I$, \IPRED, the following must hold.
	\begin{itemize}
		\item For any $x\in I\setminus \IPRED$,
		$\eta(I, \IPRED \cup \{ x\}) \leq \eta(\IPRED,I)$.
		\item For any $y\in \IPRED\setminus I$, $\eta(I, \IPRED \setminus \{ y\}) \leq \eta(\IPRED,I)$.
	\end{itemize}
      \item[Lipschitz property]\mbox{}

        This property requires the error to be at least equal to the net difference between $\OPT(I)$ and $\OPT(\IPRED)$, that is,
	\[\eta(\IPRED,I) \geq |\OPT(I)-\OPT(\IPRED)|\,.\]
\end{description}
	Note that the Lipschitz property ensures that the error is not ``too
	small'', causing all algorithms to have a bad competitive
        ratio, even when the error, according to the error measure, is
        low. To be able to distinguish between good and mediocre
        algorithms, we must also avoid that the error becomes ``too large''.
        Hence, we also define a notion of Lipschitz-completeness,
        imposing an upper bound on the error.
        The particular upper bound that we use is motivated by the
        following example. 
        \begin{example}
          \label{ex:completeness}
	  The input is formed by a set $I = A \cup B$ of requests, with $A=\{A_1, A_2, \ldots, A_k\}$ and $B = \{B_1, B_2, \ldots, B_{k-1}\}$, where the $A_i$'s are disjoint, the $B_i$'s are disjoint, and $B_i$ overlaps $A_i$ and $A_{i+1}$. The profit of the optimal solution is then $\OPT(I) = |A| = k$. Suppose the prediction is $\IPRED = ( A \setminus \{A_1, A_2\} ) \cup B$, and note that $\OPT(\IPRED) = |B| = k-1$. The optimal solutions for $I$ and \IPRED are disjoint but $|\OPT(I) - \OPT(\IPRED)| = 1$, $\FP=0$ and $\FN=2$ (Figure~\ref{fig:lipExample}).
        \end{example}
        \begin{figure}
        \centering
        \scalebox{.8}{\lipCompExample}
        \caption{An illustration of Example~\ref{ex:completeness}. All requests appear in both $I$ and $\IPRED$, except $\{A_1,A_2\}$, which are false negatives.}\label{fig:lipExample}
        \end{figure}
        
        In this example, the error should be relatively small, independent of $k$. More generally, 
	the error measure must not grow with the dissimilarity between
        the optimal solutions for $I$ and $\IPRED$, but rather with
        the size of the optimal solution for $\FP$ and $\FN$.  This is
        guaranteed by the following property.
\begin{description}
\item[Lipschitz-completeness]\mbox{}
  
        An error
        measure is Lipschitz-complete, if for any $I$, \IPRED, the
        following holds.
	\[\eta(I,\IPRED) \leq  \OPT(\FP \cup
	\FN).\] 
	
\end{description}

\begin{proposition}
\label{propo:proper}
The error measure $\eta(\IPRED,I) = \OPT(\FP\cup\FN)$ is monotone, Lipschitz, and Lipschitz-complete.
\end{proposition}  

\begin{proof}
  We check all properties listed above. In all cases, we leverage the straightforward monotonicity property of $\OPT$ that the optimal profit of an input $I$ is no greater than the optimal profit of $I \cup \SET{x}$ for any item $x$, since \OPT can always choose to not include $x$ in the solution.
	\begin{itemize}
		\item Monotonicity: First, consider increasing
		the number of true positives. Let $x\in I\setminus \IPRED$.
		Since $x$ is a false negative, it
		may or may not have been counted in $\OPT(\FP\cup\FN)$, but removing it from
		\FN (thus adding it to \TP) cannot make $\OPT(\FP\cup\FN)$
		larger, i.e.,  \[\eta(I,\IPRED\cup\{ x\}) = \OPT(\FP\cup(\FN\setminus\{ x\})) \leq \OPT(\FP\cup\FN)=\eta(I,\IPRED)\,.\]
		Similarly, for any $y \in \IPRED \setminus I$, $\OPT((\FP\setminus\{ y\})
		\cup \FN)$ cannot be larger than $\OPT(\FP\cup\FN)=\eta(I,\IPRED)$, so
		\[\eta(I,\IPRED\setminus\{ y\}) = \OPT((\FP\setminus\{ y\})
		\cup \FN) \leq \OPT(\FP \cup \FN) = \eta(I,\IPRED)\,.\]
		
		\item Lipschitz property: We need to show that \[\OPT(\FP\cup\FN)
		\geq |\OPT(I)-\OPT(\IPRED)|\,.\]
		
		We note that
		\begin{align*}
			\OPT(I) & = \OPT((\IPRED \setminus \FP) \cup \FN)\\
			& \leq \OPT(\IPRED \cup \FN)\\
			& \leq \OPT(\IPRED) + \OPT(\FN)\,,
		\end{align*}
		which implies
		\[\OPT(I) - \OPT(\IPRED) \leq \OPT(\FN) \leq \OPT(\FP \cup \FN).
		\]
		
		\item Lipschitz-completeness: Follows trivially with the suggested bound,
		since $\eta =\OPT(\FP\cup\FN)$. 
	\end{itemize}
\end{proof}

\subsubsection{Alternative Error Measures.}
\label{sec:otherMeasures}
In what follows, we review a few alternative error measures that do
not have all of our desired properties of monotonicity, Lipschitz, and Lipschitz-completeness.
\begin{itemize}
	\item Hamming distance between the bit strings representing the
	request sequence and the predictions, that is, the total number of false positives and false negatives:
	\[|\FP|+|\FN| = |I \cup \IPRED| - |I \cap \IPRED| =  |(I \cup
	\IPRED) \setminus (I \cap \IPRED)|\]
	This measure fails Lipschitz-completeness. For instance, consider $I = \{(1,m)\}$ and $\IPRED = \bigcup_{i=1}^{m-1} (i,m)$. In this case, $\OPT(\FP \cup \FN) = 1$ (any pair of intervals intersect) and $|\FP| + |\FN|  = m-2$. Thus, $|\FP|+ |\FN| \not\leq \OPT(\FP \cup \FN)$.

        Note that using either $|\FP|$ or $|\FN|$ alone also fails the Lipschitz-completeness property, by the same example.
	
        \item Let $\OPT[I]$ and $\OPT[\IPRED]$ denote optimal solutions to the instances formed by $I$ and \IPRED, respectively.
        Using $\OPT[I]$ and $\OPT[\IPRED]$ instead of $I$ and \IPRED
	in the above measure:
	\[\big| \big(\OPT[I] \cup \OPT[\IPRED]\big) \setminus \big(\OPT[I]
	\cap \OPT[\IPRED]\big) \big|\]
	also fails Lipschitz-completeness, according to Example~\ref{ex:completeness}.
	
\end{itemize}

Some care must also be put into how the error is normalized.
Below, we give two examples where the normalized error becomes too small.

\begin{itemize}
	\item Normalizing the Hamming distance, we obtain the Jaccard
	distance:
	\[\frac{|I \cup \IPRED| - |I \cap \IPRED|}{|I \cup \IPRED|}\]
	This measure is sensitive to \emph{dummy requests}: The adversary can construct a bad
	input and then add a lot of intervals to $I \cap \IPRED$ that
	neither the algorithm nor \OPT will choose, driving down the error, and, thus, failing the Lipschitz property.
	
	\item Normalizing by the total number of possible
	intervals (order $m^2$), the adversary can make the error
        of any input arbitrarily small by ``scaling up'' each edge to an arbitrarily long path, without changing an algorithm's profit, therefore failing the Lipschitz property.
	
\end{itemize}

Note that by normalizing by the size of an optimal solution instead of
the size of a set of intervals, we avoid giving the adversary the
power to drive down the error value.

\section{Disjoint-Path Allocation}  \label{sect:dpa}
In this section, we show that a simple algorithm \TRUST for the disjoint path allocation problem has an optimal competitive ratio for the graph class trees.
\TRUST simply relies on the prediction being
correct. Specifically, 
it computes an optimal solution \IOPT in~\IPRED before
processing the first request. Then, it accepts
any request in \IOPT that arrives and rejects all others.

We first establish that, \emph{on any graph}, 
$\TRUST(\IPRED,I) \geq \OPT(I) - 2\eta(\IPRED,I) = (1-2\ernorm) \OPT(I)$. 
The proof follows by observing that
\begin{enumerate}
\item false negatives cause a deficit of at most $\OPT(\FN)$ in the
  schedule of $\TRUST$ compared to $\IOPT$,
\item
false positives cause a deficit of at most $\OPT(\FP)$ relative to $\IOPT$, compared to the optimal schedule for $I$, and
\item $ \OPT(\FP) + \OPT(\FN) \leq 2\OPT(\FP\cup\FN) = 2\eta$.
\end{enumerate}

\begin{theorem}
\label{pr:trustupper}
The algorithm \TRUST is strictly $(1-2\gamma)$-competitive.
\label{th:trustCompetitive}
\end{theorem}
\begin{proof}
	Since \IOPT is an optimal selection from $\TP \cup \FP$, the largest additional set of requests that 
	\OPT 
	would be able to accept from $I$ compared to \IOPT would be an
	optimal selection from~\FN.  Thus,
	$\OPT(I) \leq \OPT(\IOPT) + \OPT(\FN)$, and so
	$\OPT(\IOPT) \geq \OPT(I) - \OPT(\FN)$.
	
	Similarly, the largest number of requests that can be detracted from \TRUST is realized when requests that it planned
	to accept from \IOPT do not appear is $\OPT(\FP)$.
	Therefore, $\TRUST(\IPRED,I) \geq \OPT(\IOPT) - \OPT(\FP)$. Now,
	\[
	\begin{array}{rcl}
		\TRUST(\IPRED,I) & \geq & \OPT(\IOPT) - \OPT(\FP) \\
		& \geq & \OPT(I) - \OPT(\FN) - \OPT(\FP) \\
		& \geq & \OPT(I) - 2\OPT(\FP\cup\FN)\\
		& = & \OPT(I) - 2\eta(\IPRED,I)\\
		& = & \big(1- 2\ernorm\big) \OPT(I)         \end{array}
	\] 
\end{proof}

The following result shows that Theorem~\ref{pr:trustupper} is tight for the most interesting case of relatively small errors, and Corollary~\ref{cor:star} shows that this is the case even for trees. 

\begin{theorem}
  \label{thm:star}
  For $0 \leq \gamma \leq \frac14$,
  the competitive ratio of any
  deterministic algorithm for the online disjoint path allocation problem is at most $1-2\gamma$.
\label{th:lowerstar}
\end{theorem}
\begin{proof}
  Let \ALG be any deterministic algorithm.
  We prove that, for any $0 \leq \gamma \leq \frac14$, there exists a
  graph, a set of predicted requests $\IPRED_{w}$, and a request
  sequence $I_w$ such that $\ALG(\IPRED_{w},I_{w}) \leq
  \big(1-2\ernormw\big)\OPT(I_{w})$, \ernormw is arbitrarily close to
  $\gamma$, and $\OPT(I_w)$ is arbitrarily large. More specifically,
  we prove the following:
  
For any $0 \leq \gamma \leq \frac14$ and $0 < \eps < \frac13$, there exists a graph, a set of predicted requests $\IPRED_{w}$, and a request sequence $I_w$
such that
\begin{enumerate}
\item $\ALG(\IPRED_{w},I_{w}) \leq \big(1-2\ernormw\big)\OPT(I_{w})$\,,
\item $| \ernormw - \gamma | \leq \eps$\,, and
\item $\OPT(I_w) \geq \frac{1}{\eps}$\,.
\end{enumerate}
Let $p = \CEIL{\frac{1}{3\eps}}$ and
consider a set of $p$ disjoint copies of the star $S_8$ consisting of a center vertex with $8$ neighbor vertices.
The prediction is fixed, but the input sequence depends on the algorithm's actions. 
Given that stars do not share edges between them, the total error and the algorithm's profit are summed over all stars.

For star $0\leq i\leq p-1$, the non-center vertices are numbered $8i+j$, where $1\leq j\leq 8$, but we refer to these vertices by the value $j$.
For each star, the prediction is
\[
\IPRED_w = \SET{
(1, 2), (2, 3), (3, 4), (4, 5), (6, 7), (7, 8)}\,.
\]
Note that $\OPT(\IPRED_w) = 3$.

For each star $i$, we let $\eta_i$, $\ALG_i$, and $\OPT_i$ denote the contribution of that star to $\eta(\IPRED_w,I_w)$, $\ALG(\IPRED_w,I_w)$, and $\OPT(I_w)$, respectively.
The adversary chooses an integer $\ell$ between $0$ and $p$ and constructs the input sequence $I_w$ such that the following hold:
\begin{itemize}
\item For the first $\ell$ stars,
  \begin{itemize}
  \item $\eta_i = 1$\,,
  \item $\OPT_i \in \{3,4\}$\,, and
  \item $\ALG_i \leq \OPT_i - 2\eta_i$\,.
  \end{itemize}
\item For the last $p-\ell$ stars, the input is equal to the prediction, so
  \begin{itemize}
  \item $\eta_i = 0$\,
  \item $\OPT_i = 3$\,, and
  \item $\ALG_i \leq \OPT_i - 2\eta_i$ (trivially, since $\eta_i=0$).
  \end{itemize}
\end{itemize}
This will result in 
\begin{align*}
  \ALG(\IPRED_w,I_w)
  & \leq \sum_{i=0}^{p-1} \left(\OPT_i - 2 \eta_i\right)\\
  & =  \OPT(I_w) - 2 \eta(\IPRED_w,I_w)\\
  & = \big(1 - 2 \ernormw\big)\OPT(I_w)
\end{align*}
and
\[ \OPT(I_w) \geq 3p \geq \frac{1}{\eps}\,, \]
proving Items~1 and~3.

We now explain how the adversary constructs the input sequence $I_w$.
For $0 \leq i \leq \ell-1$,
$I_w$ starts with $\bm{\langle (2,3),(3,4),(6,7),(7,8) \rangle}$.
Note that the path between $2$ and $3$ shares an edge with the path
between $3$ and $4$, so \ALG can accept at most one of the requests
$(2,3)$ and $(3,4)$. The same is true for the requests $(6,7)$ and $(7,8)$.
The rest of the input depends on the actions of \ALG, as outlined in
the cases below.

\begin{figure}[!t]
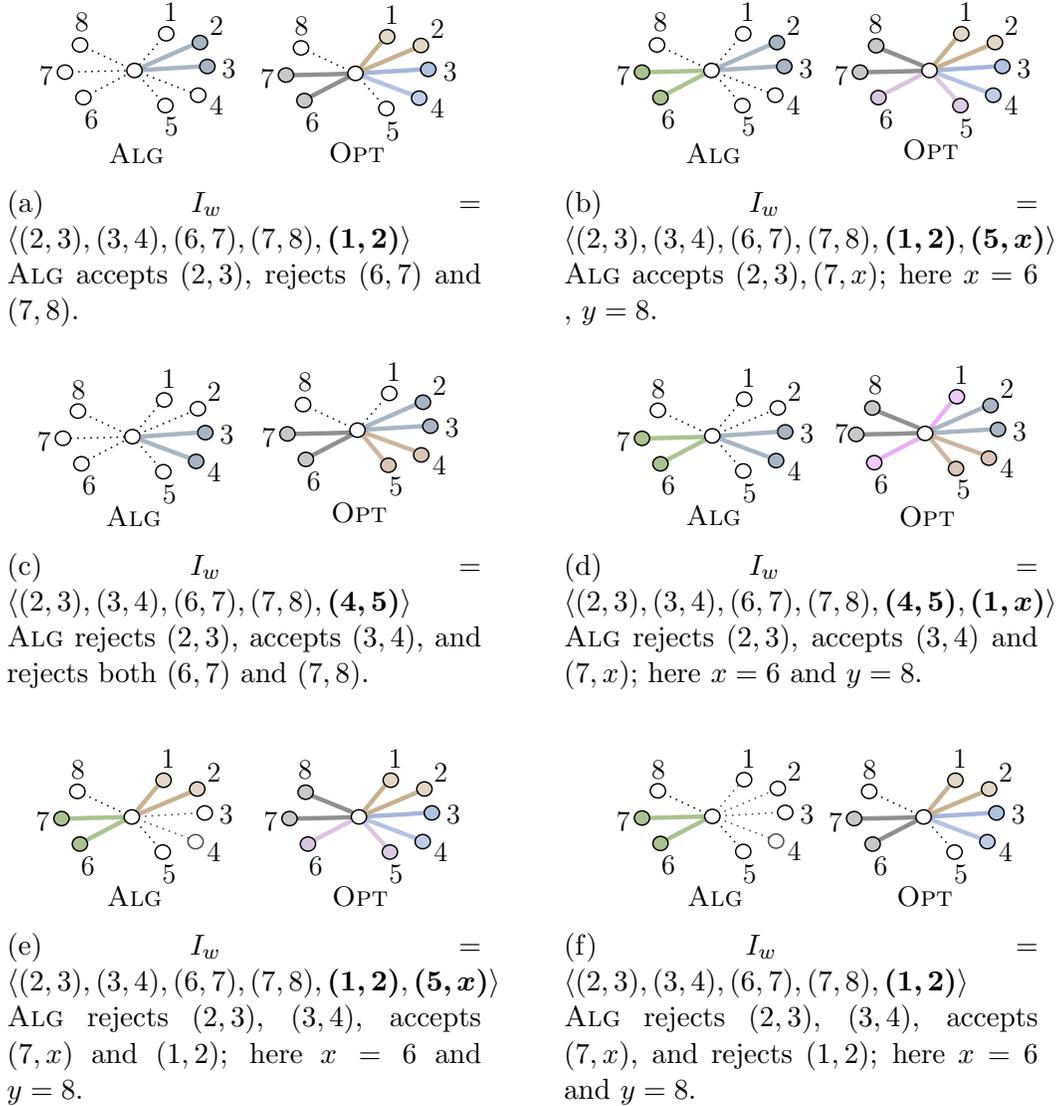

	\centering
	\begin{subfigure}[b]{0.46\textwidth}
		\centering
	\scalebox{.74}{\figStarOneA}
		\caption{$I_w = \langle (2,3),(3,4),(6,7),(7,8),\bm{(1,2)} \rangle$ \\\ALG accepts $(2,3)$, rejects $(6,7)$ and $(7,8)$.}
		\label{fig:staronea}
	\end{subfigure}\hfill \vspace*{3mm}
	\begin{subfigure}[b]{0.46\textwidth}	
		\hspace*{3mm}
\scalebox{.74}{\figStarOneB}
		\caption{$I_w = \langle (2,3),(3,4),(6,7),(7,8),\bm{(1,2)},\bm{(5,x)} \rangle$ \\ \ALG accepts $(2,3), (7,x)$; here $x=6$ , $y=8$.}
		\label{fig:staroneb}
	\end{subfigure}	
\vspace*{1mm}
	\begin{subfigure}[b]{0.46\textwidth}
	\centering
	\scalebox{.74}{\figStarTwoA}
	\caption{$I_w = \langle (2,3),(3,4),(6,7),(7,8),\bm{(4,5)} \rangle$ \\ \ALG rejects  $(2,3)$, accepts $(3,4)$, and rejects both $(6,7)$ and $(7,8)$.}
	\label{fig:startwoa}
\end{subfigure}\hfill
\begin{subfigure}[b]{0.46\textwidth}			\hspace*{3mm}
	\scalebox{.74}{\figStarTwoB}
	\caption{$I_w = \langle (2,3),(3,4),(6,7),(7,8),\bm{(4,5)},\bm{(1,x)} \rangle$ \\ \ALG rejects  $(2,3)$, accepts $(3,4)$ and $(7,x)$; here $x=6$ and $y=8$.}
	\label{fig:startwob}
\end{subfigure} \ \\ \ \\
\centering
\begin{subfigure}[b]{0.46\textwidth}
	\centering
	\scalebox{.74}{\figStarTwoC}
	\caption{$I_w = \langle (2,3),(3,4),(6,7),(7,8),\bm{(1,2)},\bm{(5,x)} \rangle$ \\ \ALG rejects $(2,3)$, $(3,4)$, accepts $(7,x)$ and $(1,2)$; here $x=6$ and $y=8$.}
	\label{fig:startwoc}
\end{subfigure}\hfill
\begin{subfigure}[b]{0.46\textwidth}			\hspace*{3mm}
	\scalebox{.74}{\figStarTwoD}
	\caption{$I_w = \langle (2,3),(3,4),(6,7),(7,8),\bm{(1,2)} \rangle$ \\ \ALG rejects $(2,3)$, $(3,4)$, accepts $(7,x)$, and rejects $(1,2$); here $x=6$ and $y=8$.}
	\label{fig:startwod}
\end{subfigure}	
\caption{Illustration of the proof of Theorem~\ref{th:lowerstar}. Highlighted edges indicate paths between accepted pairs.}

\end{figure}

\begin{description}
\item[Case \ALG accepts $\bm{(2,3)}$:]
  In this case, the next request to arrive is $\bm{(1,2)}$.
  \ALG cannot accept any of the requests $(1,2)$ and $(3,4)$.
  The predicted request $(4,5)$ does not arrive, i.e., $(4,5)$ is a false
positive. Note that there are no other false positives.
  \begin{description}
  \item[Subcase \ALG rejects both $\bm{(6,7)}$ and $\bm{(7,8)}$:]
In this case, no further requests arrive (see
Figure~\ref{fig:staronea}), so there are no false negatives.
Thus, $\OPT(\FN\cup\FP)=1$. Moreover, \OPT accepts three requests, $(1,2)$ and $(3,4)$
combined with $(6,7)$ or $(7,8)$, while \ALG accepts only
$(2,3)$.
\item[Subcase \ALG accepts $\bm{(6,7)}$ or $\bm{(7,8)}$:]
We let $\{x,y\}=\{6,8\}$ such that \ALG accepts $(7,x)$ and rejects $(7,y)$.
Now, a false
negative, $\bm{(5,x)}$, arrives (see Figure~\ref{fig:staroneb}). Since
$(5,x)$ shares an edge with the false positive $(4,5)$, $\OPT(\FN\cup\FP)=1$.
$\ALG$ accepts $\{ (2,3),(7,x)\}$, and \OPT accepts $\{ (1,2),(3,4),(5,x),
(7,y)\}$. 
  \end{description}
In both subcases, $\eta_i=1$, $\OPT_i \in \{3,4\}$, and $\ALG_i = \OPT_i-2 = \OPT_i - 2\eta_i$.

\item[Case \ALG rejects $\bm{(2,3)}$:] \mbox{}

  \begin{description}
  \item[Subcase \ALG accepts $\bm{(3,4)}$:]
The next request to arrive is $\bm{(4,5)}$, which the algorithm cannot accept. The request $(1,2)$ does not arrive, so
it is a false positive.
    \begin{description}
    \item[Subsubcase \ALG rejects both $\bm{(6,7)}$ and $\bm{(7,8)}$:]
    In this case, no further requests arrive
    (see Figure~\ref{fig:startwoa}).
    Thus, \ALG accepts only $(3,4)$, \OPT accepts $(2,3)$ and $(4,5)$
    together with $(6,7)$ or $(7,8)$, and
    $\OPT(\FN\cup\FP)=1$.

    \item[Subsubcase \ALG accepts $\bm{(6,7)}$ or $\bm{(7,8)}$:]
As above, we define $\{x,y\} = \{6,8\}$ such that \ALG accepts $(7,x)$ and rejects $(7,y)$.
Now, a false
negative, $\bm{(1,x)}$, arrives  (see Figure~\ref{fig:startwob}). Since $(1,2)$ and $(1,x)$ share an edge, $\OPT(\FN\cup\FP)=1$.
Moreover, $\ALG$ accepts $\{ (3,4),(7,x)\}$, and \OPT accepts the set $\{ (1,x),(2,3),(4,5),
(7,y)\}$.
    \end{description}
  In both subsubcases, $\eta_i=1$, $\OPT_i \in \{3,4\}$, and $\ALG_i = \OPT_i-2 = \OPT_i - 2\eta_i$.

  \item[Subcase \ALG rejects $\bm{(3,4)}$:]
    The next request to arrive is $\bm{(1,2)}$.
    The request $(4,5)$ is a false positive.
    \begin{description}
    \item[Subsubcase \ALG accepts $\bm{(6,7)}$ or $\bm{(7,8)}$:]
As above, we define $\{x,y\} = \{6,8\}$ such that \ALG accepts $(7,x)$ and rejects $(7,y)$.

      \begin{description}
      \item[Subsubsubcase \ALG accepts $\bm{(1,2)}$:]
The final request is
a false
negative, $\bm{(5,x)}$  (see Figure~\ref{fig:startwoc}).
Since $(4,5)$ and $(5,x)$ share an edge, $\OPT(\FN\cup\FP)=1$.
$\ALG$ accepts $\{ (1,2),(7,x)\}$ and \OPT accepts
$\{ (1,2),(3,4),(5,x),(7,y)\}$.

      \item[Subsubsubcase \ALG rejects $\bm{(1,2)}$:]
        In this case, no further requests arrive (see Figure~\ref{fig:startwod}).
        Thus, $\ALG$ accepts $\{ (7,x)\}$, \OPT accepts
$\{ (3,4),(5,x),(7,y)\}$, and
$\OPT(\FN\cup\FP)=1$. 
      \end{description}
      In both subsubsubcases, $\eta_i=1$, $\OPT_i \in \{3,4\}$, and $\ALG_i = \OPT_i-2 = \OPT_i - 2\eta_i$.
    \item[Subsubcase \ALG rejects both $\bm{(6,7)}$ and $\bm{(7,8)}$:]
      In this case, no further requests arrive.
      Thus, the profit of \ALG is $1$ if it
      accepts $(1,2)$ and $0$ otherwise, while \OPT accepts
      $(1,2)$ and $(3,4)$ together with $(6,7)$ or $(7,8)$, 
and
$\OPT(\FN\cup\FP)=1$. 
Thus, $\eta_i=1$, $\OPT_i=3$, and $\ALG_i \leq \OPT_i-2 = \OPT_i - 2\eta_i$.
    \end{description}
  \end{description}
\end{description}

This completes the proof of Items~1 and~3.
For Item~2, note that
\begin{itemize}
\item $\ernormw = 0$, for $\ell=0$, and
\item $\displaystyle \frac14 \leq \ernormw \leq \frac13$, for $\ell=p$\,.
\end{itemize}
Also note that incrementing $\ell$ by one increases $\eta(\IPRED_w,I_w)$ by $1$ and does not decrease $\OPT(I_w)$.
Thus, since $\OPT(I_w) \geq 1/\eps$, incrementing $\ell$ adds at most
\[ \Delta_{\gamma} \leq \frac{1}{1/\eps} = \eps \]
to \ernormw. This proves that for any $0 \leq \gamma \leq \frac14$, the adversary can choose $\ell$ such that
\[ | \ernormw - \gamma | \leq \eps \,.\]
\end{proof}

\begin{corollary}
\label{cor:star}
Any deterministic algorithm for the online disjoint path problem on star
graphs with arbitrarily high degree or trees with arbitrarily many
vertices of degree at least~$8$ has competitive ratio at most
$1-2\gamma$, $1 \leq \gamma \leq \frac14$.
\end{corollary}
\begin{proof}
The stars used in the proof of Theorem~\ref{thm:star} can be connected in several ways to form a connected graph. For instance, the proof still holds, if the vertex $8i+1$ is connected by an edge to the vertex $8(i+1)+2$, $0 \leq i \leq p-2$, or if all center vertices are identified resulting in a star of one center vertex with $8p$ neighbors.
This is true, since the paths are only required to be edge disjoint, not vertex disjoint. 
\end{proof}

\section{Interval Scheduling}
In this section, we first show tight upper and lower bounds on the competitive ratio of deterministic algorithms for interval scheduling.
We then study trade-offs between consistency and robustness for
deterministic and randomized algorithms and present some experimental
results.
Recall that $m$ denotes the number of edges on the given path.

\subsection{A General Upper Bound for Deterministic Algorithms}

As an introduction to the difficulties in designing algorithms for the problem,
we start by proving a general lower bound.
We show that for $0 < \gamma < 1$, no deterministic
algorithm can have a competitive ratio better than $1-\gamma$:
\begin{theorem}
  \label{thm:generallower}
  For any $0 \leq \gamma \leq 1$, the competitive ratio of any
  deterministic algorithm for the online interval scheduling problem is at
  most $1-\gamma$.
\end{theorem}

\begin{proof}
  Let $\ALG$ be any deterministic algorithm.
  We prove that, for any $0 \leq \gamma \leq 1$, there exists a path graph, a set of predicted requests $\IPRED_{w}$, and a request
  sequence $I_w$ such that $\ALG(\IPRED_{w},I_{w}) \leq
  \big(1-\ernormw\big)\OPT(I_{w})$, \ernormw is arbitrarily close to
  $\gamma$, and $\OPT(I_w)$ is arbitrarily large.
  More specifically, we prove the following:
  
For any $0 \leq \gamma \leq 1$ and $0 < \eps < 1$,
there is an input sequence $I_w$ and a set of predictions $\IPRED_w$ such that
\begin{enumerate}
  \item $\displaystyle \ALG(\IPRED_w,I_w) \leq \big(1 - \ernormw\big)\OPT(I_w)$\,,
  \item $|\ernormw - \gamma| \leq \eps$\,, and
  \item $\displaystyle \OPT(I_w) \geq \frac{1}{\eps^2}$\,.
\end{enumerate}
For any $0 < \eps < 1$, let $c = \CEIL{\frac{1}{\eps}}$ and  $p = \CEIL{\frac{1}{\eps^2}}$.
The prediction consists of $2p$ requests:
\[\IPRED_w = \bigcup_{i=0}^{p-1} \Big\{\big(c i, c (i+1) \big), \big(c i, c
i+1\big)\Big\}\,.\]
The input $I_w$ is formed by $p$ phases, $0 \leq i \leq p-1$.
For each phase $i$, we let $\eta_i$, $\ALG_i$, and $\OPT_i$ denote the contribution of that phase to $\eta(\IPRED_w,I_w)$, $\ALG(\IPRED_w,I_w)$, and $\OPT(I_w)$, respectively.
The adversary chooses an integer $\ell$ between $0$ and $p-1$ and does the following.
\begin{itemize}
\item For $0 \leq i \leq \ell-1$, the
$i$th~phase starts with the true positive $(c i,c (i+1))$, and the
remainder of the phase depends on whether the algorithm accepts this
request or not.
\begin{description}
\item[Case \ALG accepts $\bm{(c i, c (i+1))}$:]
  In this case, the phase
continues with
\[\SETOF{(c i+j, c i+(j+1))}{0\leq j\leq c-1}\,.\]
The first of these requests is a true positive, and the other $c-1$ are
false negatives.

Note that $\ALG$ cannot accept any of these $c$
requests. The optimal algorithm rejects the
original request $(c i, c (i+1))$ and accepts all of the $c$
following unit-length requests.
We conclude that
\begin{itemize}
  \item $\eta_i=c-1$,
  \item $\OPT_i=c$, and
  \item $\ALG_i = 1 = \OPT_i-\eta_i$.
\end{itemize}

\item [Case \ALG rejects $\bm{(c i, c (i+1))}$:]
  In this case, the phase ends
with no further requests. Thus, $(c i, c i+1)$
is a false positive, and we obtain
\begin{itemize}
  \item $\eta_i=1$,
  \item $\OPT_i=1$, and
  \item $\ALG_i = 0 = \OPT_i-\eta_i$.
\end{itemize}

\end{description}

\item For $\ell \leq i \leq p-1$, the $i$th phase consists of the true positives $(c i, c (i+1))$ and $(c i, c
i+1)$. Thus,
\begin{itemize}
  \item $\eta_i=0$,
  \item $\OPT_i=1$, and
  \item $\ALG_i \leq 1 = \OPT_i-\eta_i$
\end{itemize}
\end{itemize}

Since the intervals in
$\FP \cup \FN$ are disjoint, we can write
\[ \eta(\IPRED_w,I_w) = \OPT(\FP\cup \FN) = \sum_{i=0}^{\ell-1} \eta_i \,.\]
Thus,
\begin{align*}
  \ALG(\IPRED_w,I_w)
  & = \sum_{i=0}^{p-1} \ALG_i \\
  & \leq \sum_{i=0}^{p-1} \left(\OPT_i - \eta_i\right)\\
  & =  \OPT(I_w) -  \eta(\IPRED_w,I_w)\\
  & = \big(1 - \ernormw\big)\OPT(I_w)\,.
\end{align*}  
This proves Item~1.

Since \OPT accepts at least one interval in each phase,
\[ \OPT(I_w) \geq p \,, \]
which proves Item~3.

Finally, for Item~2, note that
\begin{itemize}
\item $\ernormw=0$, for $\ell=0$, and
\item $\displaystyle \frac{c-1}{c} \leq \ernormw \leq 1$, for $\ell=p$,
\end{itemize}
where
\[ 1-\frac{c-1}{c} = \frac1c = \frac{1}{\CEIL{\frac{1}{\eps}}} \leq \eps  \,. \]
Also note that incrementing $\ell$ by one, increases $\eta(\IPRED_w,I_w)$ by $1$ or $c-1$ and does not decrease $\OPT(I_w)$.
Thus, incrementing $\ell$ adds at most
\[ \Delta_{\gamma} = \frac{c-1}{p} < \frac{\frac{1}{\eps}}{\CEIL{\frac{1}{\eps^2}}} \leq \eps \]
to \ernormw.
This proves that for any $0 \leq \gamma \leq 1$, the adversary can choose $\ell$ such that
\[ | \ernormw - \gamma | < \eps \,. \]
\end{proof}

\subsection{\TRUST}

The next theorem gives an upper bound on the competitive ratio of
\TRUST which is lower than the general upper bound of
Theorem~\ref{thm:generallower}.
The proof uses an adversarial sequence similar to that of Theorem~\ref{thm:generallower}. 

\begin{theorem}
\label{th:intervaltrustlower}
For the online interval scheduling problem, the competitive ratio of \TRUST is at most $1-2\gamma$.
\end{theorem}
\begin{proof}
  We prove that, for any $0 \leq \gamma \leq \frac12$, there exists a path
  graph, a set of predicted requests $\IPRED_{w}$, and a request
  sequence $I_w$ such that $\TRUST(\IPRED_{w},I_{w}) \leq
  \big(1-2\ernormw\big)\OPT(I_{w})$, \ernormw is arbitrarily close to
  $\gamma$, and $\OPT(I_w)$ is arbitrarily large. More specifically,
  we prove the following:
  
For any $0 \leq \gamma \leq \frac12$ and $0 < \eps \leq 1$, there exists a path graph, a set of predicted requests $\IPRED_{w}$, and a request sequence $I_w$
such that
\begin{enumerate}
\item $\TRUST(\IPRED_{w},I_{w}) \leq \big(1-2\ernormw\big)\OPT(I_{w})$\,,
\item $| \ernormw - \gamma | \leq \eps$\,, and
\item $\OPT(I_w) \geq \frac{1}{\eps}$\,.
\end{enumerate}

For any $0 < \eps \leq 1$, let $p = \CEIL{\frac{1}{\eps}}$.
	Let the prediction be a set of length-$2$ intervals:
	\[\IPRED_w = \bigcup_{i=0}^{p-1} \big\{(3i,3i+2), (3i+1,3i+3)\big\}\,.\]
	\TRUST chooses an optimal solution \IOPT from $\IPRED_w$. For each $i$,
	\IOPT will contain either $(3i, 3i+2)$ or $(3i+1, 3i+3)$.
        The adversary chooses an integer $\ell$ between $0$ and $p$.
        \begin{itemize}
        \item For $0 \leq i \leq \ell-1$, the adversary does the following.
        
	If $(3i, 3i+2)$ is in \IOPT, that interval will be in \FP, and
	\OPT will select $(3i+1, 3i+3)$, which will be a \TP-interval.
	Further, $I_w$ will contain the \FN-interval, $(3i, 3i+1)$.
	
	If, instead, $(3i+1, 3i+3)$ is in \IOPT, that interval will be in \FP, and
	\OPT will select $(3i, 3i+2)$, which will be a \TP-interval.
	Further, $I_w$ will then contain the \FN-interval, $(3i+2, 3i+3)$.

        \item For $\ell \leq i \leq p$, the adversary gives the two
          predicted intervals, $(3i, 3i+2)$ and $(3i+1, 3i+3)$.

        \end{itemize}

	Thus, \[ \OPT(I_w)=2 \ell + (p-\ell) = p+\ell \geq
        \frac{1}{\eps}\,,\]
        proving Item~3.

        For each $i < \ell$, the interval in \FP and the interval in
	\FN overlap, so \[ \error(\hat{I}_w,I_w) = \OPT(\FN \cup \FP) = \ell\,. \] 
        Since the first $\ell$ intervals in $I^*$ are false positives,
\begin{align*}
  \TRUST(\IPRED_w, I_w)
  & = p-\ell\\
  & = \OPT(I_w) - 2 \ell \\
  & = \OPT(I_w) - 2\error(\hat{I}_w,I_w) \\
  & =  \big(1-2\ernormw\big)\OPT(I_w)\,,
\end{align*}
proving Item~1.

For Item~2, note that $\ernormw = 0$, for $\ell=0$, and $\ernormw = p/2p = 1/2$, for $\ell=p$.
Also note that incrementing $\ell$ by one increases $\eta(\IPRED_w,I_w)$ by~$1$ and does not decrease $\OPT(I_w)$.
Thus, since $\OPT(I_w) \geq 1/\eps$, incrementing $\ell$ adds at most
\[ \Delta_{\gamma} \leq \frac{1}{1/\eps} = \eps \]
to~\ernormw.
This proves that for any $0 \leq \gamma \leq \frac12$, the adversary can choose $\ell$ such that
\[ | \ernormw - \gamma | \leq \eps \,.\]
\end{proof}
In light of Theorems~\ref{th:trustCompetitive} and~\ref{th:intervaltrustlower}, the competitive ratio of \TRUST for interval scheduling is $1-2\gamma$. Meanwhile, Theorem~\ref{thm:generallower} gives an upper bound of $1-\gamma$ on the competitiveness of deterministic algorithms. This gap suggests potential for improvement, which we explore in the next section.

\subsection{\TRUSTGREEDY}
In this section, we introduce an algorithm \TRUSTGREEDY, which achieves an optimal competitive ratio for interval scheduling. 

\subsubsection{The algorithm.}

\TRUSTGREEDY starts by choosing an optimal solution, $\IOPT$, from the predictions
in \IPRED.
This optimal offline solution is selected by repeatedly including
an interval that ends earliest possible
among those in $\IPRED$ that do not overlap any
already selected intervals.

During the online processing after this initialization,
\TRUSTGREEDY maintains an updated plan, $A$.
Initially, $A$ is \IOPT.
When a request, $r$, is contained in $A$, it is accepted.
When a request, $r$, is in~\FN, \TRUSTGREEDY accepts if $r$ overlaps no
previously accepted intervals and can be accepted by
replacing at most one other interval in $A$ that ends no earlier than~$r$.
In that case, $r$ is added to $A$, possibly replacing an overlapping interval
to maintain the feasibility of~$A$ (no two intervals overlap).

As a comment, only the first interval from \FN that replaces an interval~$r$
in the current~$A$ is said to ``replace'' it. There may be other intervals
from \FN that overlap~$r$ and are accepted by \TRUSTGREEDY, but they are not said to
``replace'' it. Figure~\ref{fig:tgExample} provides an illustration.

\begin{figure}
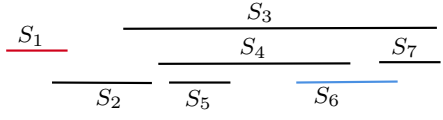

\centering
\scalebox{1.1}{\tgExample}
\caption{An illustration of \TRUSTGREEDY: Suppose $I = \langle S_1, S_2, S_3, S_4, S_5, S_7 \rangle$ and $\IPRED = (I \setminus {S_1}) \cup {S_6}$, so that $S_1$ is a false negative and $S_6$ is a false positive. The optimal solution for $\IPRED$ constructed by \TRUSTGREEDY is $I^* = \{S_2, S_5, S_6\}$. Initially, the set $A$ is initialized as $I^*$. Upon the arrival of $S_1$, it replaces $S_2$, and $A$ becomes $\{S_1, S_5, S_6\}$. Eventually, the algorithm accepts $\{S_1, S_5\}$.\label{fig:tgExample}}
\end{figure}

\subsubsection{Analysis.}

Let $\TGI$ denote the set of intervals chosen by \TRUSTGREEDY on input~$I$
and prediction~\IPRED, and \OPTI the intervals chosen by some optimal offline algorithm.
We define the following subsets of \TG and \OPT:
\begin{itemize}
\item $\TGTP = \TG \cap \IPRED = \TG \cap \TP$ \\ $\OPTHAT = \OPTI
\cap \IPRED = \OPTI \cap \TP$
\item $\TGFN = \TGI \cap \FN$ \\ $\OPTFN = \OPTI \cap \FN$ 
\end{itemize}
Note that since $I = \TP \cup \FN$,
  \[ \TG = \TGTP \cup \TGFN \text{ and } \OPT = \OPTHAT \cup \OPTFN\,. \]

In the example of Figure~\ref{fig:tgExample},
we have $\OPT = \{S_1,S_5,S_7\}$, $\TG = \{S_1,S_5\}$, $\TGTP = \{S_5\}$, $\OPTHAT = \{S_5,S_7\}$, $\TGFN = \{S_1\}$ and $\OPTFN = \{S_1\}$.

\begin{lemma}
\label{lemma:IOPT}
Each interval $i \in \OPTHAT$ overlaps an interval in \IOPT extending no
further to the right than $i$.
\end{lemma}
\begin{proof}
Assume to the contrary that there is no interval in \IOPT that
overlaps $i$ and ends no later than $i$.  If $i$ does not overlap
anything in~\IOPT, we could have added $i$ to \IOPT and have a
feasible solution (non-overlapping intervals), contradicting the
fact that \IOPT is optimal.
Thus, $i$ must overlap an interval $r \in \IOPT$, which, by assumption,
must end strictly later than~$i$.

If $r$ is the only interval in \IOPT overlapping $i$, this contradicts
the construction of~\IOPT, since $i$ would have been in 
\IOPT instead of~$r$.

Otherwise, there is an interval $s \in \IOPT \setminus \SET{r}$ overlapping $i$.
Since $r$ and $s$ are both in \IOPT and $r$ overlaps the right
endpoint of $i$,
$s$ ends no later than the start of $r$, meaning that $s$ overlaps $i$
and ends before $i$, contradicting the assumption that no interval in
\IOPT overlaps $i$ and ends no later than $i$.
\end{proof}

We let $U$ denote the set of intervals in $\IOPT \cap \FP$ that are
not replaced during the execution of \TRUSTGREEDY, i.e., those
intervals that stay in the plan $A$ but never show up in $I$.
We define a set \OFN consisting of a copy of each interval in \OPTFN and
let $\F = \OFN \cup U$.
We define a mapping \[f \colon \OPTI \rightarrow \TGI \cup \F\] as
follows.
For each $i \in \OPTI$:
\begin{enumerate}
\item \label{stepI} If there is an interval in \IOPT that overlaps $i$ and ends no
  later than $i$, then let $r$ be the rightmost such interval.
  \begin{enumerate}
  \item \label{stepFU} If $r \in U \cup \TGTP$, then $f(i)=r$.
  \item \label{stepTGFN} Otherwise, $r$ has been replaced by some
    interval $t \in \FN$.
    In this case, $f(i)=t$.
  \end{enumerate}
\item Otherwise, by Lemma~\ref{lemma:IOPT}, $i$ belongs to \OPTFN.
  \begin{enumerate}
  \item \label{stepTGmI}If there is
    \begin{itemize}
      \item an interval in \TGFN that overlaps
        $i$ and ends no later than $i$ and
      \item an interval in $U$ that overlaps
          $i$'s right endpoint,
    \end{itemize}
    let $r$ be the rightmost
    interval in \TGFN that overlaps $i$ and ends no later than $i$.
    In this case, $f(i)=r$.
  \item \label{stepFFN} Otherwise, let $o_i$ be the copy
    of $i$ in \OFN. In this case, $f(i)=o_i$.
  \end{enumerate}
\end{enumerate}

In the example of Figure~\ref{fig:tgExample}, we have $U = \{S_6\}$ and $\OFN = \{S'_1\}$, where $S'_1$ is a copy of $S_1$. Thus, we have $\F = \{S'_1, S_6\}$ and the mapping $f$ is from $\{S_1, S_5, S_7 \}$ to $\{S'_1, S_1, S_5, S_6\}$.
The mapping would be $S_1\rightarrow S'_1$, $S_5 \rightarrow S_5$, $S_7 \rightarrow S_6$. 

We let \FEAS denote the subset of \F mapped to by $f$
and note that in step~\ref{stepFU}, intervals are
added to $F \cap U$ when $r\in U$.
In step~\ref{stepFFN}, all
intervals are added to $F \cap \OFN$.

We prove that $f$ is an injection (Lemma~\ref{lemma:injection}) and $F$ is
a feasible solution (Lemma~\ref{lemma:feasible}), and conclude that
\TRUSTGREEDY has an optimal competitive ratio (Theorem~\ref{th:trustgreedymain}).

\begin{lemma}
\label{lemma:injection}
The mapping $f$ is an injection.
\end{lemma}
\begin{proof}
Intervals in $U \cup \TGTP$ are only mapped to in
step~\ref{stepFU}.
If an interval $i \in \OPTI$ is mapped to an interval $r \in U \cup
\TGTP$, $i$ overlaps the right endpoint of $r$. There can be only
one interval in \OPTI overlapping the right endpoint of $r$, so this part of the
mapping is injective.

Intervals in \TGFN are only mapped to in steps~\ref{stepTGFN}
and~\ref{stepTGmI}.
In step~\ref{stepTGFN}, only intervals that replace intervals in \IOPT are
mapped to.
Since each interval in \TGFN replaces at most one interval in \IOPT and
the right endpoint of each interval in \IOPT overlaps at most one interval
in $\OPTI$, no interval is mapped to twice in step~\ref{stepTGFN}.
If, in step~\ref{stepTGmI}, an interval, $i$, is mapped to an interval,
$r$, $i$ overlaps the right endpoint of $r$.
There can be only one interval in \OPTI overlapping the right endpoint
of $r$, so no interval is mapped to twice in step~\ref{stepTGmI}.

We now argue that no interval is mapped to in both steps~\ref{stepTGFN}
and~\ref{stepTGmI}.
Assume that an interval, $i_1$, is mapped to an interval, $t$, in
step~\ref{stepTGFN}.
Then, there is an interval, $r$, such that $r$ overlaps the right
endpoint of $t$ and $i_1$ overlaps the right endpoint of $r$.
This means that the right endpoint of $i_1$ is no further to the left
than the right endpoint of $t$.
Assume for the sake of contradiction that an interval $i_2 \neq i_1$
is mapped to $t$ in step~\ref{stepTGmI}.
Then, $i_2$ overlaps the right endpoint of $t$, and there is an
interval, $u \in U$, overlapping the right endpoint of $i_2$.
Since $i_2$ overlaps $t$ and $t$ does not extend to the right of $i_1$, $i_2$ must be to the left of $i_1$.
Since $i_2$ is mapped to $t$, $t$ extends no further to the right than
$i_2$.
Thus, since $r$ overlaps both $t$ and $i_1$, $r$ must overlap the
right endpoint of $i_2$, and hence, $r$ overlaps $u$.
This is a contradiction since $r$ and $u$ are both in \IOPT.

Intervals in $F \cap \OFN$ are only mapped to in
step~\ref{stepFFN} and no two intervals are mapped to the same
interval in this step.
\end{proof}

\begin{lemma}
\label{lemma:feasible}
The subset \FEAS of \F mapped to by $f$ is a feasible
solution.
\end{lemma}
\begin{proof}
We first note that $F \cap U$ is feasible since $F \cap U
\subseteq U \subseteq \IOPT$ and \IOPT is feasible.
Moreover, $\FEAS \cap \OFN$ is feasible since the intervals of $\FEAS
\cap \OFN$ are identical to the corresponding subsets of \OPTI.
Thus, we only need to show that
no interval in $\FEAS \cap U$ overlaps any interval in $\FEAS \cap \OFN$.

Consider an interval $u \in \FEAS \cap U$ mapped to from an interval
$i \in \OPTI$.
Since $i$ is not mapped to its own copy in \F, its copy does not belong to \FEAS.
Since $i \in \OPTI$, no interval in $\FEAS \cap \OFN$ overlaps $i$.
Thus, it is sufficient to argue that $\FEAS \cap
\OFN$ contains no interval strictly to the left of $i$ overlapping $u$.

Assume for the sake of contradiction that there is an interval $\ell \in
\FEAS \cap \OFN$ to the left of $i$ overlapping $u$. 
Since $\ell$ ended up in \FEAS although its right endpoint is
overlapped by an interval from $U$, there is no interval in \IOPT (because of
step~\ref{stepI} in the mapping algorithm) or in $\TGFN$ (because of step~\ref{stepTGmI} in the
mapping algorithm)
overlapping $\ell$ and ending no later than $\ell$.
Thus, $\IOPT \cup \TGFN$ contains no interval strictly to the left of $u$
overlapping $\ell$.
This contradicts the fact that $u$ has not been replaced since the
interval in \OPTFN corresponding to $\ell$ could have replaced it. 
\end{proof}

Using Lemmas~\ref{lemma:injection} and \ref{lemma:feasible}, we obtain
a lower bound matching the general upper bound of Theorem~\ref{thm:generallower}:

\begin{theorem}
\label{th:trustgreedymain} 
The competitive ratio of \TRUSTGREEDY for the online interval scheduling problem is at least $1-\gamma$.
\end{theorem}
\begin{proof}
	We show that
	\begin{align*}\TRUSTGREEDY(\IPRED,I) & \geq \OPT(I) - \OPT(\FP \cup \FN) \\ & = \big(1-\ernorm\big)\OPT(I):\end{align*}
	\begin{align*}
		\OPT(I) & \leq |\TGI| + |\FEAS|,
		\text{ since, by Lemma~\ref{lemma:injection}, $f$ is an injection}\\
		& \leq |\TGI| + \OPT(\F),
		\text{ since, by Lemma~\ref{lemma:feasible}, \FEAS is feasible} \\
		&\leq |\TGI| + \OPT(\FP \cup \FN), \text{ since } U \subseteq
		\FP \text{ and } \OPTFN \subseteq \FN  
	\end{align*}
	\end{proof}

\subsection{Consistency-Robustness Trade-off}

We study the trade-off between the competitive ratio of the interval scheduling algorithm when predictions are error-free (consistency) and when predictions are adversarial (robustness).

\subsubsection{General Upper Bounds}
The following proposition shows an obvious trade-off between the consistency and robustness of deterministic algorithms. 

\begin{proposition}
For the online interval scheduling problem on a path of $m$ vertices, any deterministic algorithm with consistency $\alpha \in \Theta(1)$ has robustness $\beta \in O(\frac{1}{m})$.
\label{prop:deterconstrob}
\end{proposition}
\begin{proof}
  If a deterministic algorithm \ALG has a consistency of $\alpha \in \Theta(1)$, there exists a non-negative constant $b$ such that, for any input $I$ with a correct prediction $\IPRED=I$, $\ALG(\IPRED,I) \geq \alpha \OPT(I) - b$.
  Let $p = \CEIL{(b+1)/\alpha}$ and $\ell = \FLOOR{m/p}$ and consider a prediction consisting of $p$ disjoint intervals of length $\ell$:
  \[ \IPRED = \bigcup_{i=0}^{p-1} \{ (i\ell, (i+1)\ell) \} \,.\]
  Note that $\OPT(\IPRED)=p$.
  Thus, if the input starts with the $p$ predicted intervals, \ALG must accept at least one of them, since $\alpha p - b \geq 1$.
  Assume now that, for each of the predicted intervals accepted by \ALG, the input continues with $\ell$ disjoint unit length intervals overlapping that interval.
  Since \OPT will accept each of the unit length intervals, this shows that the robustness of \ALG is
 \[ \beta \leq \frac{1}{\ell} \leq \frac{p}{m-p+1} \,. \]
    Since $p= \CEIL{(b+1)/\alpha}$ is a constant, $\beta\in O\left(\frac1m\right)$.
\end{proof}

The more interesting case is randomized algorithms. The proof of the following was inspired
by the proof of Theorem 13.8 in~\cite{BE98}
for the online case without predictions, and that $\mathrm{\Omega}(\log m)$ result was
originally proven in~\cite{ABFR94}.
We address the more relevant case of trade-offs when the robustness is non-trivial.

\begin{theorem}\label{th:constrob}
  Consider a (possibly randomized) $\alpha$-consistent and
  $\beta$-robust algorithm \ALG for the online interval scheduling problem.
  If there exists an $\eps$, $0<\eps<1$, such that $\beta \geq 1/m^{1-\eps}$, then
\begin{align*}
& \alpha \leq 1 - \beta \cdot \left( \frac{\eps \cdot \log m}{2} - O(1) \right)
\text{~~and~~}\\
&\beta \leq (1-\alpha) \cdot \frac{2}{ \eps \cdot \log m} + O \left( \frac{1}{\log^2 m} \right) \,.
\end{align*}
\end{theorem}

\begin{proof}
Let $r = \lfloor\log m\rfloor-1$ and let $m'=2^{r+1}$.
Consider an input sequence
\begin{align*}
& \sigma=\langle I_0, I_1,\ldots, I_{r+1} \rangle, \text{ where }\\
& I_i=\langle (0,m'/2^i),(m'/2^i,2m'/2^i),\ldots,(m'-m'/2^i,m')\rangle, \text{ for } 0\leq i\leq r+1\,.\\
\intertext{Note that $I_i$ consists of $2^i$ disjoint intervals of length $m'/2^i$. Let}
&\sigma_i=\langle I_0, I_1,\ldots, I_i\rangle,  \text{ for } 0\leq i\leq r.
\end{align*}

In order to maximize the number of small intervals that can be accepted
if they arrive, an algorithm would minimize the (expected) fraction of the line
occupied by the larger intervals, to leave space for the small intervals, while
maintaining $\beta$-robustness.

For \ALG to be $\beta$-robust, there must exist a constant $b_{\beta}$, independent of $m$, such that, for each $i$ and any prediction $\hat{\sigma}_i$ of $\sigma_i$,
\[ E[\ALG(\hat{\sigma}_i,\sigma_i)] \geq \beta \cdot \OPT(\sigma_i) - b_{\beta} = \beta \cdot 2^i - b_{\beta} \,. \]
Let $j = \max \{ 0, \CEIL{\log (b_{\beta}/\beta)} \}$ and note that, for $i \geq j$, $\beta \cdot 2^i - b_{\beta}$ is non-negative.

Let $n_i$ be the number of intervals from $I_i$ accepted by \ALG and let $L_r$ be the total length of intervals from $\sigma_r$ accepted by \ALG.
Since $\OPT(\sigma_i)-\OPT(\sigma_{i-1}) = 2^{i-1}$, we conclude, by linearity of expectations, that $E[L_r]$ is minimized, if
\[
E[n_i] =
\begin{cases}
  0, & 0 \leq i < j\\
  \beta \cdot \OPT(\sigma_j) - b_{\beta}, & i=j\\
  \beta \cdot 2^{i-1}, & j+1 \leq i \leq r\,.
\end{cases}
\]
With these values of $E[n_i]$, the expected total length of intervals from $I_i$, $j+1 \leq i \leq r$, accepted by \ALG is $\beta \cdot 2^{i-1} \times m'/2^i = \beta \cdot m'/2$.
Thus, by the linearity of expectations, 
\[ E[L_r] \geq \sum_{i=j+1}^r \beta \cdot \frac{m'}{2} = \beta \cdot \frac{m' (r-j)}{2} \,.\]
If the expected number of intervals that \ALG accepts from $\sigma_r$ is more than $\beta \cdot 2^r - b_{\beta}$, then $E[L_r]$ increases by more than one for each additional interval.
Thus, by linearity of expectations, for any prediction $\hat{\sigma}$,
\begin{equation}
\label{eq:expUpper}
E[\ALG(\hat{\sigma},\sigma)] \leq \beta \cdot 2^r - b_{\beta} + \left( m' - \beta \cdot \frac{m'(r-j)}{2} \right) \,.
\end{equation}

Now, let $\hat{\sigma}$ be the prediction consisting of exactly the intervals in $\sigma$.
Then, for \ALG to be $\alpha$-consistent, there must be a constant $b_{\alpha}$ such that
\begin{equation}
  \label{eq:expLower}
  E[\ALG(\hat{\sigma},\sigma)]\geq \alpha \cdot m' - b_{\alpha} \,.
\end{equation}
Combining Inequalities~(\ref{eq:expUpper}) and~(\ref{eq:expLower}), we obtain
\[ \beta \cdot \frac{2^r}{m'}+1-\beta \cdot \frac{r-j}{2} + \frac{b_{\alpha}}{m'} \geq\alpha \,. \]
Since $2^r/m'=1/2$ and $m' \geq m/2$,
this reduces to
\begin{align*}
\alpha & \leq 1 - \beta \cdot \frac{r-j-1}{2} + \frac{2b_{\alpha}}{m} \\
& =  1 - \beta \cdot \frac{\FLOOR{\log m}-1 - \CEIL{\log(b_{\beta}/\beta)}-1}{2} + \frac{2b_{\alpha}}{m}\\
& =  1 - \beta \cdot \frac{\log m - \log(1/\beta) - \log b_{\beta} - 4 }{2} + \frac{2b_{\alpha}}{m}\\
& \leq 1 - \beta \cdot \left( \frac{\log m - \log (m^{1-\eps}) - \log b_{\beta} - 4 }{2} \right) + \frac{2b_{\alpha}}{m}\\
& = 1 - \beta \cdot \left( \frac{ \eps \cdot \log m -c}{2} \right) + \frac{2b_{\alpha}}{m} \,, \text{ where } c = \log b_{\beta} +4 \\
& \leq 1 - \beta \cdot \left( \frac{ \eps \cdot \log m -c}{2} \right) + \beta \cdot 2b_{\alpha} \,, \text{ since } \beta \geq \frac{1}{m}\\
& = 1 - \beta \cdot \left( \frac{\eps \cdot \log m}{2} - O(1) \right)
\end{align*}
Solving for $\beta$, we get
\begin{align*}
  \beta & \leq \left(1-\alpha + \frac{2b_{\alpha}}{m}\right) \cdot \frac{2}{ \eps \cdot \log m - c} \\
  & = \left(1-\alpha + \frac{2b_{\alpha}}{m}\right) \cdot \frac{2 \cdot \left(1 + \frac{c}{\eps \cdot \log m - c}\right)}{ (\eps \cdot \log m - c) \cdot \left(1 + \frac{c}{\eps \cdot \log m - c}\right)} \\
  & = \left(1-\alpha + \frac{2b_{\alpha}}{m}\right) \cdot \frac{2 \cdot \left(1 + \frac{c}{\eps \cdot \log m - c}\right)}{\eps \cdot \log m} \\
  & = \left(1-\alpha + \frac{2b_{\alpha}}{m}\right) \cdot \left( \frac{2}{ \eps \cdot \log m} + \frac{2c}{(\eps \cdot \log m)^2 - c\eps \cdot \log m} \right) \\
  & = \left(1-\alpha\right) \cdot \left( \frac{2}{ \eps \cdot \log m} + O\left( \frac{1}{\log^2 m} \right) \right) + O\left( \frac{1}{m \log m} \right)\\
  & \leq \left(1-\alpha\right) \cdot \frac{2}{ \eps \cdot \log m} + O\left( \frac{1}{\log^2 m} \right)
\end{align*}
\end{proof}

Note that as $\alpha$ approaches~1 (optimal consistency), $\beta$ goes to~$O(1/\log^2 m)$ (worst-case robustness), and as $\beta$ goes to $\frac{2}{\eps \log m}$ (optimal robustness), $\alpha$ goes to $O(1/\log m)$ (worst-case consistency).

\subsubsection{\ROBUSTTRUST}
Next, we present a family of randomized algorithms, \ROBUSTTRUST, which has a parameter $0\leq \alpha \leq 1$ and works as follows. 
With a probability of $\alpha$, \ROBUSTTRUST applies \TRUSTGREEDY.
(Applying \TRUST, instead of \TRUSTGREEDY, gives the same consistency and robustness
results.) With probability 
$1-\alpha$, 
\ROBUSTTRUST ignores the predictions, and applies the Classify-and-Randomly-Select (\CRS) algorithm described in Theorem 13.7 in~\cite{BE98}. \CRS is strictly $\lceil \log m \rceil$-competitive (they use competitive ratios at least one in a version of the problem where requests have a limited time duration). A similar algorithm
was originally proven $O(\log m)$-competitive in~\cite{ABFR94}.

For completeness, we include the \CRS algorithm. 
To avoid the problem of $m$ possibly not being a power of $2$, we define
$\ell=\lceil \log m \rceil$ and $m'=2^{\ell}$. Thus, the algorithm will define
its behavior for a longer line and some sequences that cannot exist.

We define a set of $\ell$ \emph{levels} for the possible requests.
Since $m'$ is a power of two, there is an odd number of edges, so
the middle edge, $e_1$, in the line is well
defined. We define the set $E_1=\{ e_1\}$ and let Level~1 consist of all intervals
containing $e_1$. After Levels~1 through $i$ are defined, we define $E_{i+1}$
and  Level~$i+1$ as follows: After removing all edges in $E_1\cup E_2 \cup
\cdots \cup E_i$ from the line, we are left with $2^{i}$ segments, each consisting
of $2^{\ell-i}$ vertices. The set $E_{i+1}$ consists of the middle edges of these
segments, and Level $i+1$ consists of all intervals, not in any of the Levels
$1$ through $i$, but containing an edge in $E_{i+1}$. Thus, the levels create a partition
of all possible intervals.

The algorithm \CRS initially chooses a level $i$ between $1$ and $\ell$, each with
probability $\frac{1}{\ell}$. It accepts any interval in Level~$i$ that does not
overlap an interval it already has accepted. Any intervals not in Level~$i$
are rejected. 

\begin{theorem}\label{th:robtrust}
For the online interval scheduling problem, \ROBUSTTRUST (\RT) with parameter $\alpha$ 
has consistency at least $\alpha$ and robustness at least $\frac{1-\alpha}{\lceil \log m \rceil}$. 
\end{theorem}
\begin{proof}
	We investigate \ROBUSTTRUST when all predictions are correct (consistency) and when predictions may be incorrect (robustness).
	
	Suppose all predictions are correct. \ROBUSTTRUST applies \TRUSTGREEDY with probability $\alpha$. 
	Since \TRUSTGREEDY is optimal when all predictions are correct, the expected profit of \ROBUSTTRUST is at least $\alpha \cdot \OPT$.
	Therefore, the competitive ratio (consistency) of \ROBUSTTRUST is at least~$\alpha$.

	Suppose some predictions are incorrect.
	If the intervals in Level~$i$
	are the only intervals given, and \CRS chooses that level, then
	\CRS accepts as many intervals as \OPT does,
	since each interval in Level~$i$ contains an edge in $E_i$, and no intervals
	containing more than one edge in $E_i$ exist.
        Since the number of levels
	is $\lceil\log m\rceil$, the expected number of
	intervals that \CRS accepts from any given level of
        \OPT's configuration is at least
	$\frac{1}{\lceil\log m\rceil}$ times the number of intervals \OPT accepted from that level,
	so by the linearity of expectations, this totals
	$\frac{1}{\lceil \log m \rceil}\OPT$. \CRS is
	chosen with probability $1-\alpha$, so the robustness is
	at least $\frac{1-\alpha}{\lceil \log m \rceil}$.
\end{proof}

\begin{center}
	\begin{table}[!b]\label{table:stat}
		\centering
\resizebox{\textwidth}{!}{
		\begin{tabular}{|c|c|c|c|c|c|}
			\hline
			name & input size ($N$) & no. timesteps ($m$) & max. length & avg. length \\
			\hline \hline 
			LLNL-uBGL-2006-2 & 13,225 & 16,671,553 &  14,403 & 1,933.92\\
			\hline
			NASA-iPSC-1993-3.1 & 18,066 & 7,947,562 &  62,643 & 772.21 \\
			\hline 
			CTC-SP2-1996-3.1 & 77,205 & 8,986,769 & 71,998 & 11,279.61 \\ 
			\hline 
			SDSC-DS-2004-2.1 & 84,893 & 31,629,689 & 6,589,808 & 7,579.36 \\ 
			\hline
		\end{tabular}
}
		\caption{Details on the benchmarks from~\cite{ChapinCFJLSST99} used in our experiments.\label{tablek}}
	\end{table}
\end{center}

\subsection{Experimental Results}
We present an experimental evaluation of \TRUST and \TRUSTGREEDY for interval scheduling in comparison with the \GREEDY algorithm, which serves as a baseline online algorithm, and \OPT, which serves as the performance upper bound. Our code and results are available at~\cite{OurCode}.

We evaluate our algorithms using real-world scheduling data for parallel machines~\cite{ChapinCFJLSST99}. Each benchmark from~\cite{ChapinCFJLSST99} specifies the start and finish times of tasks as scheduled on parallel machines with several processors.  
We use these tasks to generate inputs to the interval scheduling problem; Table~\ref{tablek} details the interval scheduling inputs we generated from benchmarks of~\cite{ChapinCFJLSST99}. 
For each benchmark with $N$ tasks, we create an instance $I$ of an interval scheduling problem by randomly selecting $n = \lfloor N/2 \rfloor$ tasks from the benchmark and randomly permuting them.
This sequence serves as the input to all algorithms. To generate the prediction, we consider $1000$ equally distanced values of $d \in [0,n]$. For each value of $d$, we initiate the prediction set $\IPRED$ with the set of intervals in $I$, remove $|\FN|=d$ randomly selected intervals from $\IPRED$ and add to it $|\FP|=d$ randomly selected intervals from the remaining $N-n$ tasks in the benchmark. The resulting set $\IPRED$ is given to \TRUST and \TRUSTGREEDY as prediction $\IPRED$. For each value of $d$, we compute the normalized error $\gamma(\IPRED,I) = \frac{\OPT(\FN\cup \FP)}{\OPT(I)}$, and report the profit of \TRUST and \TRUSTGREEDY as a function of $\gamma$.

\begin{figure}[!b]
	\centering
	\begin{subfigure}[b]{0.495\textwidth}
		\centering
		\includegraphics[page=2,trim = 3.3cm 17.55cm 3.2cm 3.25cm,clip,scale=.55]{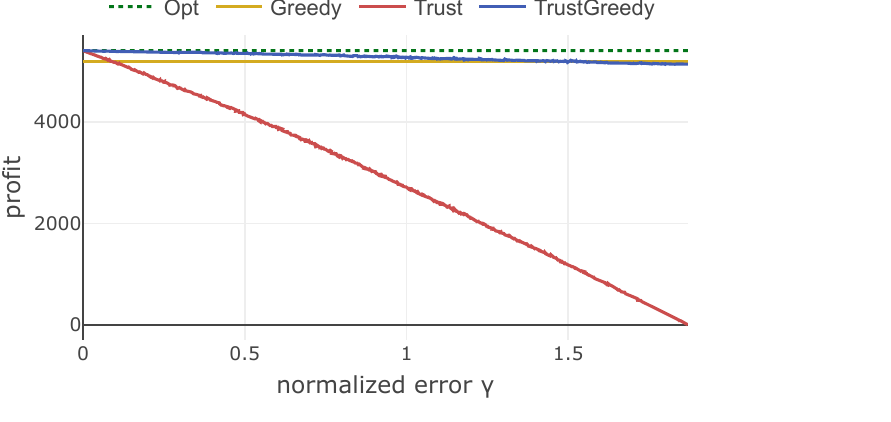}
		\caption{LLNL-uBGL-2006-2}
		\label{fig:LLNLEqual}
	\end{subfigure}
	\begin{subfigure}[b]{0.495\textwidth}	
		\centering
		\includegraphics[page=6,trim = 3.3cm 8.7cm 3.2cm 12.1cm,clip,scale=.55]{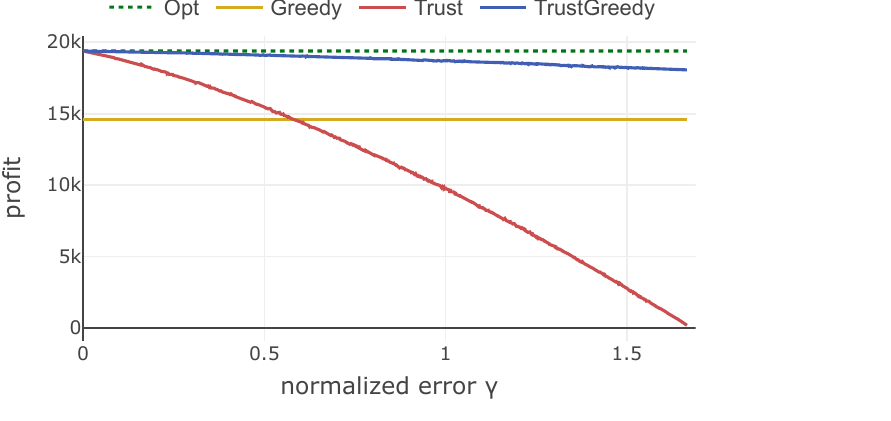}
		\caption{ SDSC-DS-2004-2.1-cln}
		\label{fig:SDSCEqual}
	\end{subfigure}	
	\begin{subfigure}[b]{0.495\textwidth}
		\centering
			\includegraphics[page=4,trim = 3.3cm 8.7cm 3.2cm 11.1cm,clip,scale=.55]{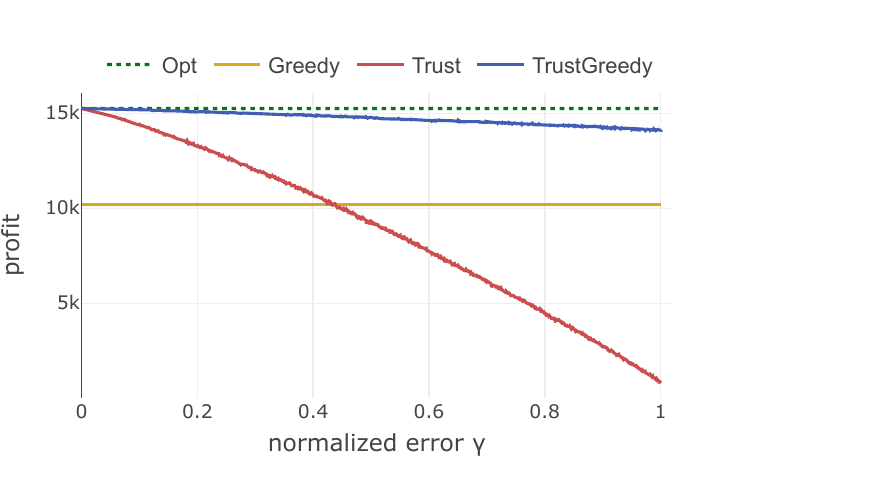}
		\caption{CTC-SP2-1996-3.1-cln}
		\label{fig:CTCEqual}
	\end{subfigure}
	\hfill
	\begin{subfigure}[b]{0.495\textwidth}
		\centering
		\includegraphics[page=3,trim = 3.3cm 8.7cm 3.2cm 11.1cm,clip,scale=.55]{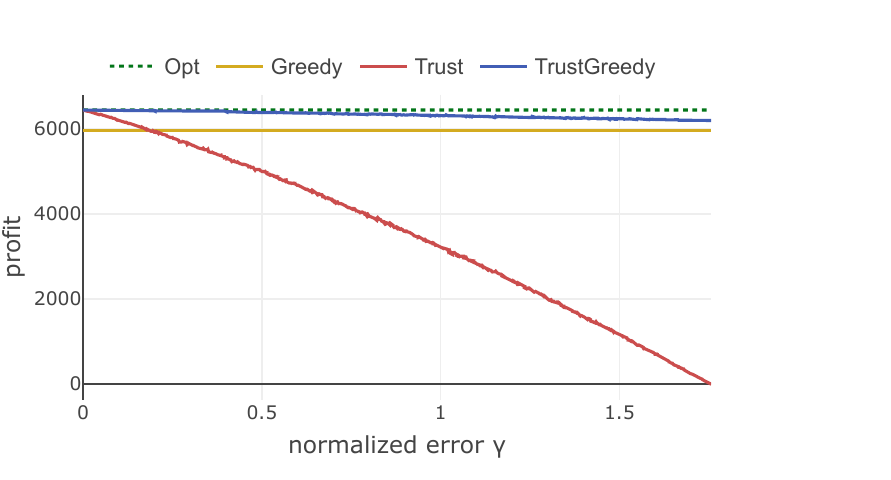}
		\caption{ NASA-iPSC-1993}
		\label{fig:NASAEqual}
	\end{subfigure}
	
	\caption{Profit as a function of normalized error value \vspace*{1mm}}
	
	\label{fig:mainexp}
\end{figure}

\begin{figure}[!b]
	\centering
	\begin{subfigure}[b]{0.495\textwidth}
		\centering
		\includegraphics[page=1,trim = 3.3cm 17.55cm 3.2cm 3.15cm,clip,scale=.55]{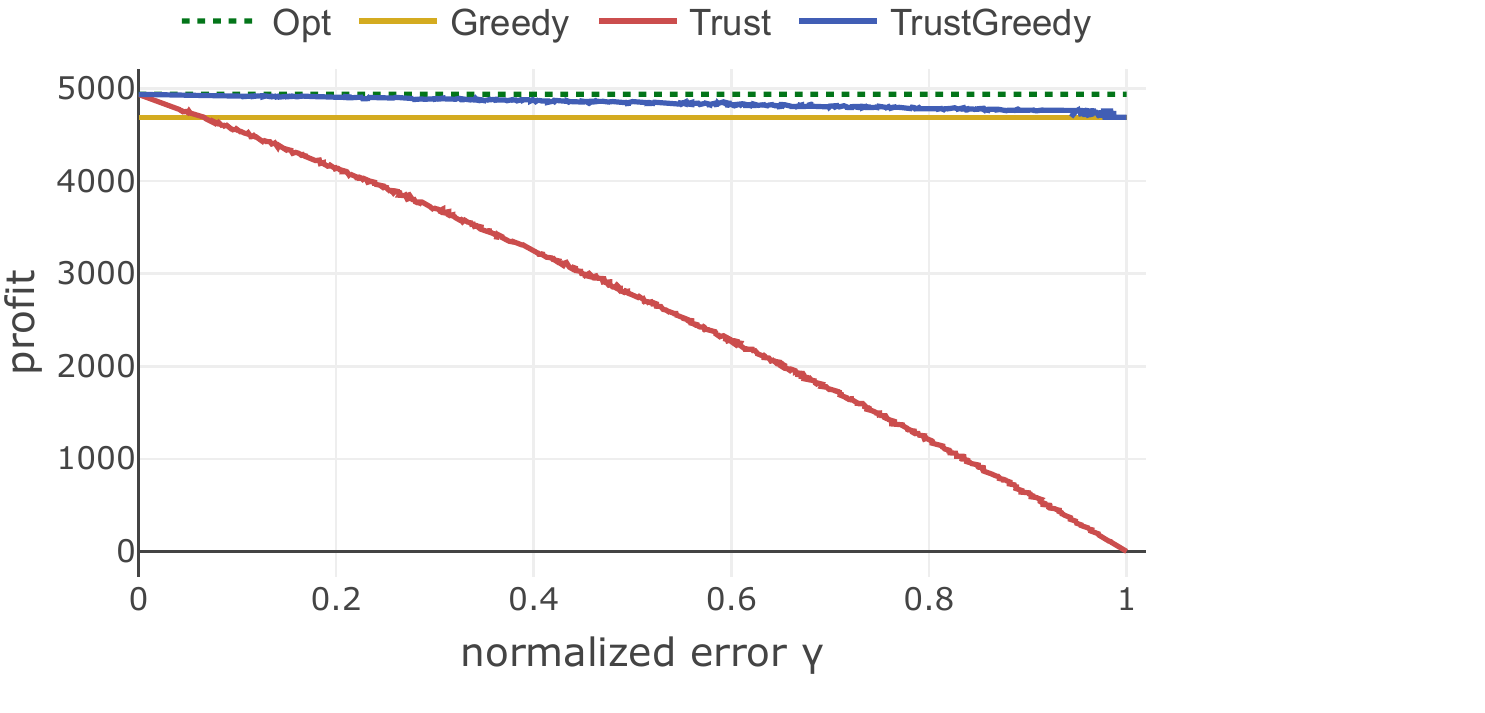}
		\caption{LLNL-uBGL-2006-2, no \FP}
		\label{fig:LLNLNOFP}
	\end{subfigure}
	\hfill
	\begin{subfigure}[b]{0.495\textwidth}
		\centering
		\includegraphics[page=1,trim = 3.3cm 8.7cm 3.2cm 12.1cm,clip,scale=.55]{allplots}		
		\caption{LLNL-uBGL-2006-2, no \FN}
		\label{fig:LLNLNoFN}
	\end{subfigure} \vspace*{3mm} \\
	\begin{subfigure}[b]{0.495\textwidth}
	\centering
		\includegraphics[page=5,trim = 3.3cm 7.6cm 3.2cm 13.2cm,clip,scale=.55]{allplots}
		\caption{ SDSC-DS-2004-2.1-cln, no \FP}
		\label{fig:SDSCNoFP}
	\end{subfigure}
	\hfill	
	\begin{subfigure}[b]{0.495\textwidth}
		\centering
		\includegraphics[page=6,trim = 3.3cm 17.55cm 3.2cm 3.25cm,clip,scale=.55]{allplots}%
		\caption{ SDSC-DS-2004-2.1-cln, no \FN}
		\label{fig:SDSCNoFN}
	\end{subfigure}	
	\caption{Profit as a function of normalized error value in the absence of false positives (a), (c) and false negatives (b), (d).}
	\label{fig:three graphs}
\end{figure}

Figure~\ref{fig:mainexp} shows the results for two representative benchmarks from~\cite{ChapinCFJLSST99}, namely, LLNL (the workload of the BlueGene/L system installed at Lawrence Livermore National Lab), SDSC (the workload log from San Diego Supercomputer Center),  NASA-iPSC (scheduling log from Numerical Aerodynamic Simulation -NAS- Systems Division at NASA Ames Research Center) and CTC-SP2 (Cornell Theory Center IBM SP2 log). These four benchmarks are selected to represent a variety of input sizes and interval lengths. The results are aligned with our theoretical findings: \TRUST quickly becomes worse than \GREEDY as the error value increases, while \TRUSTGREEDY degrades gently as a function of the prediction error. In particular, \TRUSTGREEDY is better than \GREEDY for almost all error values. We note that \GREEDY
performs better when there is less overlap between the input intervals, which is the case in LLNL compared to SDSC. In an extreme case, when no two intervals overlap, \GREEDY is trivially optimal. Nevertheless, even for LLNL, \TRUSTGREEDY is not much worse than \GREEDY for extreme values of error: the profit for the largest normalized error of $\gamma = 1.87$ was 5149 and 5198 for \TRUSTGREEDY and \GREEDY, respectively. Note that for SDSC, where there are more overlaps between intervals, \TRUSTGREEDY is strictly better than \GREEDY, even for the largest error values. It is worth noting that, in an extreme case, where $\FP=\FN=n$, the predictions contain a completely different set from the input sequence. In that case, $|\FP\cup\FN| = 2n$, and
$\gamma = \frac{\OPT(\FP\cup\FN)}{\OPT(I)}$ takes values in $[1.5,2]$. 

We also experiment in a setting where false positives and negatives contribute differently to the error set. We generate the input sequences in the same way as in the previous experiments. To generate the prediction set $\IPRED$, we consider $1000$ equally-distanced values of $d$ in the range $[0,n]$ as before. We first consider a setting in which all error is due to false negatives; for that, we generate $\IPRED$ by removing $d$ randomly selected intervals from $I$. In other words, $\IPRED$ is a subset of the intervals in $I$. 
Figures~\ref{fig:LLNLNOFP} and~\ref{fig:SDSCNoFP}  illustrate the profit of \TRUST and \TRUSTGREEDY in this case. We note that \TRUSTGREEDY is strictly better than both \TRUST and \GREEDY. In an extreme case, when $d=n$, \IPRED becomes empty and \TRUSTGREEDY becomes \GREEDY; in other words, \GREEDY is the same algorithm as \TRUSTGREEDY with the empty predictions set $\IPRED$.

We also consider a setting in which there are no false negatives. For that, we generate $\IPRED$ by adding $d$ intervals to \IPRED. In other words, $\IPRED$ will be a superset of intervals in $I$. Figures~\ref{fig:LLNLNOFP} and~\ref{fig:SDSCNoFP}  illustrate the profit of \TRUST and \TRUSTGREEDY in this case. In this case, the profit of \TRUST and \TRUSTGREEDY is similar to the setting where both false positives and negatives contributed to the error set. In particular, \TRUST quickly becomes worse than \GREEDY as the error increases, while \TRUSTGREEDY degrades gently as a function of the prediction error.

\section{Related Problems: Matching and Independent Set}

In~\cite{GVV97}, the authors
observe that finding
disjoint paths on stars is equivalent to finding maximal matchings on
general graphs, where each request in the input to the disjoint path allocation problem bijects to an edge in the input graph for the matching problem. Therefore, we can extend the results of Section~\ref{sect:dpa} to the following \emph{online matching problem}. The input is a graph $G=(V, E)$, where $V$ is known, and edges in $E$ appear in an online manner; upon arrival of an edge, it must be added to the matching or rejected. The prediction is a set $\EPRED$ that specifies edges in $E$. As before, we use $\FP$ and $\FN$ to indicate the set of false positives and false negatives and define \[ \ernormE = \frac{\OPT(\FP \cup\FN)}{\OPT(E)}\,, \] where $\OPT(E)$ indicates the size of an optimal matching for graph $G=(V,E)$. 

The correspondence between the two problems is as follows: Consider a set of
requests on a star. Each such
request is a pair of vertices. 
We can assume no pair contains the star's center since all such requests should be accepted if they can be.
For the matching problem, the pairs of vertices from the disjoint paths
problem on the star can be the edges in the graph. A feasible solution
to the disjoint paths problem corresponds to a matching and vice versa.
One can similarly consider an instance of a matching problem, and the
endpoints of the edges can be the non-center vertices of the star in
the disjoint paths problem.

\newsavebox\mysubpic
\sbox{\mysubpic}{%
	\begin{tikzpicture}[scale=.5]
		\draw[rotate=0,dashed] (0,0) ellipse (1.9cm and .75cm);
	\end{tikzpicture}
}
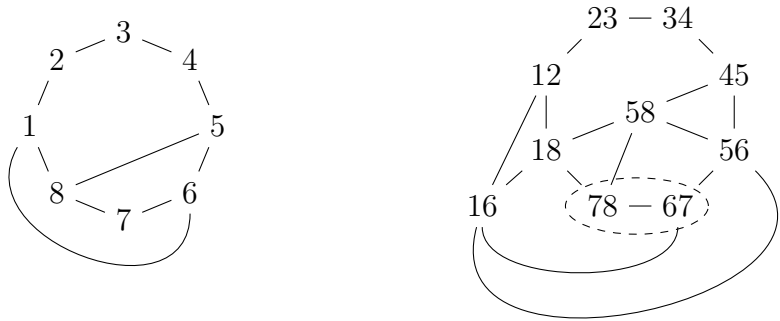
\begin{figure}[!t]

	\begin{center}
	  \begin{subfigure}[t]{0.46\textwidth}
            \centering
		{\begin{tikzpicture}[scale=.5]
				\foreach \X [count=\Y] in {1,2,3,4,5,6,7,8}
				{\node  (cn\Y) at ({-(\Y+3)*360/8}:2.5) {$\X$}; }
				\foreach \Y [remember=\Y as \LastY (initially 8)]in {1,...,8}
				{\draw (cn\LastY) -- (cn\Y);}
				\draw (cn8) -- (cn5);
				\draw (cn1) to [out=240,in=270,looseness=1.5] (cn6);
                                \node (fake) at (0,-7) {\mbox{}};
		\end{tikzpicture}}
		\label{subfig:matching}
	\caption{The graph, $G$, for matching corresponding to the star graph for disjoint paths with requests $(1,2)$, $(2,3)$, $(3,4)$, $(4,5)$, $(5,6)$, $(6,7)$, $(7,8)$, $(1,8)$, $(1,6)$, and $(5,8)$.}
	\end{subfigure}\hfill
	  \begin{subfigure}[t]{0.46\textwidth}
            \centering
		\begin{tikzpicture}[scale=.5]
			\node (n1) at (-2.5,1) {12};
			\node (n2) at (-1,2.5) {23};
			\node (n3) at (1,2.5) {34};
			\node (n4) at (2.5,1) {45};
			\node (n5) at (2.5,-1) {56};
			\node (n6) at (1,-2.5) {67};
			\node (n7) at (-1,-2.5) {78};
			\node (n8) at (-2.5,-1) {18};
			\node (n9) at (-4.2,-2.5) {16};
			\node (n10) at (0,0) {58};
			\draw (n9) -- (n1) -- (n2) -- (n3) -- (n4) -- (n5) -- (n6) -- (n7) -- (n8) -- (n9);
			\draw (n10) -- (n4);
			\draw (n10) -- (n5);
			\draw (n10) -- (n7);
			\draw (n10) -- (n8);
			\draw (n1) -- (n8);
			\draw (n9) to [out=270,in=270,looseness=.8] (n6);
			\draw (n9) to [out=255,in=315,looseness=1.6] (n5);
			\node at (0.05,-2.5) {\usebox{\mysubpic}};
		\end{tikzpicture}
		\label{subfig:independent set}
	\caption{The line graph, $G'$, corresponding to~$G$, where the two digits in a vertex name in~$G'$ indicate the edge (given by its two endpoints) from~$G$ that the vertex corresponds to. The dashed ellipse indicates the contraction of the two vertices showing that $G'$ contains $K_{2,3}$ as a minor.}
	\end{subfigure}
	\end{center}
	\caption{Graphs for matching and independent set}
	\label{figure-matchings-independent-set}
\end{figure}
Using this correspondence between disjoint paths on a star and matchings in
general graphs, for the star $S_8$ with non-center vertices $1,2,\ldots,8$ and requests $(1,2),(2,3),(3,4),(4,5),(5,6),(6,7),(7,8),(1,8),(1,6),(5,8)$, we get the graph $G=(V,E)$ for
matching, where
\begin{align*}
  & V=\SET{1,2,3,4,5,6,7,8} \text{  and }\\
  & E=\SET{(1,2),(2,3),(3,4),(4,5),(5,6),(6,7),(7,8),(1,8),(1,6),(5,8)}.
\end{align*}
See also Figure~\ref{figure-matchings-independent-set}.
Note that the edges in this graph correspond to the requests that are
used in the proof of Theorem~\ref{thm:star}.  The proof can be
simulated in this new setting so that the number of requests accepted
in the different cases in Theorem~\ref{thm:star} is the same as the
number of edges in the matchings found in the corresponding subgraphs
of~$G$.  Thus, the same result holds for matchings in any graph class
containing this graph.

All edges have one even-numbered endpoint and one odd, so this includes
the bipartite graph class.
It is also planar but not an interval or chordal graph.

Given the correspondence between interval scheduling and the matching problem, the following is immediate from Theorems~\ref{pr:trustupper} and~\ref{thm:star}. 

\begin{corollary}
\label{cor:matching}
  There is a strictly $(1-2\gamma)$-competitive algorithm, \TRUST, for
  the online matching problem under the edge-arrival model.
        For $0 \leq \gamma \leq 1/4$, this is optimal among deterministic algorithms, even on bipartite graphs as well as planar graphs.
\end{corollary}

Using the correspondence between matchings in a graph, $G$, and
an independent set in the line graph of~$G$,
we can get the same result for the {\em independent set problem} on line graphs.
The line graph of a graph, $G$, has a vertex for each edge in~$G$
and an edge between two vertices if the corresponding edges in~$G$
share a vertex.

The line graph $G'=(V',E')$ of the graph above used for matching is defined by
\[V'=\SET{12, 23, 34, 45, 56, 67, 78, 18, 16, 58},\]
where, for brevity, we use the notation $12$ to denote the vertex
corresponding to the edge $(1,2)$ from~$G$.
The set of edges is then
\[\begin{array}{r@{}l}
	E'=\{ & (12,23), (23,34), (34,45), (45,56), (56,67), (67,78), (78,18), (18,16), \\
	& (16,12), (58,18), (58,78), (58,56), (58,45), (12,18), (16,67), (16,56)
	\}\,.
\end{array}\]
Requests from the proof in Theorem~\ref{thm:star} correspond to vertices here.
See also Figure~\ref{figure-matchings-independent-set}.

We note that the graph $G'$ is planar, but not outerplanar, since,
contracting $67$ and $78$ into one vertex, $67\textrm{-}78$, the sets
$\SET{16,58}$ and $\SET{18,56,67\textrm{-}78}$ form a $K_{2,3}$ minor,
which is a so-called forbidden subgraph for outerplanarity~\cite{CH67,H69}.
Also, it is not chordal.
However, the lower bound of $1-\gamma$ for deterministic interval scheduling algorithms (Theorem~\ref{thm:generallower}) clearly holds for independent sets in interval
graphs, too, by considering the interval graph corresponding to a set
of intervals on the line.  

Summing up, we obtain:
\begin{corollary}
  \label{cor:independentset}
  For the online independent set problem under the vertex-arrival model, the following hold.
\begin{itemize}
	\item   On line graphs, there is a strictly $(1-2\gamma)$-competitive algorithm, \TRUST. 
	\item   For $0 \leq \gamma \leq 1/4$, no deterministic algorithms can be better than $(1-2\gamma)$-competitive, on line graphs as well as planar graphs.
	\item   On interval graphs, no deterministic algorithm is better than $(1-\gamma)$-competitive.
\end{itemize}

\end{corollary}

\bibliography{refs}
\bibliographystyle{plain}

\end{document}

%% file: figs.tex
\newcommand{\tgExample}{

	\tikzset{every picture/.style={line width=0.75pt}} 
	
	\begin{tikzpicture}[x=0.75pt,y=0.75pt,yscale=-1,xscale=1]
		
		\draw [color={rgb, 255:red, 208; green, 2; blue, 27 }  ,draw opacity=1 ]   (33.33,109.91) -- (61.47,109.91) ;
		\draw [color={rgb, 255:red, 0; green, 0; blue, 0 }  ,draw opacity=1 ]   (87,99.91) -- (230.8,99.47) ;
		\draw [color={rgb, 255:red, 0; green, 0; blue, 0 }  ,draw opacity=1 ]   (54.33,124.58) -- (100.13,124.8) ;
		\draw [color={rgb, 255:red, 0; green, 0; blue, 0 }  ,draw opacity=1 ]   (108,124.58) -- (136.13,124.58) ;
		\draw    (103.07,116) -- (191.13,115.8) ;
		\draw [color={rgb, 255:red, 74; green, 144; blue, 226 }  ,draw opacity=1 ]   (166.33,124.58) -- (212.8,124.13) ;
		\draw    (204.33,115.24) -- (232.47,115.24) ;
		
		\draw (80.45,131.77) node  [font=\scriptsize]  {$S_{2}$};
		\draw (146.45,108.77) node  [font=\scriptsize]  {$S_{4}$};
		\draw (180.45,131.1) node  [font=\scriptsize]  {$S_{6}$};
		\draw (216.12,108.44) node  [font=\scriptsize]  {$S_{7}$};
		\draw (44.67,103.24) node  [font=\scriptsize]  {$S_{1}$};
		\draw (149.78,92.44) node  [font=\scriptsize]  {$S_{3}$};
		\draw (121.45,132.44) node  [font=\scriptsize]  {$S_{5}$};

	\end{tikzpicture}
}

\newcommand{\lipCompExample}{

	\tikzset{every picture/.style={line width=0.75pt}} 
	
	\begin{tikzpicture}[x=0.75pt,y=0.75pt,yscale=-1,xscale=1]
		
		\draw [color={rgb, 255:red, 111; green, 2; blue, 27 }  ,draw opacity=1 ]   (64.67,133.24) -- (92.8,133.24) ;
		\draw [color={rgb, 255:red, 111; green, 2; blue, 27 }  ,draw opacity=1 ]   (109.33,133.24) -- (137.47,133.24) ;
		\draw    (151.33,133.24) -- (179.47,133.24) ;
		\draw    (196,133.24) -- (224.13,133.24) ;
		\draw    (316.33,132.91) -- (344.47,132.91) ;
		\draw    (361,132.91) -- (389.13,132.91) ;
		\draw    (85.67,147.91) -- (113.8,147.91) ;
		\draw    (130.33,147.91) -- (158.47,147.91) ;
		\draw    (172.33,147.91) -- (200.47,147.91) ;
		\draw    (217,147.91) -- (245.13,147.91) ;
		\draw    (337.33,147.58) -- (365.47,147.58) ;
		
		\draw (81.78,121.1) node  [font=\small,color={rgb, 255:red, 111; green, 2; blue, 27 }  ,opacity=1 ]  {$A_{1}$};
		\draw (122.78,121.44) node  [font=\small,color={rgb, 255:red, 111; green, 2; blue, 27 }  ,opacity=1 ]  {$A_{2}$};
		\draw (168.45,121.1) node  [font=\small]  {$A_{3}$};
		\draw (209.45,121.44) node  [font=\small]  {$A_{4}$};
		\draw (333.45,120.77) node  [font=\small]  {$A_{k-1}$};
		\draw (374.45,121.1) node  [font=\small]  {$A_{k}$};
		\draw (102.78,159.77) node  [font=\small]  {$B_{1}$};
		\draw (143.78,160.1) node  [font=\small]  {$B_{2}$};
		\draw (189.45,159.77) node  [font=\small]  {$B_{3}$};
		\draw (230.45,160.1) node  [font=\small]  {$B_{4}$};
		\draw (354.45,159.44) node  [font=\small]  {$B_{k-1}$};
		\draw (273.78,138.77) node  [font=\large]  {$\dotsc $};

	\end{tikzpicture}
}
\newcommand{\figStarTwoD}{

	\tikzset{every picture/.style={line width=0.75pt}} 
	
	\begin{tikzpicture}[x=0.75pt,y=0.75pt,yscale=-1,xscale=1]
		
		\draw [color={rgb, 255:red, 65; green, 117; blue, 5 }  ,draw opacity=0.44 ][fill={rgb, 255:red, 255; green, 255; blue, 255 }  ,fill opacity=1 ][line width=2.25]    (448,398.78) -- (478.8,382.69) ;
		\draw [color={rgb, 255:red, 65; green, 117; blue, 5 }  ,draw opacity=0.44 ][fill={rgb, 255:red, 177; green, 77; blue, 55 }  ,fill opacity=0.55 ][line width=2.25]    (436.66,383.82) -- (478.8,382.69) ;
		\draw [fill={rgb, 255:red, 255; green, 255; blue, 255 }  ,fill opacity=1 ] [dash pattern={on 0.84pt off 2.51pt}]  (442.31,365.31) -- (479.4,382.72) ;
		\draw [fill={rgb, 255:red, 255; green, 255; blue, 255 }  ,fill opacity=1 ] [dash pattern={on 0.84pt off 2.51pt}]  (479.4,382.72) -- (500.48,406.67) ;
		\draw [color={rgb, 255:red, 74; green, 74; blue, 74 }  ,draw opacity=1 ][fill={rgb, 255:red, 255; green, 255; blue, 255 }  ,fill opacity=1 ][line width=0.75]  [dash pattern={on 0.84pt off 2.51pt}]  (479.4,382.72) -- (518.33,397.39) ;
		\draw [color={rgb, 255:red, 74; green, 74; blue, 74 }  ,draw opacity=1 ][fill={rgb, 255:red, 255; green, 255; blue, 255 }  ,fill opacity=1 ][line width=0.75]  [dash pattern={on 0.84pt off 2.51pt}]  (479.4,382.72) -- (512.55,380.82) -- (523.89,380) ;
		\draw [color={rgb, 255:red, 74; green, 74; blue, 74 }  ,draw opacity=1 ][fill={rgb, 255:red, 255; green, 255; blue, 255 }  ,fill opacity=1 ][line width=0.75]  [dash pattern={on 0.84pt off 2.51pt}]  (479.4,382.72) -- (519.43,365.85) ;
		\draw [color={rgb, 255:red, 74; green, 74; blue, 74 }  ,draw opacity=1 ][fill={rgb, 255:red, 255; green, 255; blue, 255 }  ,fill opacity=1 ][line width=0.75]  [dash pattern={on 0.84pt off 2.51pt}]  (497.03,361.45) -- (479.4,382.72) ;
		\draw  [fill={rgb, 255:red, 255; green, 255; blue, 255 }  ,fill opacity=1 ] (495.87,357.69) .. controls (495.87,355.01) and (498.19,352.83) .. (501.06,352.83) .. controls (503.93,352.83) and (506.26,355.01) .. (506.26,357.69) .. controls (506.26,360.37) and (503.93,362.54) .. (501.06,362.54) .. controls (498.19,362.54) and (495.87,360.37) .. (495.87,357.69) -- cycle ;
		\draw  [fill={rgb, 255:red, 255; green, 255; blue, 255 }  ,fill opacity=1 ] (518.83,363.67) .. controls (518.83,360.99) and (521.16,358.82) .. (524.03,358.82) .. controls (526.9,358.82) and (529.23,360.99) .. (529.23,363.67) .. controls (529.23,366.36) and (526.9,368.53) .. (524.03,368.53) .. controls (521.16,368.53) and (518.83,366.36) .. (518.83,363.67) -- cycle ;
		\draw  [fill={rgb, 255:red, 255; green, 255; blue, 255 }  ,fill opacity=1 ] (523.89,380) .. controls (523.89,377.32) and (526.22,375.14) .. (529.09,375.14) .. controls (531.96,375.14) and (534.29,377.32) .. (534.29,380) .. controls (534.29,382.68) and (531.96,384.86) .. (529.09,384.86) .. controls (526.22,384.86) and (523.89,382.68) .. (523.89,380) -- cycle ;
		\draw  [color={rgb, 255:red, 74; green, 74; blue, 74 }  ,draw opacity=1 ][fill={rgb, 255:red, 255; green, 255; blue, 255 }  ,fill opacity=1 ][line width=0.75]  (517.67,399.59) .. controls (517.67,396.91) and (519.99,394.74) .. (522.86,394.74) .. controls (525.74,394.74) and (528.06,396.91) .. (528.06,399.59) .. controls (528.06,402.27) and (525.74,404.45) .. (522.86,404.45) .. controls (519.99,404.45) and (517.67,402.27) .. (517.67,399.59) -- cycle ;
		\draw  [fill={rgb, 255:red, 255; green, 255; blue, 255 }  ,fill opacity=1 ] (495.28,406.67) .. controls (495.28,403.98) and (497.61,401.81) .. (500.48,401.81) .. controls (503.35,401.81) and (505.68,403.98) .. (505.68,406.67) .. controls (505.68,409.35) and (503.35,411.52) .. (500.48,411.52) .. controls (497.61,411.52) and (495.28,409.35) .. (495.28,406.67) -- cycle ;
		\draw  [fill={rgb, 255:red, 255; green, 255; blue, 255 }  ,fill opacity=1 ] (437.11,365.31) .. controls (437.11,362.63) and (439.44,360.45) .. (442.31,360.45) .. controls (445.18,360.45) and (447.51,362.63) .. (447.51,365.31) .. controls (447.51,367.99) and (445.18,370.16) .. (442.31,370.16) .. controls (439.44,370.16) and (437.11,367.99) .. (437.11,365.31) -- cycle ;
		\draw  [fill={rgb, 255:red, 255; green, 255; blue, 255 }  ,fill opacity=1 ] (474.2,382.72) .. controls (474.2,380.04) and (476.53,377.87) .. (479.4,377.87) .. controls (482.27,377.87) and (484.6,380.04) .. (484.6,382.72) .. controls (484.6,385.4) and (482.27,387.58) .. (479.4,387.58) .. controls (476.53,387.58) and (474.2,385.4) .. (474.2,382.72) -- cycle ;
		\draw  [fill={rgb, 255:red, 65; green, 117; blue, 5 }  ,fill opacity=0.44 ] (426.26,383.82) .. controls (426.26,381.14) and (428.59,378.97) .. (431.46,378.97) .. controls (434.33,378.97) and (436.66,381.14) .. (436.66,383.82) .. controls (436.66,386.5) and (434.33,388.68) .. (431.46,388.68) .. controls (428.59,388.68) and (426.26,386.5) .. (426.26,383.82) -- cycle ;
		\draw  [fill={rgb, 255:red, 65; green, 117; blue, 5 }  ,fill opacity=0.44 ] (438.85,401.19) .. controls (438.85,398.51) and (441.17,396.34) .. (444.04,396.34) .. controls (446.91,396.34) and (449.24,398.51) .. (449.24,401.19) .. controls (449.24,403.87) and (446.91,406.05) .. (444.04,406.05) .. controls (441.17,406.05) and (438.85,403.87) .. (438.85,401.19) -- cycle ;
		\draw [fill={rgb, 255:red, 255; green, 255; blue, 255 }  ,fill opacity=1 ] [dash pattern={on 0.84pt off 2.51pt}]  (586.31,365.7) -- (623.4,383.11) ;
		\draw [fill={rgb, 255:red, 255; green, 255; blue, 255 }  ,fill opacity=1 ] [dash pattern={on 0.84pt off 2.51pt}]  (623.4,383.11) -- (644.48,407.06) ;
		\draw [color={rgb, 255:red, 122; green, 144; blue, 199 }  ,draw opacity=0.55 ][fill={rgb, 255:red, 255; green, 255; blue, 255 }  ,fill opacity=1 ][line width=2.25]    (623.4,383.11) -- (662.33,397.78) ;
		\draw [color={rgb, 255:red, 122; green, 144; blue, 199 }  ,draw opacity=0.66 ][fill={rgb, 255:red, 255; green, 255; blue, 255 }  ,fill opacity=1 ][line width=2.25]    (623.4,383.11) -- (656.55,381.21) -- (667.89,380.39) ;
		\draw [color={rgb, 255:red, 199; green, 177; blue, 144 }  ,draw opacity=1 ][fill={rgb, 255:red, 255; green, 255; blue, 255 }  ,fill opacity=1 ][line width=2.25]    (623.4,383.11) -- (663.43,366.24) ;
		\draw [color={rgb, 255:red, 199; green, 177; blue, 144 }  ,draw opacity=1 ][fill={rgb, 255:red, 255; green, 255; blue, 255 }  ,fill opacity=1 ][line width=2.25]    (641.03,361.84) -- (623.4,383.11) ;
		\draw [color={rgb, 255:red, 0; green, 0; blue, 0 }  ,draw opacity=0.44 ][fill={rgb, 255:red, 255; green, 255; blue, 255 }  ,fill opacity=1 ][line width=2.25]    (592.76,399.04) -- (625.4,382.11) ;
		\draw [color={rgb, 255:red, 0; green, 0; blue, 0 }  ,draw opacity=0.44 ][fill={rgb, 255:red, 255; green, 255; blue, 255 }  ,fill opacity=1 ][line width=2.25]    (581.26,384.24) -- (606.63,383.73) -- (623.4,383.11) ;
		\draw  [fill={rgb, 255:red, 199; green, 177; blue, 144 }  ,fill opacity=0.5 ] (639.87,358.08) .. controls (639.87,355.4) and (642.19,353.22) .. (645.06,353.22) .. controls (647.93,353.22) and (650.26,355.4) .. (650.26,358.08) .. controls (650.26,360.76) and (647.93,362.93) .. (645.06,362.93) .. controls (642.19,362.93) and (639.87,360.76) .. (639.87,358.08) -- cycle ;
		\draw  [fill={rgb, 255:red, 199; green, 177; blue, 144 }  ,fill opacity=0.5 ] (662.83,364.07) .. controls (662.83,361.38) and (665.16,359.21) .. (668.03,359.21) .. controls (670.9,359.21) and (673.23,361.38) .. (673.23,364.07) .. controls (673.23,366.75) and (670.9,368.92) .. (668.03,368.92) .. controls (665.16,368.92) and (662.83,366.75) .. (662.83,364.07) -- cycle ;
		\draw  [fill={rgb, 255:red, 111; green, 144; blue, 199 }  ,fill opacity=0.55 ] (667.89,380.39) .. controls (667.89,377.71) and (670.22,375.54) .. (673.09,375.54) .. controls (675.96,375.54) and (678.29,377.71) .. (678.29,380.39) .. controls (678.29,383.07) and (675.96,385.25) .. (673.09,385.25) .. controls (670.22,385.25) and (667.89,383.07) .. (667.89,380.39) -- cycle ;
		\draw  [fill={rgb, 255:red, 111; green, 144; blue, 199 }  ,fill opacity=0.44 ] (661.67,399.98) .. controls (661.67,397.3) and (663.99,395.13) .. (666.86,395.13) .. controls (669.74,395.13) and (672.06,397.3) .. (672.06,399.98) .. controls (672.06,402.66) and (669.74,404.84) .. (666.86,404.84) .. controls (663.99,404.84) and (661.67,402.66) .. (661.67,399.98) -- cycle ;
		\draw  [fill={rgb, 255:red, 255; green, 255; blue, 255 }  ,fill opacity=1 ] (639.28,407.06) .. controls (639.28,404.38) and (641.61,402.2) .. (644.48,402.2) .. controls (647.35,402.2) and (649.68,404.38) .. (649.68,407.06) .. controls (649.68,409.74) and (647.35,411.91) .. (644.48,411.91) .. controls (641.61,411.91) and (639.28,409.74) .. (639.28,407.06) -- cycle ;
		\draw  [fill={rgb, 255:red, 177; green, 177; blue, 177 }  ,fill opacity=0.7 ] (570.86,384.24) .. controls (570.86,381.56) and (573.19,379.39) .. (576.06,379.39) .. controls (578.93,379.39) and (581.26,381.56) .. (581.26,384.24) .. controls (581.26,386.93) and (578.93,389.1) .. (576.06,389.1) .. controls (573.19,389.1) and (570.86,386.93) .. (570.86,384.24) -- cycle ;
		\draw  [color={rgb, 255:red, 0; green, 0; blue, 0 }  ,draw opacity=1 ][fill={rgb, 255:red, 177; green, 177; blue, 177 }  ,fill opacity=0.7 ] (583.45,401.62) .. controls (583.45,398.93) and (585.77,396.76) .. (588.64,396.76) .. controls (591.51,396.76) and (593.84,398.93) .. (593.84,401.62) .. controls (593.84,404.3) and (591.51,406.47) .. (588.64,406.47) .. controls (585.77,406.47) and (583.45,404.3) .. (583.45,401.62) -- cycle ;
		\draw  [fill={rgb, 255:red, 255; green, 255; blue, 255 }  ,fill opacity=1 ] (581.11,365.7) .. controls (581.11,363.02) and (583.44,360.84) .. (586.31,360.84) .. controls (589.18,360.84) and (591.51,363.02) .. (591.51,365.7) .. controls (591.51,368.38) and (589.18,370.55) .. (586.31,370.55) .. controls (583.44,370.55) and (581.11,368.38) .. (581.11,365.7) -- cycle ;
		\draw  [fill={rgb, 255:red, 255; green, 255; blue, 255 }  ,fill opacity=1 ] (618.2,383.11) .. controls (618.2,380.43) and (620.53,378.26) .. (623.4,378.26) .. controls (626.27,378.26) and (628.6,380.43) .. (628.6,383.11) .. controls (628.6,385.79) and (626.27,387.97) .. (623.4,387.97) .. controls (620.53,387.97) and (618.2,385.79) .. (618.2,383.11) -- cycle ;
		
		\draw (504.17,343.03) node  [font=\large] [align=left] {1};
		\draw (535.3,352.28) node  [font=\large] [align=left] {2};
		\draw (543.85,381.13) node  [font=\large] [align=left] {3};
		\draw (535.3,406.7) node  [font=\large] [align=left] {4};
		\draw (504.76,421.85) node  [font=\large] [align=left] {5};
		\draw (442.67,350.56) node  [font=\large] [align=left] {8};
		\draw (481.28,436.97) node  [font=\large] [align=left] {\ALG};
		\draw (418.7,385.99) node  [font=\large] [align=left] {7};
		\draw (449.48,415.92) node  [font=\large] [align=left] {6};
		\draw (648.17,343.42) node  [font=\large] [align=left] {1};
		\draw (679.3,352.67) node  [font=\large] [align=left] {2};
		\draw (687.85,381.52) node  [font=\large] [align=left] {3};
		\draw (679.3,407.09) node  [font=\large] [align=left] {4};
		\draw (648.76,422.24) node  [font=\large] [align=left] {5};
		\draw (586.67,350.95) node  [font=\large] [align=left] {8};
		\draw (563.3,386.41) node  [font=\large] [align=left] {7};
		\draw (594.08,416.35) node  [font=\large] [align=left] {6};
		\draw (624.28,437.36) node  [font=\large] [align=left] {\OPT};

	\end{tikzpicture}
	
}

\newcommand{\figStarTwoC}{

\tikzset{every picture/.style={line width=0.75pt}} 

\begin{tikzpicture}[x=0.75pt,y=0.75pt,yscale=-1,xscale=1]
	
	\draw [color={rgb, 255:red, 65; green, 117; blue, 5 }  ,draw opacity=0.44 ][fill={rgb, 255:red, 255; green, 255; blue, 255 }  ,fill opacity=1 ][line width=2.25]    (58,399.78) -- (88.8,383.69) ;
	\draw [color={rgb, 255:red, 65; green, 117; blue, 5 }  ,draw opacity=0.44 ][fill={rgb, 255:red, 177; green, 77; blue, 55 }  ,fill opacity=0.55 ][line width=2.25]    (46.66,384.82) -- (88.8,383.69) ;
	\draw [fill={rgb, 255:red, 255; green, 255; blue, 255 }  ,fill opacity=1 ] [dash pattern={on 0.84pt off 2.51pt}]  (52.31,366.31) -- (89.4,383.72) ;
	\draw [fill={rgb, 255:red, 255; green, 255; blue, 255 }  ,fill opacity=1 ] [dash pattern={on 0.84pt off 2.51pt}]  (89.4,383.72) -- (110.48,407.67) ;
	\draw [color={rgb, 255:red, 74; green, 74; blue, 74 }  ,draw opacity=1 ][fill={rgb, 255:red, 255; green, 255; blue, 255 }  ,fill opacity=1 ][line width=0.75]  [dash pattern={on 0.84pt off 2.51pt}]  (89.4,383.72) -- (128.33,398.39) ;
	\draw [color={rgb, 255:red, 74; green, 74; blue, 74 }  ,draw opacity=1 ][fill={rgb, 255:red, 255; green, 255; blue, 255 }  ,fill opacity=1 ][line width=0.75]  [dash pattern={on 0.84pt off 2.51pt}]  (89.4,383.72) -- (122.55,381.82) -- (133.89,381) ;
	\draw [color={rgb, 255:red, 199; green, 177; blue, 144 }  ,draw opacity=1 ][fill={rgb, 255:red, 255; green, 255; blue, 255 }  ,fill opacity=1 ][line width=2.25]    (89.4,383.72) -- (129.43,366.85) ;
	\draw [color={rgb, 255:red, 199; green, 177; blue, 144 }  ,draw opacity=1 ][fill={rgb, 255:red, 255; green, 255; blue, 255 }  ,fill opacity=1 ][line width=2.25]    (107.03,362.45) -- (89.4,383.72) ;
	\draw  [fill={rgb, 255:red, 199; green, 177; blue, 144 }  ,fill opacity=0.5 ] (105.87,358.69) .. controls (105.87,356.01) and (108.19,353.83) .. (111.06,353.83) .. controls (113.93,353.83) and (116.26,356.01) .. (116.26,358.69) .. controls (116.26,361.37) and (113.93,363.54) .. (111.06,363.54) .. controls (108.19,363.54) and (105.87,361.37) .. (105.87,358.69) -- cycle ;
	\draw  [fill={rgb, 255:red, 199; green, 177; blue, 144 }  ,fill opacity=0.5 ] (128.83,364.67) .. controls (128.83,361.99) and (131.16,359.82) .. (134.03,359.82) .. controls (136.9,359.82) and (139.23,361.99) .. (139.23,364.67) .. controls (139.23,367.36) and (136.9,369.53) .. (134.03,369.53) .. controls (131.16,369.53) and (128.83,367.36) .. (128.83,364.67) -- cycle ;
	\draw  [fill={rgb, 255:red, 255; green, 255; blue, 255 }  ,fill opacity=1 ] (133.89,381) .. controls (133.89,378.32) and (136.22,376.14) .. (139.09,376.14) .. controls (141.96,376.14) and (144.29,378.32) .. (144.29,381) .. controls (144.29,383.68) and (141.96,385.86) .. (139.09,385.86) .. controls (136.22,385.86) and (133.89,383.68) .. (133.89,381) -- cycle ;
	\draw  [color={rgb, 255:red, 74; green, 74; blue, 74 }  ,draw opacity=1 ][fill={rgb, 255:red, 255; green, 255; blue, 255 }  ,fill opacity=1 ][line width=0.75]  (127.67,400.59) .. controls (127.67,397.91) and (129.99,395.74) .. (132.86,395.74) .. controls (135.74,395.74) and (138.06,397.91) .. (138.06,400.59) .. controls (138.06,403.27) and (135.74,405.45) .. (132.86,405.45) .. controls (129.99,405.45) and (127.67,403.27) .. (127.67,400.59) -- cycle ;
	\draw  [fill={rgb, 255:red, 255; green, 255; blue, 255 }  ,fill opacity=1 ] (105.28,407.67) .. controls (105.28,404.98) and (107.61,402.81) .. (110.48,402.81) .. controls (113.35,402.81) and (115.68,404.98) .. (115.68,407.67) .. controls (115.68,410.35) and (113.35,412.52) .. (110.48,412.52) .. controls (107.61,412.52) and (105.28,410.35) .. (105.28,407.67) -- cycle ;
	\draw  [fill={rgb, 255:red, 255; green, 255; blue, 255 }  ,fill opacity=1 ] (47.11,366.31) .. controls (47.11,363.63) and (49.44,361.45) .. (52.31,361.45) .. controls (55.18,361.45) and (57.51,363.63) .. (57.51,366.31) .. controls (57.51,368.99) and (55.18,371.16) .. (52.31,371.16) .. controls (49.44,371.16) and (47.11,368.99) .. (47.11,366.31) -- cycle ;
	\draw  [fill={rgb, 255:red, 255; green, 255; blue, 255 }  ,fill opacity=1 ] (84.2,383.72) .. controls (84.2,381.04) and (86.53,378.87) .. (89.4,378.87) .. controls (92.27,378.87) and (94.6,381.04) .. (94.6,383.72) .. controls (94.6,386.4) and (92.27,388.58) .. (89.4,388.58) .. controls (86.53,388.58) and (84.2,386.4) .. (84.2,383.72) -- cycle ;
	\draw  [fill={rgb, 255:red, 65; green, 117; blue, 5 }  ,fill opacity=0.44 ] (36.26,384.82) .. controls (36.26,382.14) and (38.59,379.97) .. (41.46,379.97) .. controls (44.33,379.97) and (46.66,382.14) .. (46.66,384.82) .. controls (46.66,387.5) and (44.33,389.68) .. (41.46,389.68) .. controls (38.59,389.68) and (36.26,387.5) .. (36.26,384.82) -- cycle ;
	\draw  [fill={rgb, 255:red, 65; green, 117; blue, 5 }  ,fill opacity=0.44 ] (48.85,402.19) .. controls (48.85,399.51) and (51.17,397.34) .. (54.04,397.34) .. controls (56.91,397.34) and (59.24,399.51) .. (59.24,402.19) .. controls (59.24,404.87) and (56.91,407.05) .. (54.04,407.05) .. controls (51.17,407.05) and (48.85,404.87) .. (48.85,402.19) -- cycle ;
	\draw [color={rgb, 255:red, 188; green, 155; blue, 199 }  ,draw opacity=0.55 ][fill={rgb, 255:red, 255; green, 255; blue, 255 }  ,fill opacity=1 ][line width=2.25]    (244.21,383.69) -- (262.44,403.84) ;
	\draw [color={rgb, 255:red, 122; green, 144; blue, 199 }  ,draw opacity=0.55 ][fill={rgb, 255:red, 255; green, 255; blue, 255 }  ,fill opacity=1 ][line width=2.25]    (244.21,383.69) -- (283.14,398.36) ;
	\draw [color={rgb, 255:red, 122; green, 144; blue, 199 }  ,draw opacity=0.66 ][fill={rgb, 255:red, 255; green, 255; blue, 255 }  ,fill opacity=1 ][line width=2.25]    (244.21,383.69) -- (277.36,381.79) -- (288.7,380.97) ;
	\draw [color={rgb, 255:red, 199; green, 177; blue, 144 }  ,draw opacity=1 ][fill={rgb, 255:red, 255; green, 255; blue, 255 }  ,fill opacity=1 ][line width=2.25]    (244.21,383.69) -- (284.24,366.81) ;
	\draw [color={rgb, 255:red, 199; green, 177; blue, 144 }  ,draw opacity=1 ][fill={rgb, 255:red, 255; green, 255; blue, 255 }  ,fill opacity=1 ][line width=2.25]    (261.84,362.41) -- (244.21,383.69) ;
	\draw [color={rgb, 255:red, 188; green, 144; blue, 199 }  ,draw opacity=0.55 ][fill={rgb, 255:red, 255; green, 255; blue, 255 }  ,fill opacity=1 ][line width=2.25]    (213.57,399.61) -- (246.21,382.69) ;
	\draw [color={rgb, 255:red, 0; green, 0; blue, 0 }  ,draw opacity=0.44 ][fill={rgb, 255:red, 255; green, 255; blue, 255 }  ,fill opacity=1 ][line width=2.25]    (202.06,384.82) -- (227.44,384.31) -- (244.21,383.69) ;
	\draw  [fill={rgb, 255:red, 199; green, 177; blue, 144 }  ,fill opacity=0.5 ] (260.67,358.66) .. controls (260.67,355.98) and (263,353.8) .. (265.87,353.8) .. controls (268.74,353.8) and (271.07,355.98) .. (271.07,358.66) .. controls (271.07,361.34) and (268.74,363.51) .. (265.87,363.51) .. controls (263,363.51) and (260.67,361.34) .. (260.67,358.66) -- cycle ;
	\draw  [fill={rgb, 255:red, 199; green, 177; blue, 144 }  ,fill opacity=0.5 ] (283.64,364.64) .. controls (283.64,361.96) and (285.97,359.79) .. (288.84,359.79) .. controls (291.71,359.79) and (294.04,361.96) .. (294.04,364.64) .. controls (294.04,367.32) and (291.71,369.5) .. (288.84,369.5) .. controls (285.97,369.5) and (283.64,367.32) .. (283.64,364.64) -- cycle ;
	\draw  [fill={rgb, 255:red, 111; green, 144; blue, 199 }  ,fill opacity=0.55 ] (288.7,380.97) .. controls (288.7,378.29) and (291.03,376.11) .. (293.9,376.11) .. controls (296.77,376.11) and (299.1,378.29) .. (299.1,380.97) .. controls (299.1,383.65) and (296.77,385.82) .. (293.9,385.82) .. controls (291.03,385.82) and (288.7,383.65) .. (288.7,380.97) -- cycle ;
	\draw  [fill={rgb, 255:red, 111; green, 144; blue, 199 }  ,fill opacity=0.44 ] (282.47,400.56) .. controls (282.47,397.88) and (284.8,395.7) .. (287.67,395.7) .. controls (290.54,395.7) and (292.87,397.88) .. (292.87,400.56) .. controls (292.87,403.24) and (290.54,405.42) .. (287.67,405.42) .. controls (284.8,405.42) and (282.47,403.24) .. (282.47,400.56) -- cycle ;
	\draw  [fill={rgb, 255:red, 188; green, 155; blue, 199 }  ,fill opacity=0.5 ] (260.09,407.63) .. controls (260.09,404.95) and (262.42,402.78) .. (265.29,402.78) .. controls (268.16,402.78) and (270.49,404.95) .. (270.49,407.63) .. controls (270.49,410.32) and (268.16,412.49) .. (265.29,412.49) .. controls (262.42,412.49) and (260.09,410.32) .. (260.09,407.63) -- cycle ;
	\draw  [fill={rgb, 255:red, 177; green, 177; blue, 177 }  ,fill opacity=0.7 ] (191.67,384.82) .. controls (191.67,382.14) and (193.99,379.97) .. (196.87,379.97) .. controls (199.74,379.97) and (202.06,382.14) .. (202.06,384.82) .. controls (202.06,387.5) and (199.74,389.68) .. (196.87,389.68) .. controls (193.99,389.68) and (191.67,387.5) .. (191.67,384.82) -- cycle ;
	\draw  [color={rgb, 255:red, 0; green, 0; blue, 0 }  ,draw opacity=1 ][fill={rgb, 255:red, 188; green, 155; blue, 199 }  ,fill opacity=0.55 ] (204.25,402.19) .. controls (204.25,399.51) and (206.58,397.34) .. (209.45,397.34) .. controls (212.32,397.34) and (214.65,399.51) .. (214.65,402.19) .. controls (214.65,404.87) and (212.32,407.05) .. (209.45,407.05) .. controls (206.58,407.05) and (204.25,404.87) .. (204.25,402.19) -- cycle ;
	\draw  [fill={rgb, 255:red, 255; green, 255; blue, 255 }  ,fill opacity=1 ] (239.01,383.69) .. controls (239.01,381.01) and (241.34,378.83) .. (244.21,378.83) .. controls (247.08,378.83) and (249.41,381.01) .. (249.41,383.69) .. controls (249.41,386.37) and (247.08,388.55) .. (244.21,388.55) .. controls (241.34,388.55) and (239.01,386.37) .. (239.01,383.69) -- cycle ;
	\draw [color={rgb, 255:red, 0; green, 0; blue, 0 }  ,draw opacity=0.44 ][fill={rgb, 255:red, 255; green, 255; blue, 255 }  ,fill opacity=1 ][line width=2.25]    (211.64,369.04) -- (239.81,380.49) ;
	\draw  [fill={rgb, 255:red, 177; green, 177; blue, 177 }  ,fill opacity=0.7 ] (202.67,366.82) .. controls (202.67,364.14) and (204.99,361.97) .. (207.87,361.97) .. controls (210.74,361.97) and (213.06,364.14) .. (213.06,366.82) .. controls (213.06,369.5) and (210.74,371.68) .. (207.87,371.68) .. controls (204.99,371.68) and (202.67,369.5) .. (202.67,366.82) -- cycle ;
	
	\draw (114.17,344.03) node  [font=\large] [align=left] {1};
	\draw (145.3,353.28) node  [font=\large] [align=left] {2};
	\draw (153.85,382.13) node  [font=\large] [align=left] {3};
	\draw (145.3,407.7) node  [font=\large] [align=left] {4};
	\draw (114.76,422.85) node  [font=\large] [align=left] {5};
	\draw (52.67,351.56) node  [font=\large] [align=left] {8};
	\draw (91.28,437.97) node  [font=\large] [align=left] {\ALG};
	\draw (28.7,386.99) node  [font=\large] [align=left] {7};
	\draw (59.48,416.92) node  [font=\large] [align=left] {6};
	\draw (268.98,344) node  [font=\large] [align=left] {1};
	\draw (300.1,353.25) node  [font=\large] [align=left] {2};
	\draw (308.66,382.09) node  [font=\large] [align=left] {3};
	\draw (300.1,407.67) node  [font=\large] [align=left] {4};
	\draw (269.56,422.82) node  [font=\large] [align=left] {5};
	\draw (207.48,351.53) node  [font=\large] [align=left] {8};
	\draw (184.11,386.99) node  [font=\large] [align=left] {7};
	\draw (214.89,416.92) node  [font=\large] [align=left] {6};
	\draw (245.09,437.94) node  [font=\large] [align=left] {\OPT};

\end{tikzpicture}

}
\newcommand{\figStarTwoB}{

	\tikzset{every picture/.style={line width=0.75pt}} 
	
	\begin{tikzpicture}[x=0.75pt,y=0.75pt,yscale=-1,xscale=1]
		
\draw [fill={rgb, 255:red, 255; green, 255; blue, 255 }  ,fill opacity=1 ] [dash pattern={on 0.84pt off 2.51pt}] (436.31,202.7) -- (473.4,220.11) ;
\draw [fill={rgb, 255:red, 255; green, 255; blue, 255 }  ,fill opacity=1 ] [dash pattern={on 0.84pt off 2.51pt}] (473.4,220.11) -- (494.48,244.06) ;
\draw [fill={rgb, 255:red, 255; green, 255; blue, 255 }  ,fill opacity=1 ] [dash pattern={on 0.84pt off 2.51pt}] (473.4,220.11) -- (513.43,203.24) ;
\draw [fill={rgb, 255:red, 255; green, 255; blue, 255 }  ,fill opacity=1 ] [dash pattern={on 0.84pt off 2.51pt}] (491.03,198.84) -- (473.4,220.11) ;
		\draw  [color={rgb, 255:red, 0; green, 0; blue, 0 }  ,draw opacity=1 ][fill={rgb, 255:red, 255; green, 255; blue, 255 }  ,fill opacity=1 ][line width=0.75]  (489.87,195.08) .. controls (489.87,192.4) and (492.19,190.22) .. (495.06,190.22) .. controls (497.93,190.22) and (500.26,192.4) .. (500.26,195.08) .. controls (500.26,197.76) and (497.93,199.93) .. (495.06,199.93) .. controls (492.19,199.93) and (489.87,197.76) .. (489.87,195.08) -- cycle ;
		\draw  [fill={rgb, 255:red, 255; green, 255; blue, 255 }  ,fill opacity=1 ] (512.83,201.07) .. controls (512.83,198.38) and (515.16,196.21) .. (518.03,196.21) .. controls (520.9,196.21) and (523.23,198.38) .. (523.23,201.07) .. controls (523.23,203.75) and (520.9,205.92) .. (518.03,205.92) .. controls (515.16,205.92) and (512.83,203.75) .. (512.83,201.07) -- cycle ;
		\draw  [fill={rgb, 255:red, 100; green, 122; blue, 144 }  ,fill opacity=0.55 ] (517.89,217.39) .. controls (517.89,214.71) and (520.22,212.54) .. (523.09,212.54) .. controls (525.96,212.54) and (528.29,214.71) .. (528.29,217.39) .. controls (528.29,220.07) and (525.96,222.25) .. (523.09,222.25) .. controls (520.22,222.25) and (517.89,220.07) .. (517.89,217.39) -- cycle ;
		\draw  [fill={rgb, 255:red, 100; green, 122; blue, 144 }  ,fill opacity=0.55 ] (511.67,236.98) .. controls (511.67,234.3) and (513.99,232.13) .. (516.86,232.13) .. controls (519.74,232.13) and (522.06,234.3) .. (522.06,236.98) .. controls (522.06,239.66) and (519.74,241.84) .. (516.86,241.84) .. controls (513.99,241.84) and (511.67,239.66) .. (511.67,236.98) -- cycle ;
		\draw  [color={rgb, 255:red, 0; green, 0; blue, 0 }  ,draw opacity=1 ][fill={rgb, 255:red, 255; green, 255; blue, 255 }  ,fill opacity=1 ][line width=0.75]  (489.28,244.06) .. controls (489.28,241.38) and (491.61,239.2) .. (494.48,239.2) .. controls (497.35,239.2) and (499.68,241.38) .. (499.68,244.06) .. controls (499.68,246.74) and (497.35,248.91) .. (494.48,248.91) .. controls (491.61,248.91) and (489.28,246.74) .. (489.28,244.06) -- cycle ;
		\draw  [color={rgb, 255:red, 0; green, 0; blue, 0 }  ,draw opacity=1 ][fill={rgb, 255:red, 255; green, 255; blue, 255 }  ,fill opacity=1 ][line width=0.75]  (431.11,202.7) .. controls (431.11,200.02) and (433.44,197.84) .. (436.31,197.84) .. controls (439.18,197.84) and (441.51,200.02) .. (441.51,202.7) .. controls (441.51,205.38) and (439.18,207.55) .. (436.31,207.55) .. controls (433.44,207.55) and (431.11,205.38) .. (431.11,202.7) -- cycle ;
		\draw [color={rgb, 255:red, 100; green, 122; blue, 144 }  ,draw opacity=0.55 ][fill={rgb, 255:red, 255; green, 255; blue, 255 }  ,fill opacity=1 ][line width=2.25]    (478.6,220.11) -- (517.89,217.39) ;
		\draw [color={rgb, 255:red, 100; green, 122; blue, 144 }  ,draw opacity=0.55 ][fill={rgb, 255:red, 255; green, 255; blue, 255 }  ,fill opacity=1 ][line width=2.25]    (478.83,222.24) -- (512.03,234.64) ;
		\draw [color={rgb, 255:red, 65; green, 117; blue, 5 }  ,draw opacity=0.44 ][fill={rgb, 255:red, 255; green, 255; blue, 255 }  ,fill opacity=1 ][line width=2.25]    (442,236.78) -- (472.8,220.69) ;
		\draw [color={rgb, 255:red, 65; green, 117; blue, 5 }  ,draw opacity=0.44 ][fill={rgb, 255:red, 177; green, 77; blue, 55 }  ,fill opacity=0.55 ][line width=2.25]    (430.66,221.82) -- (472.8,220.69) ;
		\draw  [fill={rgb, 255:red, 65; green, 117; blue, 5 }  ,fill opacity=0.44 ] (420.26,221.82) .. controls (420.26,219.14) and (422.59,216.97) .. (425.46,216.97) .. controls (428.33,216.97) and (430.66,219.14) .. (430.66,221.82) .. controls (430.66,224.5) and (428.33,226.68) .. (425.46,226.68) .. controls (422.59,226.68) and (420.26,224.5) .. (420.26,221.82) -- cycle ;
		\draw  [fill={rgb, 255:red, 65; green, 117; blue, 5 }  ,fill opacity=0.44 ] (432.85,239.19) .. controls (432.85,236.51) and (435.17,234.34) .. (438.04,234.34) .. controls (440.91,234.34) and (443.24,236.51) .. (443.24,239.19) .. controls (443.24,241.87) and (440.91,244.05) .. (438.04,244.05) .. controls (435.17,244.05) and (432.85,241.87) .. (432.85,239.19) -- cycle ;
		\draw  [color={rgb, 255:red, 0; green, 0; blue, 0 }  ,draw opacity=1 ][fill={rgb, 255:red, 255; green, 255; blue, 255 }  ,fill opacity=1 ][line width=0.75]  (468.2,220.11) .. controls (468.2,217.43) and (470.53,215.26) .. (473.4,215.26) .. controls (476.27,215.26) and (478.6,217.43) .. (478.6,220.11) .. controls (478.6,222.79) and (476.27,224.97) .. (473.4,224.97) .. controls (470.53,224.97) and (468.2,222.79) .. (468.2,220.11) -- cycle ;
		\draw [color={rgb, 255:red, 189; green, 16; blue, 224 }  ,draw opacity=0.33 ][fill={rgb, 255:red, 255; green, 255; blue, 255 }  ,fill opacity=1 ][line width=2.25]    (587.36,235.81) -- (620,218.88) ;
		\draw [color={rgb, 255:red, 189; green, 16; blue, 224 }  ,draw opacity=0.33 ][fill={rgb, 255:red, 255; green, 255; blue, 255 }  ,fill opacity=1 ][line width=2.25]    (636.56,196.77) -- (618,219.88) ;
		\draw [color={rgb, 255:red, 100; green, 122; blue, 144 }  ,draw opacity=0.55 ][fill={rgb, 255:red, 255; green, 255; blue, 255 }  ,fill opacity=1 ][line width=2.25]    (618.4,218.46) -- (662.89,215.74) ;
		\draw [color={rgb, 255:red, 100; green, 122; blue, 144 }  ,draw opacity=0.55 ][fill={rgb, 255:red, 255; green, 255; blue, 255 }  ,fill opacity=1 ][line width=2.25]    (618.8,217.99) -- (659.02,200.33) ;
		\draw  [fill={rgb, 255:red, 189; green, 16; blue, 224 }  ,fill opacity=0.22 ] (634.87,193.43) .. controls (634.87,190.74) and (637.19,188.57) .. (640.06,188.57) .. controls (642.93,188.57) and (645.26,190.74) .. (645.26,193.43) .. controls (645.26,196.11) and (642.93,198.28) .. (640.06,198.28) .. controls (637.19,198.28) and (634.87,196.11) .. (634.87,193.43) -- cycle ;
		\draw  [color={rgb, 255:red, 0; green, 0; blue, 0 }  ,draw opacity=1 ][fill={rgb, 255:red, 100; green, 122; blue, 144 }  ,fill opacity=0.55 ] (657.83,199.41) .. controls (657.83,196.73) and (660.16,194.56) .. (663.03,194.56) .. controls (665.9,194.56) and (668.23,196.73) .. (668.23,199.41) .. controls (668.23,202.09) and (665.9,204.27) .. (663.03,204.27) .. controls (660.16,204.27) and (657.83,202.09) .. (657.83,199.41) -- cycle ;
		\draw  [fill={rgb, 255:red, 100; green, 122; blue, 144 }  ,fill opacity=0.55 ] (662.89,215.74) .. controls (662.89,213.06) and (665.22,210.88) .. (668.09,210.88) .. controls (670.96,210.88) and (673.29,213.06) .. (673.29,215.74) .. controls (673.29,218.42) and (670.96,220.59) .. (668.09,220.59) .. controls (665.22,220.59) and (662.89,218.42) .. (662.89,215.74) -- cycle ;
		\draw  [fill={rgb, 255:red, 139; green, 87; blue, 42 }  ,fill opacity=0.33 ] (656.67,235.33) .. controls (656.67,232.65) and (658.99,230.47) .. (661.86,230.47) .. controls (664.74,230.47) and (667.06,232.65) .. (667.06,235.33) .. controls (667.06,238.01) and (664.74,240.19) .. (661.86,240.19) .. controls (658.99,240.19) and (656.67,238.01) .. (656.67,235.33) -- cycle ;
		\draw  [fill={rgb, 255:red, 139; green, 87; blue, 42 }  ,fill opacity=0.33 ] (634.28,242.4) .. controls (634.28,239.72) and (636.61,237.55) .. (639.48,237.55) .. controls (642.35,237.55) and (644.68,239.72) .. (644.68,242.4) .. controls (644.68,245.09) and (642.35,247.26) .. (639.48,247.26) .. controls (636.61,247.26) and (634.28,245.09) .. (634.28,242.4) -- cycle ;
		\draw [color={rgb, 255:red, 139; green, 87; blue, 42 }  ,draw opacity=0.44 ][fill={rgb, 255:red, 139; green, 87; blue, 42 }  ,fill opacity=0.44 ][line width=2.25]    (622.17,222.13) -- (636.02,238.69) ;
		\draw [color={rgb, 255:red, 139; green, 87; blue, 42 }  ,draw opacity=0.44 ][fill={rgb, 255:red, 255; green, 255; blue, 255 }  ,fill opacity=1 ][line width=2.25]    (623.51,220.35) -- (657.73,232.79) ;
		\draw  [color={rgb, 255:red, 0; green, 0; blue, 0 }  ,draw opacity=1 ][fill={rgb, 255:red, 189; green, 16; blue, 224 }  ,fill opacity=0.22 ] (578.55,238.39) .. controls (578.55,235.7) and (580.87,233.53) .. (583.74,233.53) .. controls (586.61,233.53) and (588.94,235.7) .. (588.94,238.39) .. controls (588.94,241.07) and (586.61,243.24) .. (583.74,243.24) .. controls (580.87,243.24) and (578.55,241.07) .. (578.55,238.39) -- cycle ;
		\draw [color={rgb, 255:red, 0; green, 0; blue, 0 }  ,draw opacity=0.44 ][fill={rgb, 255:red, 255; green, 255; blue, 255 }  ,fill opacity=1 ][line width=2.25]    (576.66,219.12) -- (602.03,218.61) -- (618.8,217.99) ;
		\draw  [fill={rgb, 255:red, 177; green, 177; blue, 177 }  ,fill opacity=0.7 ] (566.26,219.12) .. controls (566.26,216.44) and (568.59,214.27) .. (571.46,214.27) .. controls (574.33,214.27) and (576.66,216.44) .. (576.66,219.12) .. controls (576.66,221.8) and (574.33,223.98) .. (571.46,223.98) .. controls (568.59,223.98) and (566.26,221.8) .. (566.26,219.12) -- cycle ;
		\draw [color={rgb, 255:red, 0; green, 0; blue, 0 }  ,draw opacity=0.44 ][fill={rgb, 255:red, 255; green, 255; blue, 255 }  ,fill opacity=1 ][line width=2.25]    (586.23,203.34) -- (614.4,214.79) ;
		\draw  [fill={rgb, 255:red, 177; green, 177; blue, 177 }  ,fill opacity=0.7 ] (577.26,201.12) .. controls (577.26,198.44) and (579.59,196.27) .. (582.46,196.27) .. controls (585.33,196.27) and (587.66,198.44) .. (587.66,201.12) .. controls (587.66,203.8) and (585.33,205.98) .. (582.46,205.98) .. controls (579.59,205.98) and (577.26,203.8) .. (577.26,201.12) -- cycle ;
		\draw  [fill={rgb, 255:red, 255; green, 255; blue, 255 }  ,fill opacity=1 ] (613.2,218.46) .. controls (613.2,215.78) and (615.53,213.6) .. (618.4,213.6) .. controls (621.27,213.6) and (623.6,215.78) .. (623.6,218.46) .. controls (623.6,221.14) and (621.27,223.31) .. (618.4,223.31) .. controls (615.53,223.31) and (613.2,221.14) .. (613.2,218.46) -- cycle ;
		
		\draw (498.17,180.42) node  [font=\large] [align=left] {1};
		\draw (529.3,189.67) node  [font=\large] [align=left] {2};
		\draw (537.85,218.52) node  [font=\large] [align=left] {3};
		\draw (529.3,244.09) node  [font=\large] [align=left] {4};
		\draw (498.76,259.24) node  [font=\large] [align=left] {5};
		\draw (436.67,187.95) node  [font=\large] [align=left] {8};
		\draw (474.28,274.36) node  [font=\large] [align=left] {\ALG};
		\draw (412.7,223.99) node  [font=\large] [align=left] {7};
		\draw (443.48,253.92) node  [font=\large] [align=left] {6};
		\draw (643.17,178.77) node  [font=\large] [align=left] {1};
		\draw (674.3,188.02) node  [font=\large] [align=left] {2};
		\draw (682.85,216.86) node  [font=\large] [align=left] {3};
		\draw (674.3,242.44) node  [font=\large] [align=left] {4};
		\draw (644.76,256.59) node  [font=\large] [align=left] {5};
		\draw (686.88,279.69) node   [align=left] {};
		\draw (619.68,274.89) node  [font=\large] [align=left] {\OPT};
		\draw (588.68,253.12) node  [font=\large] [align=left] {6};
		\draw (584.47,185.47) node  [font=\large] [align=left] {8};
		\draw (561.5,218.47) node  [font=\large] [align=left] {7};

	\end{tikzpicture}
	
}

\newcommand{\figStarTwoA}{

	\tikzset{every picture/.style={line width=0.75pt}} 
	
	\begin{tikzpicture}[x=0.75pt,y=0.75pt,yscale=-1,xscale=1]
		
		\draw [color={rgb, 255:red, 0; green, 0; blue, 0 }  ,draw opacity=0.44 ][fill={rgb, 255:red, 255; green, 255; blue, 255 }  ,fill opacity=1 ][line width=2.25]    (219.76,239.04) -- (252.4,222.11) ;
		\draw [color={rgb, 255:red, 0; green, 0; blue, 0 }  ,draw opacity=0.44 ][fill={rgb, 255:red, 255; green, 255; blue, 255 }  ,fill opacity=1 ][line width=2.25]    (208.26,224.24) -- (233.63,223.73) -- (250.4,223.11) ;
\draw [fill={rgb, 255:red, 255; green, 255; blue, 255 }  ,fill opacity=1 ] [dash pattern={on 0.84pt off 2.51pt}] (60.31,208.7) -- (97.4,226.11) ;
\draw [fill={rgb, 255:red, 255; green, 255; blue, 255 }  ,fill opacity=1 ] [dash pattern={on 0.84pt off 2.51pt}] (97.4,226.11) -- (118.48,250.06) ;
\draw [fill={rgb, 255:red, 255; green, 255; blue, 255 }  ,fill opacity=1 ] [dash pattern={on 0.84pt off 2.51pt}] (97.4,226.11) -- (137.43,209.24) ;
\draw [fill={rgb, 255:red, 255; green, 255; blue, 255 }  ,fill opacity=1 ] [dash pattern={on 0.84pt off 2.51pt}] (115.03,204.84) -- (97.4,226.11) ;
\draw [fill={rgb, 255:red, 255; green, 255; blue, 255 }  ,fill opacity=1 ] [dash pattern={on 0.84pt off 2.51pt}] (66.76,242.04) -- (99.4,225.11) ;
\draw [fill={rgb, 255:red, 255; green, 255; blue, 255 }  ,fill opacity=1 ] [dash pattern={on 0.84pt off 2.51pt}] (55.26,227.24) -- (80.63,226.73) -- (97.4,226.11) ;
		\draw  [color={rgb, 255:red, 0; green, 0; blue, 0 }  ,draw opacity=1 ][fill={rgb, 255:red, 255; green, 255; blue, 255 }  ,fill opacity=1 ][line width=0.75]  (113.87,201.08) .. controls (113.87,198.4) and (116.19,196.22) .. (119.06,196.22) .. controls (121.93,196.22) and (124.26,198.4) .. (124.26,201.08) .. controls (124.26,203.76) and (121.93,205.93) .. (119.06,205.93) .. controls (116.19,205.93) and (113.87,203.76) .. (113.87,201.08) -- cycle ;
		\draw  [fill={rgb, 255:red, 255; green, 255; blue, 255 }  ,fill opacity=1 ] (136.83,207.07) .. controls (136.83,204.38) and (139.16,202.21) .. (142.03,202.21) .. controls (144.9,202.21) and (147.23,204.38) .. (147.23,207.07) .. controls (147.23,209.75) and (144.9,211.92) .. (142.03,211.92) .. controls (139.16,211.92) and (136.83,209.75) .. (136.83,207.07) -- cycle ;
		\draw  [fill={rgb, 255:red, 100; green, 122; blue, 144 }  ,fill opacity=0.55 ] (141.89,223.39) .. controls (141.89,220.71) and (144.22,218.54) .. (147.09,218.54) .. controls (149.96,218.54) and (152.29,220.71) .. (152.29,223.39) .. controls (152.29,226.07) and (149.96,228.25) .. (147.09,228.25) .. controls (144.22,228.25) and (141.89,226.07) .. (141.89,223.39) -- cycle ;
		\draw  [fill={rgb, 255:red, 100; green, 122; blue, 144 }  ,fill opacity=0.55 ] (135.67,242.98) .. controls (135.67,240.3) and (137.99,238.13) .. (140.86,238.13) .. controls (143.74,238.13) and (146.06,240.3) .. (146.06,242.98) .. controls (146.06,245.66) and (143.74,247.84) .. (140.86,247.84) .. controls (137.99,247.84) and (135.67,245.66) .. (135.67,242.98) -- cycle ;
		\draw  [color={rgb, 255:red, 0; green, 0; blue, 0 }  ,draw opacity=1 ][fill={rgb, 255:red, 255; green, 255; blue, 255 }  ,fill opacity=1 ][line width=0.75]  (113.28,250.06) .. controls (113.28,247.38) and (115.61,245.2) .. (118.48,245.2) .. controls (121.35,245.2) and (123.68,247.38) .. (123.68,250.06) .. controls (123.68,252.74) and (121.35,254.91) .. (118.48,254.91) .. controls (115.61,254.91) and (113.28,252.74) .. (113.28,250.06) -- cycle ;
		\draw  [color={rgb, 255:red, 0; green, 0; blue, 0 }  ,draw opacity=1 ][fill={rgb, 255:red, 255; green, 255; blue, 255 }  ,fill opacity=1 ][line width=0.75]  (44.86,227.24) .. controls (44.86,224.56) and (47.19,222.39) .. (50.06,222.39) .. controls (52.93,222.39) and (55.26,224.56) .. (55.26,227.24) .. controls (55.26,229.93) and (52.93,232.1) .. (50.06,232.1) .. controls (47.19,232.1) and (44.86,229.93) .. (44.86,227.24) -- cycle ;
		\draw  [color={rgb, 255:red, 0; green, 0; blue, 0 }  ,draw opacity=1 ][fill={rgb, 255:red, 255; green, 255; blue, 255 }  ,fill opacity=1 ][line width=0.75]  (57.45,244.62) .. controls (57.45,241.93) and (59.77,239.76) .. (62.64,239.76) .. controls (65.51,239.76) and (67.84,241.93) .. (67.84,244.62) .. controls (67.84,247.3) and (65.51,249.47) .. (62.64,249.47) .. controls (59.77,249.47) and (57.45,247.3) .. (57.45,244.62) -- cycle ;
		\draw  [color={rgb, 255:red, 0; green, 0; blue, 0 }  ,draw opacity=1 ][fill={rgb, 255:red, 255; green, 255; blue, 255 }  ,fill opacity=1 ][line width=0.75]  (55.11,208.7) .. controls (55.11,206.02) and (57.44,203.84) .. (60.31,203.84) .. controls (63.18,203.84) and (65.51,206.02) .. (65.51,208.7) .. controls (65.51,211.38) and (63.18,213.55) .. (60.31,213.55) .. controls (57.44,213.55) and (55.11,211.38) .. (55.11,208.7) -- cycle ;
		\draw  [color={rgb, 255:red, 0; green, 0; blue, 0 }  ,draw opacity=1 ][fill={rgb, 255:red, 255; green, 255; blue, 255 }  ,fill opacity=1 ][line width=0.75]  (92.2,226.11) .. controls (92.2,223.43) and (94.53,221.26) .. (97.4,221.26) .. controls (100.27,221.26) and (102.6,223.43) .. (102.6,226.11) .. controls (102.6,228.79) and (100.27,230.97) .. (97.4,230.97) .. controls (94.53,230.97) and (92.2,228.79) .. (92.2,226.11) -- cycle ;
		\draw [fill={rgb, 255:red, 255; green, 255; blue, 255 }  ,fill opacity=1 ] [dash pattern={on 0.84pt off 2.51pt}]  (213.71,204.28) -- (250.8,221.69) ;
		\draw [color={rgb, 255:red, 100; green, 122; blue, 144 }  ,draw opacity=0.55 ][fill={rgb, 255:red, 255; green, 255; blue, 255 }  ,fill opacity=1 ][line width=2.25]    (250.8,221.69) -- (295.29,218.97) ;
		\draw [color={rgb, 255:red, 100; green, 122; blue, 144 }  ,draw opacity=0.55 ][fill={rgb, 255:red, 255; green, 255; blue, 255 }  ,fill opacity=1 ][line width=2.25]    (250.8,221.69) -- (291.02,204.03) ;
		\draw [fill={rgb, 255:red, 255; green, 255; blue, 255 }  ,fill opacity=1 ] [dash pattern={on 0.84pt off 2.51pt}]  (272.46,196.66) -- (250.8,221.69) ;
		\draw  [fill={rgb, 255:red, 255; green, 255; blue, 255 }  ,fill opacity=1 ] (267.27,196.66) .. controls (267.27,193.98) and (269.59,191.8) .. (272.46,191.8) .. controls (275.33,191.8) and (277.66,193.98) .. (277.66,196.66) .. controls (277.66,199.34) and (275.33,201.51) .. (272.46,201.51) .. controls (269.59,201.51) and (267.27,199.34) .. (267.27,196.66) -- cycle ;
		\draw  [color={rgb, 255:red, 0; green, 0; blue, 0 }  ,draw opacity=1 ][fill={rgb, 255:red, 100; green, 122; blue, 144 }  ,fill opacity=0.55 ] (290.23,202.64) .. controls (290.23,199.96) and (292.56,197.79) .. (295.43,197.79) .. controls (298.3,197.79) and (300.63,199.96) .. (300.63,202.64) .. controls (300.63,205.32) and (298.3,207.5) .. (295.43,207.5) .. controls (292.56,207.5) and (290.23,205.32) .. (290.23,202.64) -- cycle ;
		\draw  [fill={rgb, 255:red, 100; green, 122; blue, 144 }  ,fill opacity=0.55 ] (295.29,218.97) .. controls (295.29,216.29) and (297.62,214.11) .. (300.49,214.11) .. controls (303.36,214.11) and (305.69,216.29) .. (305.69,218.97) .. controls (305.69,221.65) and (303.36,223.82) .. (300.49,223.82) .. controls (297.62,223.82) and (295.29,221.65) .. (295.29,218.97) -- cycle ;
		\draw  [fill={rgb, 255:red, 139; green, 87; blue, 42 }  ,fill opacity=0.33 ] (289.07,238.56) .. controls (289.07,235.88) and (291.39,233.7) .. (294.26,233.7) .. controls (297.14,233.7) and (299.46,235.88) .. (299.46,238.56) .. controls (299.46,241.24) and (297.14,243.42) .. (294.26,243.42) .. controls (291.39,243.42) and (289.07,241.24) .. (289.07,238.56) -- cycle ;
		\draw  [fill={rgb, 255:red, 139; green, 87; blue, 42 }  ,fill opacity=0.33 ] (266.68,245.63) .. controls (266.68,242.95) and (269.01,240.78) .. (271.88,240.78) .. controls (274.75,240.78) and (277.08,242.95) .. (277.08,245.63) .. controls (277.08,248.32) and (274.75,250.49) .. (271.88,250.49) .. controls (269.01,250.49) and (266.68,248.32) .. (266.68,245.63) -- cycle ;
		\draw  [fill={rgb, 255:red, 255; green, 255; blue, 255 }  ,fill opacity=1 ] (208.51,204.28) .. controls (208.51,201.59) and (210.84,199.42) .. (213.71,199.42) .. controls (216.58,199.42) and (218.91,201.59) .. (218.91,204.28) .. controls (218.91,206.96) and (216.58,209.13) .. (213.71,209.13) .. controls (210.84,209.13) and (208.51,206.96) .. (208.51,204.28) -- cycle ;
		\draw  [fill={rgb, 255:red, 255; green, 255; blue, 255 }  ,fill opacity=1 ] (245.6,221.69) .. controls (245.6,219.01) and (247.93,216.83) .. (250.8,216.83) .. controls (253.67,216.83) and (256,219.01) .. (256,221.69) .. controls (256,224.37) and (253.67,226.55) .. (250.8,226.55) .. controls (247.93,226.55) and (245.6,224.37) .. (245.6,221.69) -- cycle ;
		\draw [color={rgb, 255:red, 100; green, 122; blue, 144 }  ,draw opacity=0.55 ][fill={rgb, 255:red, 255; green, 255; blue, 255 }  ,fill opacity=1 ][line width=2.25]    (102.6,226.11) -- (141.89,223.39) ;
		\draw [color={rgb, 255:red, 100; green, 122; blue, 144 }  ,draw opacity=0.55 ][fill={rgb, 255:red, 255; green, 255; blue, 255 }  ,fill opacity=1 ][line width=2.25]    (102.83,228.24) -- (136.03,240.64) ;
		\draw [color={rgb, 255:red, 139; green, 87; blue, 42 }  ,draw opacity=0.44 ][fill={rgb, 255:red, 139; green, 87; blue, 42 }  ,fill opacity=0.44 ][line width=2.25]    (254.57,225.36) -- (268.42,241.92) ;
		\draw [color={rgb, 255:red, 139; green, 87; blue, 42 }  ,draw opacity=0.44 ][fill={rgb, 255:red, 255; green, 255; blue, 255 }  ,fill opacity=1 ][line width=2.25]    (255.91,223.58) -- (290.13,236.02) ;
		\draw  [fill={rgb, 255:red, 177; green, 177; blue, 177 }  ,fill opacity=0.7 ] (197.86,224.24) .. controls (197.86,221.56) and (200.19,219.39) .. (203.06,219.39) .. controls (205.93,219.39) and (208.26,221.56) .. (208.26,224.24) .. controls (208.26,226.93) and (205.93,229.1) .. (203.06,229.1) .. controls (200.19,229.1) and (197.86,226.93) .. (197.86,224.24) -- cycle ;
		\draw  [color={rgb, 255:red, 0; green, 0; blue, 0 }  ,draw opacity=1 ][fill={rgb, 255:red, 177; green, 177; blue, 177 }  ,fill opacity=0.7 ] (210.45,241.62) .. controls (210.45,238.93) and (212.77,236.76) .. (215.64,236.76) .. controls (218.51,236.76) and (220.84,238.93) .. (220.84,241.62) .. controls (220.84,244.3) and (218.51,246.47) .. (215.64,246.47) .. controls (212.77,246.47) and (210.45,244.3) .. (210.45,241.62) -- cycle ;
		
		\draw (122.17,186.42) node  [font=\large] [align=left] {1};
		\draw (153.3,195.67) node  [font=\large] [align=left] {2};
		\draw (161.85,224.52) node  [font=\large] [align=left] {3};
		\draw (153.3,250.09) node  [font=\large] [align=left] {4};
		\draw (122.76,265.24) node  [font=\large] [align=left] {5};
		\draw (60.67,193.95) node  [font=\large] [align=left] {8};
		\draw (37.3,229.41) node  [font=\large] [align=left] {7};
		\draw (68.08,259.35) node  [font=\large] [align=left] {6};
		\draw (98.28,280.36) node  [font=\large] [align=left] {\ALG};
		\draw (275.57,182) node  [font=\large] [align=left] {1};
		\draw (306.7,191.25) node  [font=\large] [align=left] {2};
		\draw (315.25,220.09) node  [font=\large] [align=left] {3};
		\draw (306.7,245.67) node  [font=\large] [align=left] {4};
		\draw (277.16,259.82) node  [font=\large] [align=left] {5};
		\draw (215.07,188.53) node  [font=\large] [align=left] {8};
		\draw (319.28,282.92) node   [align=left] {};
		\draw (252.08,278.12) node  [font=\large] [align=left] {\OPT};
		\draw (190.3,226.41) node  [font=\large] [align=left] {7};
		\draw (221.08,256.35) node  [font=\large] [align=left] {6};

	\end{tikzpicture}
}

\newcommand{\figStarOneA}{

\tikzset{every picture/.style={line width=0.75pt}} 

\begin{tikzpicture}[x=0.75pt,y=0.75pt,yscale=-1,xscale=1]
	
	\draw [fill={rgb, 255:red, 255; green, 255; blue, 255 }  ,fill opacity=1 ] [dash pattern={on 0.84pt off 2.51pt}]  (196.31,42.7) -- (233.4,60.11) ;
	\draw [fill={rgb, 255:red, 255; green, 255; blue, 255 }  ,fill opacity=1 ] [dash pattern={on 0.84pt off 2.51pt}]  (233.4,60.11) -- (254.48,84.06) ;
	\draw [color={rgb, 255:red, 122; green, 144; blue, 199 }  ,draw opacity=0.55 ][fill={rgb, 255:red, 255; green, 255; blue, 255 }  ,fill opacity=1 ][line width=2.25]    (233.4,60.11) -- (272.33,74.78) ;
	\draw [color={rgb, 255:red, 122; green, 144; blue, 199 }  ,draw opacity=0.66 ][fill={rgb, 255:red, 255; green, 255; blue, 255 }  ,fill opacity=1 ][line width=2.25]    (233.4,60.11) -- (266.55,58.21) -- (277.89,57.39) ;
	\draw [color={rgb, 255:red, 199; green, 177; blue, 144 }  ,draw opacity=1 ][fill={rgb, 255:red, 255; green, 255; blue, 255 }  ,fill opacity=1 ][line width=2.25]    (233.4,60.11) -- (273.43,43.24) ;
	\draw [color={rgb, 255:red, 199; green, 177; blue, 144 }  ,draw opacity=1 ][fill={rgb, 255:red, 255; green, 255; blue, 255 }  ,fill opacity=1 ][line width=2.25]    (251.03,38.84) -- (233.4,60.11) ;
	\draw [color={rgb, 255:red, 0; green, 0; blue, 0 }  ,draw opacity=0.44 ][fill={rgb, 255:red, 255; green, 255; blue, 255 }  ,fill opacity=1 ][line width=2.25]    (202.76,76.04) -- (235.4,59.11) ;
	\draw [color={rgb, 255:red, 0; green, 0; blue, 0 }  ,draw opacity=0.44 ][fill={rgb, 255:red, 255; green, 255; blue, 255 }  ,fill opacity=1 ][line width=2.25]    (191.26,61.24) -- (216.63,60.73) -- (233.4,60.11) ;
	\draw  [fill={rgb, 255:red, 199; green, 177; blue, 144 }  ,fill opacity=0.5 ] (249.87,35.08) .. controls (249.87,32.4) and (252.19,30.22) .. (255.06,30.22) .. controls (257.93,30.22) and (260.26,32.4) .. (260.26,35.08) .. controls (260.26,37.76) and (257.93,39.93) .. (255.06,39.93) .. controls (252.19,39.93) and (249.87,37.76) .. (249.87,35.08) -- cycle ;
	\draw  [fill={rgb, 255:red, 199; green, 177; blue, 144 }  ,fill opacity=0.5 ] (272.83,41.07) .. controls (272.83,38.38) and (275.16,36.21) .. (278.03,36.21) .. controls (280.9,36.21) and (283.23,38.38) .. (283.23,41.07) .. controls (283.23,43.75) and (280.9,45.92) .. (278.03,45.92) .. controls (275.16,45.92) and (272.83,43.75) .. (272.83,41.07) -- cycle ;
	\draw  [fill={rgb, 255:red, 111; green, 144; blue, 199 }  ,fill opacity=0.55 ] (277.89,57.39) .. controls (277.89,54.71) and (280.22,52.54) .. (283.09,52.54) .. controls (285.96,52.54) and (288.29,54.71) .. (288.29,57.39) .. controls (288.29,60.07) and (285.96,62.25) .. (283.09,62.25) .. controls (280.22,62.25) and (277.89,60.07) .. (277.89,57.39) -- cycle ;
	\draw  [fill={rgb, 255:red, 111; green, 144; blue, 199 }  ,fill opacity=0.44 ] (271.67,76.98) .. controls (271.67,74.3) and (273.99,72.13) .. (276.86,72.13) .. controls (279.74,72.13) and (282.06,74.3) .. (282.06,76.98) .. controls (282.06,79.66) and (279.74,81.84) .. (276.86,81.84) .. controls (273.99,81.84) and (271.67,79.66) .. (271.67,76.98) -- cycle ;
	\draw  [fill={rgb, 255:red, 255; green, 255; blue, 255 }  ,fill opacity=1 ] (249.28,84.06) .. controls (249.28,81.38) and (251.61,79.2) .. (254.48,79.2) .. controls (257.35,79.2) and (259.68,81.38) .. (259.68,84.06) .. controls (259.68,86.74) and (257.35,88.91) .. (254.48,88.91) .. controls (251.61,88.91) and (249.28,86.74) .. (249.28,84.06) -- cycle ;
	\draw  [fill={rgb, 255:red, 177; green, 177; blue, 177 }  ,fill opacity=0.7 ] (180.86,61.24) .. controls (180.86,58.56) and (183.19,56.39) .. (186.06,56.39) .. controls (188.93,56.39) and (191.26,58.56) .. (191.26,61.24) .. controls (191.26,63.93) and (188.93,66.1) .. (186.06,66.1) .. controls (183.19,66.1) and (180.86,63.93) .. (180.86,61.24) -- cycle ;
	\draw  [color={rgb, 255:red, 0; green, 0; blue, 0 }  ,draw opacity=1 ][fill={rgb, 255:red, 177; green, 177; blue, 177 }  ,fill opacity=0.7 ] (193.45,78.62) .. controls (193.45,75.93) and (195.77,73.76) .. (198.64,73.76) .. controls (201.51,73.76) and (203.84,75.93) .. (203.84,78.62) .. controls (203.84,81.3) and (201.51,83.47) .. (198.64,83.47) .. controls (195.77,83.47) and (193.45,81.3) .. (193.45,78.62) -- cycle ;
	\draw  [fill={rgb, 255:red, 255; green, 255; blue, 255 }  ,fill opacity=1 ] (191.11,42.7) .. controls (191.11,40.02) and (193.44,37.84) .. (196.31,37.84) .. controls (199.18,37.84) and (201.51,40.02) .. (201.51,42.7) .. controls (201.51,45.38) and (199.18,47.55) .. (196.31,47.55) .. controls (193.44,47.55) and (191.11,45.38) .. (191.11,42.7) -- cycle ;
	\draw  [fill={rgb, 255:red, 255; green, 255; blue, 255 }  ,fill opacity=1 ] (228.2,60.11) .. controls (228.2,57.43) and (230.53,55.26) .. (233.4,55.26) .. controls (236.27,55.26) and (238.6,57.43) .. (238.6,60.11) .. controls (238.6,62.79) and (236.27,64.97) .. (233.4,64.97) .. controls (230.53,64.97) and (228.2,62.79) .. (228.2,60.11) -- cycle ;
	\draw [fill={rgb, 255:red, 255; green, 255; blue, 255 }  ,fill opacity=1 ] [dash pattern={on 0.84pt off 2.51pt}]  (45.71,40.7) -- (82.8,58.11) ;
	\draw [fill={rgb, 255:red, 255; green, 255; blue, 255 }  ,fill opacity=1 ] [dash pattern={on 0.84pt off 2.51pt}]  (82.8,58.11) -- (103.88,82.06) ;
	\draw [fill={rgb, 255:red, 255; green, 255; blue, 255 }  ,fill opacity=1 ] [dash pattern={on 0.84pt off 2.51pt}]  (82.8,58.11) -- (126.26,74.98) ;
	\draw [color={rgb, 255:red, 100; green, 122; blue, 144 }  ,draw opacity=0.55 ][fill={rgb, 255:red, 255; green, 255; blue, 255 }  ,fill opacity=1 ][line width=2.25]    (82.8,58.11) -- (127.29,55.39) ;
	\draw [color={rgb, 255:red, 100; green, 122; blue, 144 }  ,draw opacity=0.55 ][fill={rgb, 255:red, 255; green, 255; blue, 255 }  ,fill opacity=1 ][line width=2.25]    (82.8,58.11) -- (123.02,40.46) ;
	\draw [fill={rgb, 255:red, 255; green, 255; blue, 255 }  ,fill opacity=1 ] [dash pattern={on 0.84pt off 2.51pt}]  (104.46,33.08) -- (82.8,58.11) ;
	\draw [fill={rgb, 255:red, 255; green, 255; blue, 255 }  ,fill opacity=1 ] [dash pattern={on 0.84pt off 2.51pt}]  (48.04,76.62) -- (82.8,58.11) ;
	\draw [fill={rgb, 255:red, 255; green, 255; blue, 255 }  ,fill opacity=1 ] [dash pattern={on 0.84pt off 2.51pt}]  (35.46,59.24) -- (82.8,58.11) ;
	\draw  [fill={rgb, 255:red, 255; green, 255; blue, 255 }  ,fill opacity=1 ] (99.27,33.08) .. controls (99.27,30.4) and (101.59,28.22) .. (104.46,28.22) .. controls (107.33,28.22) and (109.66,30.4) .. (109.66,33.08) .. controls (109.66,35.76) and (107.33,37.93) .. (104.46,37.93) .. controls (101.59,37.93) and (99.27,35.76) .. (99.27,33.08) -- cycle ;
	\draw  [color={rgb, 255:red, 0; green, 0; blue, 0 }  ,draw opacity=1 ][fill={rgb, 255:red, 100; green, 122; blue, 144 }  ,fill opacity=0.55 ] (122.23,39.07) .. controls (122.23,36.38) and (124.56,34.21) .. (127.43,34.21) .. controls (130.3,34.21) and (132.63,36.38) .. (132.63,39.07) .. controls (132.63,41.75) and (130.3,43.92) .. (127.43,43.92) .. controls (124.56,43.92) and (122.23,41.75) .. (122.23,39.07) -- cycle ;
	\draw  [fill={rgb, 255:red, 100; green, 122; blue, 144 }  ,fill opacity=0.55 ] (127.29,55.39) .. controls (127.29,52.71) and (129.62,50.54) .. (132.49,50.54) .. controls (135.36,50.54) and (137.69,52.71) .. (137.69,55.39) .. controls (137.69,58.07) and (135.36,60.25) .. (132.49,60.25) .. controls (129.62,60.25) and (127.29,58.07) .. (127.29,55.39) -- cycle ;
	\draw  [fill={rgb, 255:red, 255; green, 255; blue, 255 }  ,fill opacity=1 ] (121.07,74.98) .. controls (121.07,72.3) and (123.39,70.13) .. (126.26,70.13) .. controls (129.14,70.13) and (131.46,72.3) .. (131.46,74.98) .. controls (131.46,77.66) and (129.14,79.84) .. (126.26,79.84) .. controls (123.39,79.84) and (121.07,77.66) .. (121.07,74.98) -- cycle ;
	\draw  [fill={rgb, 255:red, 255; green, 255; blue, 255 }  ,fill opacity=1 ] (98.68,82.06) .. controls (98.68,79.38) and (101.01,77.2) .. (103.88,77.2) .. controls (106.75,77.2) and (109.08,79.38) .. (109.08,82.06) .. controls (109.08,84.74) and (106.75,86.91) .. (103.88,86.91) .. controls (101.01,86.91) and (98.68,84.74) .. (98.68,82.06) -- cycle ;
	\draw  [fill={rgb, 255:red, 255; green, 255; blue, 255 }  ,fill opacity=1 ] (30.26,59.24) .. controls (30.26,56.56) and (32.59,54.39) .. (35.46,54.39) .. controls (38.33,54.39) and (40.66,56.56) .. (40.66,59.24) .. controls (40.66,61.93) and (38.33,64.1) .. (35.46,64.1) .. controls (32.59,64.1) and (30.26,61.93) .. (30.26,59.24) -- cycle ;
	\draw  [fill={rgb, 255:red, 255; green, 255; blue, 255 }  ,fill opacity=1 ] (42.85,76.62) .. controls (42.85,73.93) and (45.17,71.76) .. (48.04,71.76) .. controls (50.91,71.76) and (53.24,73.93) .. (53.24,76.62) .. controls (53.24,79.3) and (50.91,81.47) .. (48.04,81.47) .. controls (45.17,81.47) and (42.85,79.3) .. (42.85,76.62) -- cycle ;
	\draw  [fill={rgb, 255:red, 255; green, 255; blue, 255 }  ,fill opacity=1 ] (40.51,40.7) .. controls (40.51,38.02) and (42.84,35.84) .. (45.71,35.84) .. controls (48.58,35.84) and (50.91,38.02) .. (50.91,40.7) .. controls (50.91,43.38) and (48.58,45.55) .. (45.71,45.55) .. controls (42.84,45.55) and (40.51,43.38) .. (40.51,40.7) -- cycle ;
	\draw  [fill={rgb, 255:red, 255; green, 255; blue, 255 }  ,fill opacity=1 ] (77.6,58.11) .. controls (77.6,55.43) and (79.93,53.26) .. (82.8,53.26) .. controls (85.67,53.26) and (88,55.43) .. (88,58.11) .. controls (88,60.79) and (85.67,62.97) .. (82.8,62.97) .. controls (79.93,62.97) and (77.6,60.79) .. (77.6,58.11) -- cycle ;
	
	\draw (258.17,20.42) node  [font=\large] [align=left] {1};
	\draw (289.3,29.67) node  [font=\large] [align=left] {2};
	\draw (297.85,58.52) node  [font=\large] [align=left] {3};
	\draw (289.3,84.09) node  [font=\large] [align=left] {4};
	\draw (258.76,99.24) node  [font=\large] [align=left] {5};
	\draw (196.67,27.95) node  [font=\large] [align=left] {8};
	\draw (173.3,63.41) node  [font=\large] [align=left] {7};
	\draw (204.08,93.35) node  [font=\large] [align=left] {6};
	\draw (107.57,18.42) node  [font=\large] [align=left] {1};
	\draw (138.7,27.67) node  [font=\large] [align=left] {2};
	\draw (147.25,56.52) node  [font=\large] [align=left] {3};
	\draw (138.7,82.09) node  [font=\large] [align=left] {4};
	\draw (108.16,96.24) node  [font=\large] [align=left] {5};
	\draw (47.07,24.95) node  [font=\large] [align=left] {8};
	\draw (22.7,61.41) node  [font=\large] [align=left] {7};
	\draw (53.48,91.35) node  [font=\large] [align=left] {6};
	\draw (151.28,119.35) node   [align=left] {};
	\draw (84.08,114.55) node  [font=\large] [align=left] {\ALG};
	\draw (234.28,114.36) node  [font=\large] [align=left] {\OPT};

\end{tikzpicture}

}

\newcommand{\figStarOneB}{

	\tikzset{every picture/.style={line width=0.75pt}} 
	
	\begin{tikzpicture}[x=0.75pt,y=0.75pt,yscale=-1,xscale=1]
		
		\draw [color={rgb, 255:red, 188; green, 155; blue, 199 }  ,draw opacity=0.55 ][fill={rgb, 255:red, 255; green, 255; blue, 255 }  ,fill opacity=1 ][line width=2.25]    (610.4,58.69) -- (628.63,78.84) ;
		\draw [color={rgb, 255:red, 122; green, 144; blue, 199 }  ,draw opacity=0.55 ][fill={rgb, 255:red, 255; green, 255; blue, 255 }  ,fill opacity=1 ][line width=2.25]    (610.4,58.69) -- (649.33,73.36) ;
		\draw [color={rgb, 255:red, 122; green, 144; blue, 199 }  ,draw opacity=0.66 ][fill={rgb, 255:red, 255; green, 255; blue, 255 }  ,fill opacity=1 ][line width=2.25]    (610.4,58.69) -- (643.55,56.79) -- (654.89,55.97) ;
		\draw [color={rgb, 255:red, 199; green, 177; blue, 144 }  ,draw opacity=1 ][fill={rgb, 255:red, 255; green, 255; blue, 255 }  ,fill opacity=1 ][line width=2.25]    (610.4,58.69) -- (650.43,41.81) ;
		\draw [color={rgb, 255:red, 199; green, 177; blue, 144 }  ,draw opacity=1 ][fill={rgb, 255:red, 255; green, 255; blue, 255 }  ,fill opacity=1 ][line width=2.25]    (628.03,37.41) -- (610.4,58.69) ;
		\draw [color={rgb, 255:red, 188; green, 144; blue, 199 }  ,draw opacity=0.55 ][fill={rgb, 255:red, 255; green, 255; blue, 255 }  ,fill opacity=1 ][line width=2.25]    (579.76,74.61) -- (612.4,57.69) ;
		\draw [color={rgb, 255:red, 0; green, 0; blue, 0 }  ,draw opacity=0.44 ][fill={rgb, 255:red, 255; green, 255; blue, 255 }  ,fill opacity=1 ][line width=2.25]    (568.26,59.82) -- (593.63,59.31) -- (610.4,58.69) ;
		\draw  [fill={rgb, 255:red, 199; green, 177; blue, 144 }  ,fill opacity=0.5 ] (626.87,33.66) .. controls (626.87,30.98) and (629.19,28.8) .. (632.06,28.8) .. controls (634.93,28.8) and (637.26,30.98) .. (637.26,33.66) .. controls (637.26,36.34) and (634.93,38.51) .. (632.06,38.51) .. controls (629.19,38.51) and (626.87,36.34) .. (626.87,33.66) -- cycle ;
		\draw  [fill={rgb, 255:red, 199; green, 177; blue, 144 }  ,fill opacity=0.5 ] (649.83,39.64) .. controls (649.83,36.96) and (652.16,34.79) .. (655.03,34.79) .. controls (657.9,34.79) and (660.23,36.96) .. (660.23,39.64) .. controls (660.23,42.32) and (657.9,44.5) .. (655.03,44.5) .. controls (652.16,44.5) and (649.83,42.32) .. (649.83,39.64) -- cycle ;
		\draw  [fill={rgb, 255:red, 111; green, 144; blue, 199 }  ,fill opacity=0.55 ] (654.89,55.97) .. controls (654.89,53.29) and (657.22,51.11) .. (660.09,51.11) .. controls (662.96,51.11) and (665.29,53.29) .. (665.29,55.97) .. controls (665.29,58.65) and (662.96,60.82) .. (660.09,60.82) .. controls (657.22,60.82) and (654.89,58.65) .. (654.89,55.97) -- cycle ;
		\draw  [fill={rgb, 255:red, 111; green, 144; blue, 199 }  ,fill opacity=0.44 ] (648.67,75.56) .. controls (648.67,72.88) and (650.99,70.7) .. (653.86,70.7) .. controls (656.74,70.7) and (659.06,72.88) .. (659.06,75.56) .. controls (659.06,78.24) and (656.74,80.42) .. (653.86,80.42) .. controls (650.99,80.42) and (648.67,78.24) .. (648.67,75.56) -- cycle ;
		\draw  [fill={rgb, 255:red, 188; green, 155; blue, 199 }  ,fill opacity=0.5 ] (626.28,82.63) .. controls (626.28,79.95) and (628.61,77.78) .. (631.48,77.78) .. controls (634.35,77.78) and (636.68,79.95) .. (636.68,82.63) .. controls (636.68,85.32) and (634.35,87.49) .. (631.48,87.49) .. controls (628.61,87.49) and (626.28,85.32) .. (626.28,82.63) -- cycle ;
		\draw  [fill={rgb, 255:red, 177; green, 177; blue, 177 }  ,fill opacity=0.7 ] (557.86,59.82) .. controls (557.86,57.14) and (560.19,54.97) .. (563.06,54.97) .. controls (565.93,54.97) and (568.26,57.14) .. (568.26,59.82) .. controls (568.26,62.5) and (565.93,64.68) .. (563.06,64.68) .. controls (560.19,64.68) and (557.86,62.5) .. (557.86,59.82) -- cycle ;
		\draw  [color={rgb, 255:red, 0; green, 0; blue, 0 }  ,draw opacity=1 ][fill={rgb, 255:red, 188; green, 155; blue, 199 }  ,fill opacity=0.55 ] (570.45,77.19) .. controls (570.45,74.51) and (572.77,72.34) .. (575.64,72.34) .. controls (578.51,72.34) and (580.84,74.51) .. (580.84,77.19) .. controls (580.84,79.87) and (578.51,82.05) .. (575.64,82.05) .. controls (572.77,82.05) and (570.45,79.87) .. (570.45,77.19) -- cycle ;
		\draw  [fill={rgb, 255:red, 255; green, 255; blue, 255 }  ,fill opacity=1 ] (605.2,58.69) .. controls (605.2,56.01) and (607.53,53.83) .. (610.4,53.83) .. controls (613.27,53.83) and (615.6,56.01) .. (615.6,58.69) .. controls (615.6,61.37) and (613.27,63.55) .. (610.4,63.55) .. controls (607.53,63.55) and (605.2,61.37) .. (605.2,58.69) -- cycle ;
		\draw [fill={rgb, 255:red, 255; green, 255; blue, 255 }  ,fill opacity=1 ] [dash pattern={on 0.84pt off 2.51pt}]  (424.71,41.28) -- (461.8,58.69) ;
		\draw [fill={rgb, 255:red, 255; green, 255; blue, 255 }  ,fill opacity=1 ] [dash pattern={on 0.84pt off 2.51pt}]  (461.8,58.69) -- (482.88,82.63) ;
		\draw [fill={rgb, 255:red, 255; green, 255; blue, 255 }  ,fill opacity=1 ] [dash pattern={on 0.84pt off 2.51pt}]  (461.8,58.69) -- (505.26,75.56) ;
		\draw [color={rgb, 255:red, 100; green, 122; blue, 144 }  ,draw opacity=0.55 ][fill={rgb, 255:red, 255; green, 255; blue, 255 }  ,fill opacity=1 ][line width=2.25]    (461.8,58.69) -- (506.29,55.97) ;
		\draw [color={rgb, 255:red, 100; green, 122; blue, 144 }  ,draw opacity=0.55 ][fill={rgb, 255:red, 255; green, 255; blue, 255 }  ,fill opacity=1 ][line width=2.25]    (461.8,58.69) -- (502.02,41.03) ;
		\draw [fill={rgb, 255:red, 255; green, 255; blue, 255 }  ,fill opacity=1 ] [dash pattern={on 0.84pt off 2.51pt}]  (483.46,33.66) -- (461.8,58.69) ;
		\draw [color={rgb, 255:red, 65; green, 117; blue, 5 }  ,draw opacity=0.44 ][fill={rgb, 255:red, 255; green, 255; blue, 255 }  ,fill opacity=1 ][line width=2.25]    (431,74.78) -- (461.8,58.69) ;
		\draw [color={rgb, 255:red, 65; green, 117; blue, 5 }  ,draw opacity=0.44 ][fill={rgb, 255:red, 177; green, 77; blue, 55 }  ,fill opacity=0.55 ][line width=2.25]    (419.66,59.82) -- (461.8,58.69) ;
		\draw  [fill={rgb, 255:red, 255; green, 255; blue, 255 }  ,fill opacity=1 ] (478.27,33.66) .. controls (478.27,30.98) and (480.59,28.8) .. (483.46,28.8) .. controls (486.33,28.8) and (488.66,30.98) .. (488.66,33.66) .. controls (488.66,36.34) and (486.33,38.51) .. (483.46,38.51) .. controls (480.59,38.51) and (478.27,36.34) .. (478.27,33.66) -- cycle ;
		\draw  [color={rgb, 255:red, 0; green, 0; blue, 0 }  ,draw opacity=1 ][fill={rgb, 255:red, 100; green, 122; blue, 144 }  ,fill opacity=0.55 ] (501.23,39.64) .. controls (501.23,36.96) and (503.56,34.79) .. (506.43,34.79) .. controls (509.3,34.79) and (511.63,36.96) .. (511.63,39.64) .. controls (511.63,42.32) and (509.3,44.5) .. (506.43,44.5) .. controls (503.56,44.5) and (501.23,42.32) .. (501.23,39.64) -- cycle ;
		\draw  [fill={rgb, 255:red, 100; green, 122; blue, 144 }  ,fill opacity=0.55 ] (506.29,55.97) .. controls (506.29,53.29) and (508.62,51.11) .. (511.49,51.11) .. controls (514.36,51.11) and (516.69,53.29) .. (516.69,55.97) .. controls (516.69,58.65) and (514.36,60.82) .. (511.49,60.82) .. controls (508.62,60.82) and (506.29,58.65) .. (506.29,55.97) -- cycle ;
		\draw  [fill={rgb, 255:red, 255; green, 255; blue, 255 }  ,fill opacity=1 ] (500.07,75.56) .. controls (500.07,72.88) and (502.39,70.7) .. (505.26,70.7) .. controls (508.14,70.7) and (510.46,72.88) .. (510.46,75.56) .. controls (510.46,78.24) and (508.14,80.42) .. (505.26,80.42) .. controls (502.39,80.42) and (500.07,78.24) .. (500.07,75.56) -- cycle ;
		\draw  [fill={rgb, 255:red, 255; green, 255; blue, 255 }  ,fill opacity=1 ] (477.68,82.63) .. controls (477.68,79.95) and (480.01,77.78) .. (482.88,77.78) .. controls (485.75,77.78) and (488.08,79.95) .. (488.08,82.63) .. controls (488.08,85.32) and (485.75,87.49) .. (482.88,87.49) .. controls (480.01,87.49) and (477.68,85.32) .. (477.68,82.63) -- cycle ;
		\draw  [fill={rgb, 255:red, 65; green, 117; blue, 5 }  ,fill opacity=0.44 ] (409.26,59.82) .. controls (409.26,57.14) and (411.59,54.97) .. (414.46,54.97) .. controls (417.33,54.97) and (419.66,57.14) .. (419.66,59.82) .. controls (419.66,62.5) and (417.33,64.68) .. (414.46,64.68) .. controls (411.59,64.68) and (409.26,62.5) .. (409.26,59.82) -- cycle ;
		\draw  [fill={rgb, 255:red, 65; green, 117; blue, 5 }  ,fill opacity=0.44 ] (421.85,77.19) .. controls (421.85,74.51) and (424.17,72.34) .. (427.04,72.34) .. controls (429.91,72.34) and (432.24,74.51) .. (432.24,77.19) .. controls (432.24,79.87) and (429.91,82.05) .. (427.04,82.05) .. controls (424.17,82.05) and (421.85,79.87) .. (421.85,77.19) -- cycle ;
		\draw  [fill={rgb, 255:red, 255; green, 255; blue, 255 }  ,fill opacity=1 ] (419.51,41.28) .. controls (419.51,38.59) and (421.84,36.42) .. (424.71,36.42) .. controls (427.58,36.42) and (429.91,38.59) .. (429.91,41.28) .. controls (429.91,43.96) and (427.58,46.13) .. (424.71,46.13) .. controls (421.84,46.13) and (419.51,43.96) .. (419.51,41.28) -- cycle ;
		\draw  [fill={rgb, 255:red, 255; green, 255; blue, 255 }  ,fill opacity=1 ] (456.6,58.69) .. controls (456.6,56.01) and (458.93,53.83) .. (461.8,53.83) .. controls (464.67,53.83) and (467,56.01) .. (467,58.69) .. controls (467,61.37) and (464.67,63.55) .. (461.8,63.55) .. controls (458.93,63.55) and (456.6,61.37) .. (456.6,58.69) -- cycle ;
		\draw [color={rgb, 255:red, 0; green, 0; blue, 0 }  ,draw opacity=0.44 ][fill={rgb, 255:red, 255; green, 255; blue, 255 }  ,fill opacity=1 ][line width=2.25]    (577.83,44.04) -- (606,55.49) ;
		\draw  [fill={rgb, 255:red, 177; green, 177; blue, 177 }  ,fill opacity=0.7 ] (568.86,41.82) .. controls (568.86,39.14) and (571.19,36.97) .. (574.06,36.97) .. controls (576.93,36.97) and (579.26,39.14) .. (579.26,41.82) .. controls (579.26,44.5) and (576.93,46.68) .. (574.06,46.68) .. controls (571.19,46.68) and (568.86,44.5) .. (568.86,41.82) -- cycle ;
		
		\draw (635.17,19) node  [font=\large] [align=left] {1};
		\draw (666.3,28.25) node  [font=\large] [align=left] {2};
		\draw (674.85,57.09) node  [font=\large] [align=left] {3};
		\draw (666.3,82.67) node  [font=\large] [align=left] {4};
		\draw (635.76,97.82) node  [font=\large] [align=left] {5};
		\draw (573.67,26.53) node  [font=\large] [align=left] {8};
		\draw (550.3,61.99) node  [font=\large] [align=left] {7};
		\draw (581.08,91.92) node  [font=\large] [align=left] {6};
		\draw (486.57,19) node  [font=\large] [align=left] {1};
		\draw (517.7,28.25) node  [font=\large] [align=left] {2};
		\draw (526.25,57.09) node  [font=\large] [align=left] {3};
		\draw (517.7,82.67) node  [font=\large] [align=left] {4};
		\draw (487.16,96.82) node  [font=\large] [align=left] {5};
		\draw (427.07,25.53) node  [font=\large] [align=left] {8};
		\draw (401.7,61.99) node  [font=\large] [align=left] {7};
		\draw (432.48,91.92) node  [font=\large] [align=left] {6};
		\draw (530.28,119.92) node   [align=left] {};
		\draw (463.08,115.12) node  [font=\large] [align=left] {\ALG};
		\draw (611.28,112.94) node  [font=\large] [align=left] {\OPT};

	\end{tikzpicture}
		
}